\documentclass[
  aps,
  pre,
  longbibliography, 
  onecolumn,  
  nofootinbib,showpacs,superscriptaddress,preprintnumbers]{revtex4-2}

\usepackage[utf8]{inputenc}
\usepackage{graphicx, color, hyperref}
\usepackage{dcolumn}
\usepackage{bm}
\usepackage{amsmath}
\usepackage{amsthm}
\usepackage{xparse} 
\usepackage{mathtools} 
\usepackage[caption=false]{subfig}
\usepackage{blkarray} 
\usepackage{amscd,verbatim}
\usepackage{amssymb}
\usepackage[all]{xy}
\usepackage{tikz}
\usepackage{pgf,pgfplots}
\usepackage{enumitem}

\newcommand\coker{\mathop{{\rm coker}}}

\newcommand\p{\partial}
\newcommand\rd{{\rm d}}

\usepackage[normalem]{ulem}  

\renewcommand\sout{\bgroup \color{red} \ULdepth=-.5ex \ULset}

\definecolor{mydarkred}{RGB}{233,20,35}
\definecolor{mypurple}{RGB}{120, 35, 160}
\definecolor{mydarkpurple}{RGB}{128, 100, 162}
\definecolor{mybrown}{RGB}{255, 195, 0}
\definecolor{myaqua}{RGB}{29, 153, 168}
\definecolor{myblue}{RGB}{91, 129, 184}  
\definecolor{mygreen}{RGB}{155, 187, 89}  

\definecolor{mybrightblue}{RGB}{0, 140, 255}  

\tikzstyle{species}=[
  circle,draw=black!100, thick, inner sep=0pt,minimum size=6mm
  ] 
\tikzstyle{reaction}=[rectangle,draw=black!100,fill=black!15,thick, inner sep=0pt,minimum size=6mm]

\theoremstyle{remark}
\newtheorem{remark}{Remark}
\newtheorem{theorem}{Theorem}

\newtheorem{corollary}{Corollary}
\newtheorem{definition}{Definition}
\newtheorem{example}{Example}

\begin{document}
\preprint{KEK-TH-2296}

\title{
Structural reduction of chemical reaction networks \\
based on topology 
}

\author{Yuji~Hirono}
\email{yuji.hirono@gmail.com}
\affiliation{Asia Pacific Center for Theoretical Physics, Pohang 37673, Korea}
\affiliation{Department of Physics, POSTECH, Pohang 37673, Korea}

\author{Takashi~Okada}
\email{takashi.okada@riken.jp}
\affiliation{RIKEN iTHEMS, RIKEN, Wako 351-0198, Japan}

\author{Hiroyasu~Miyazaki}
\email{hiroyasu.miyazaki@riken.jp}
\affiliation{RIKEN iTHEMS, RIKEN, Wako 351-0198, Japan}

\author{Yoshimasa~Hidaka}
\email{hidaka@post.kek.jp}
\affiliation{KEK Theory Center, Tsukuba 305-0801, Japan}
\affiliation{Graduate University for Advanced Studies (Sokendai), Tsukuba 305-0801, Japan}
\affiliation{RIKEN iTHEMS, RIKEN, Wako 351-0198, Japan}

\date{\today}

\begin{abstract}
We develop a model-independent reduction method of  chemical reaction systems based on the stoichiometry, 
which determines their network topology.
A subnetwork can be eliminated systematically 
to give a reduced system with fewer degrees of freedom. 
This subnetwork removal is accompanied by rewiring 
of the network, which is prescribed by the Schur complement of the stoichiometric matrix. 
Using homology and cohomology groups 
to characterize the topology of chemical reaction networks, 
we can track the changes of the network topology 
induced by the reduction through the changes in those groups. 
We prove that, when certain topological conditions are met, the steady-state chemical concentrations and reaction rates of the reduced system are ensured to be the same as those of the original system. 
This result holds regardless of the modeling of the reactions, 
namely chemical kinetics, since the conditions only involve topological information. 
This is advantageous because the details of reaction kinetics and parameter values are difficult to identify in many practical situations. 
The method allows us to reduce a reaction network while preserving its original steady-state properties, thereby complex reaction systems can be studied efficiently. 
We demonstrate the reduction method in hypothetical networks and the central carbon metabolism of {\it Escherichia coli}.
\end{abstract}

\maketitle

\tableofcontents

\section{Introduction}
Chemical reactions in living systems form complex networks~\cite{kanehisa2000kegg,
jeong2000large,ravasz2002hierarchical}. 
They operate in a highly coordinated manner, 
and are responsible for various cellular functions. 
Experimentally, high-throughput measurements have been conducted 
to study cellular responses to perturbations 
for the purpose of elucidating underlying regulatory mechanisms (see, for example, Refs.~\cite{ishii2007multiple, munger2008systems,zamboni200913}).
One approach to the theoretical studies of biological systems is to build elaborated models, employing particular kinetics, parameter values, and initial/external conditions.
Although these models can provide detailed quantitative predictions, 
a faithful modeling is challenging for most biological systems, because our prior knowledge about kinetics and parameter values is limited, and also because many parameters are difficult to measure experimentally. 
Furthermore, the complexity of models may confound model-independent features with model-dependent ones.

To address these difficulties, it is desirable to reduce complex reaction systems to simpler ones.
A reduction is practically useful since it can reduce the number of variables and parameters needed to be included in the analysis, and it can also identify features essential to focal phenomena or properties of interest (such as biomass production of a metabolic network).
It also relates to a conceptual question of the robustness of biochemical processes \cite{barkai1997robustness,alon1999robustness,kitano2004biological,kitano2007towards,shinar2010structural,LARHLIMI20111,10.3389/fgene.2012.00067,eloundou2016network}. 
Chemical reaction networks inside living organisms are highly interconnected, and yet are robust under internal fluctuations and environmental perturbations.
If a system is insensitive to the details of its substructure, 
it is natural to expect that a reduction is possible, 
in the spirit of renormalization. 
To the best of our knowledge, the reduction methods 
\cite{okino1998simplification,snowden2017methods}
of chemical reaction systems 
studied so far are 
based on timescale separation, 
lumping \cite{kuo1969lumping,wei1969lumping}, 
sensitivity analysis \cite{iet:/content/journals/10.1049/iet-syb.2011.0015,DEGENRING2004729,APRI201216} 
or optimizations \cite{dano2006reduction,taylor2008oscillator}. 
To apply those approaches, 
we need detailed information about the reactions. 
For example, to exploit the timescale separation, 
we should know which reactions are fast and which are slow. 
The sensitivity analysis also requires 
the dependence of the system on various parameters.

In this paper, we develop a systematic method of reducing chemical reaction networks based on their topology (see Figs.~\ref{fig:schem-reduction} and \ref{fig:schematic}). 
One motivation for the reduction method comes from 
the {\it law of localization}  \cite{PhysRevLett.117.048101,PhysRevE.96.022322, PhysRevE.98.012417};
 if a certain topological index, which we call the {\it  influence index}, is zero for a subnetwork, perturbations inside the subnetwork do not affect the steady state of the remaining elements of the network (see Sec.~\ref{sec:lol} for the precise statement). 
This observation indicates that certain subnetworks are `irrelevant', as far as the remaining part of the network is concerned.  
As we will show, a reduction can be systematically performed through the Schur complementation 
of the stoichiometric matrix with respect to a subset of 
chemical species and reactions. 
The well-definedness of the reduction process 
requires that the subnetwork should satisfy a condition called the {\it output-completeness}. 
The behavior of the reduced system depends on the topological nature of the subnetwork. 
As a central result, 
we prove that, when the influence index of the subnetwork vanishes,
the steady-state chemical concentrations and reaction rates of the reduced system are exactly the same as those of the original system, as far as the remaining degrees of freedom are concerned.

We emphasize that those conditions are {\it topological} ones
and determined solely by the network structure;  
hence, are insensitive to the details of how the reactions are modeled. 
Thus, the result is broadly applicable, 
because it holds regardless of the kinetics or parameter values. 
This is of practical merit since the kinetics of reactions or the values of parameters are difficult to identify in many situations. 
To characterize the topology of reaction networks, 
we introduce the homology and cohomology groups for chemical reaction networks. 
The change of the topology of chemical reaction networks 
under the reduction is captured by the change of the (co)homology groups.  
The tools of algebraic topology are convenient for tracking those changes. 
We recommend the readers who are interested in practical aspects of reduction to directly go to Sec.~\ref{sec:reduction}, where we discuss the reduction procedure with simple examples.

\begin{figure}[tb] 
\centering
\includegraphics[clip, trim=0cm 6cm 0cm 0cm, width=12cm] 
{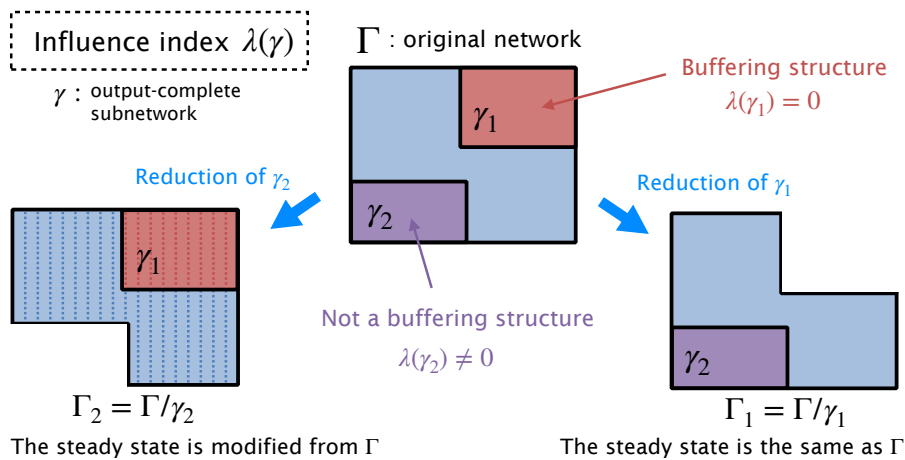}
\caption{
Schematic of the reduction procedure. 
For a given subnetwork $\gamma$
satisfying a condition called output-completeness, 
we assign a nonnegative integer that we call the influence index, $\lambda(\gamma)$. 
A subnetwork with vanishing influence index is called a buffering structure. 
Although an elimination of a subnetwork ($\gamma_2$ in the figure) generally modifies the steady state of the remaining part of the network, 
a buffering structure ($\gamma_1$) can be reduced while preserving 
the original steady state of the remaining part.
}
 \label{fig:schem-reduction} 
\end{figure}

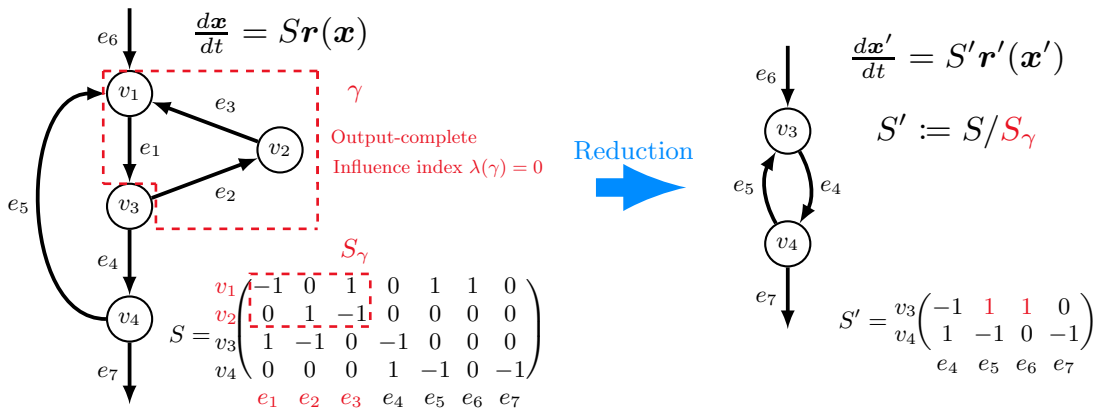
\begin{figure}[tb] 
\centering
\begin{tikzpicture} 
\node at (2, 0.8) {  \scalebox{1.4} { $\frac{d \bm x}{dt}=S \bm r (\bm x)$
} };  
\node at (11, 0.5) {  \scalebox{1.4} { $\frac{d \bm  x'}{dt}=S'
 \bm r' (\bm x')$
} };  
\node at (11, -0.5) {\scalebox{1.4} {$S'\coloneqq S / \color{mydarkred}{S_\gamma}$}};

    \node[species] (v1) at (0,0) {$v_1$}; 
    \node[species] (v3) at (0,-1.5) {$v_3$};
    \node[species] (v4) at (0,-3) {$v_4$};
    \node[species] (v2) at (2,-0.75) {$v_2$}; 
    
    \node (d1) at (0,1.25) {}; 
    \node (d2) at (0,-4.25) {}; 
  
    \draw [-latex,line width=0.5mm] (d1) edge node[left]{$e_6$} (v1);
    \draw [-latex,line width=0.5mm] (v1) edge node[right]{$e_1$} (v3);
    \draw [-latex,line width=0.5mm] (v3) edge node[below right]{$e_2$} (v2);
    \draw [-latex,line width=0.5mm] (v2) edge node[above right]{$e_3$} (v1);
    \draw [-latex,line width=0.5mm] (v3) edge node[left]{$e_4$} (v4);
    \draw [-latex,line width=0.5mm] (v4) edge node[left]{$e_7$} (d2); 
    \draw [-latex,line width=0.5mm,out=180,in=180] 
    (v4) edge node[left]{$e_5$} (v1);

 \node at (3.,0) { \scalebox{1.1} {\color{mydarkred} $\gamma$}};
 \node at (3.6,-0.6) { \scalebox{0.8} { \color{mydarkred}Output-complete}};
 \node at (4.1,-1.) { \scalebox{0.8} {\color{mydarkred}Influence index $\lambda (\gamma) =0$
 }};
  
  \draw [mydarkred,dashed,line width=1]
  (-0.35,0.3) -- (2.5,0.3); 
   \draw [mydarkred,dashed,line width=1]
  (-0.35,-1.2) -- (0.35,-1.2); 
     \draw [mydarkred,dashed,line width=1]
  (0.35,-1.2) -- (0.35,-1.8); 
     \draw [mydarkred,dashed,line width=1]
  (0.35,-1.8) -- (2.5,-1.8); 
    \draw [mydarkred,dashed,line width=1]
  (-0.35,0.3) -- (-0.35,-1.2); 
     \draw [mydarkred,dashed,line width=1]
  (2.5,0.3) -- (2.5,-1.2); 
     \draw [mydarkred,dashed,line width=1]
  (2.5,0.3) -- (2.5,-1.8); 

 \node at (0.8, -3.2) {
 $S=$};
 \node at (3.2, -3.5) {
    {\small $
 \begin{blockarray}{cccccccc}
 \begin{block}{c(ccccccc)}
 {\color{mydarkred}v_1}\,\, & -1 & 0 & 1 & 0 & 1 & 1 & 0 \\
 {\color{mydarkred}v_2}\,\,& 0 & 1 & -1 & 0 & 0 & 0 & 0 \\
 {v_3}\,\,& 1 & -1 & 0 & -1 & 0 & 0 & 0 \\
 {v_4}\,\,& 0 & 0 & 0 & 1 & -1 & 0 & -1 \\
 \end{block} 
 & {\color{mydarkred}e_1} & {\color{mydarkred}e_2}&  {\color{mydarkred}e_3} & {e_4}& {e_5}& e_6&  e_7
  \end{blockarray} 
          $}
      }; 
  \draw[mydarkred, dashed,line width=1] (1.6,-2.4) rectangle (3.2,-3.1);
   \node at (3, -2.1) {  \scalebox{1.1} {\color{mydarkred} $S_\gamma$} };

 \draw[mybrightblue,-latex, line width=7pt] (6.2,-1.25) -- (7.4,-1.25); 
 
 
 \node at (6.7,-0.75){ \scalebox{1.1} {\color{mybrightblue}Reduction}};

    \node[species] (red_v2) at (8.75,-0.5) {$v_3$}; 
    \node[species] (red_v3) at (8.75,-2) {$v_4$};

    \node (red_d1) at (8.75,0.75) {};
    \node (red_d2) at (8.75,-3.25) {};
    
    \draw [-latex,line width=0.5mm] 
    (red_d1) edge node[left]{$e_6$} (red_v2);
    \draw [-latex,line width=0.5mm,out=-60,in=60]
    (red_v2) edge node[right]{$e_4$} (red_v3);
   \draw [-latex,line width=0.5mm,out=120,in=240] 
    (red_v3) edge node[left]{$e_5$} (red_v2); 
    \draw [-latex,line width=0.5mm]
    (red_v3) edge node[left]{$e_7$} (red_d2);

\node at (9.75, -3.){$S'=$};
 \node at (11.4, -3.4) {
   {\small  $
 \begin{blockarray}{cccccccc}
 \begin{block}{c(ccccccc)}
 {\color{black}v_3}\,\,&  -1 & \color{mydarkred}{1} & \color{mydarkred}{1} & 0 \\
 {v_4}\,\,&  1 & -1 & 0 & -1 \\
 \end{block} 
 & {\color{black}e_4}& {\color{black}e_5}& e_6&  e_7
  \end{blockarray} 
          $}
          };
\end{tikzpicture}
\caption{
Example of the reduction. 
In a chemical reaction network with a stoichiometric matrix $S$, if a subnetwork $\gamma$ with $S_\gamma$ is output-complete, the system can be reduced to a smaller system, whose stoichiometric matrix $S'$ is given by the generalized Schur complement $S'\coloneqq S/S_\gamma$. 
The reduced system can reproduce the same steady-state properties of the original system, if the influence index of the subnetwork is zero. 
Note that $S'$ is in general different from the corresponding submatrix of $S$ (compare the lower-right block of $S$ with $S'$, where the difference between them is indicated by the colored components in $S'$). This alteration is pictorially represented as `rewiring' of the network (e.g., the head of $e_5$ is rewired to $v_3$). See Fig.~\ref{fig:ecoli} for an application to the central carbon metabolism of {\it E.~coli}. }
 \label{fig:schematic}
\end{figure}

The rest of the paper is organized as follows. 
In Sec.~\ref{sec:crn}, we introduce concepts for characterizing the structure of chemical reaction networks. 
Further, we introduce the homology and cohomology groups for chemical reaction networks, 
and the steady-state reaction rates and concentrations 
are determined by the elements of the cohomology groups. 
In Sec.~\ref{sec:lol}, we review the structural sensitivity analysis and the law of localization. 
We also show that the influence index is submodular 
as a function over output-complete subnetworks. 
In Sec.~\ref{sec:reduction}, we introduce the reduction procedure and illustrate the method with simple examples. 
In Sec.~\ref{sec:red-buff}, we discuss the relation between 
the structural sensitivity analysis and the reduction method. 
We show that the reduction of a buffering structure, 
that is an output-complete subnetwork with vanishing influence index, 
has a particularly nice property: The reduced system admits the same steady states as the original system.
In Sec.~\ref{sec:ex-ecoli}, as an application to realistic networks, we demonstrate the reduction method for the metabolic pathways of 
{\it Escherichia~coli}. 
Section \ref{sec:summary} is devoted to summary and outlook. 
In Appendix~\ref{sec:app-hodge}, we discuss the Hodge decomposition and Laplace operators for chemical reaction networks. 
In Appendix~\ref{sec:app-cycles}, we provide intuitive interpretations of the cycles and conserved charges of various types 
that appear in the decomposition of the influence index. 
We also illustrate how the decomposition of the index 
can be seen visually in the structure of the A-matrix, which characterizes the response of the steady state to the perturbations of parameters.  
In Appendix \ref{sec:app-emergent}, the role of emergent conserved charges in subnetworks is discussed. 
In Appendix \ref{sec:app-ecoli}, 
we provide the details of 
the metabolic pathways of {\it E.~coli.} discussed in Sec.~\ref{sec:ex-ecoli}.

\section{Topology of chemical reaction networks }\label{sec:crn}

In this section, we introduce definitions and concepts
for characterizing the topology of chemical reaction networks. 
Those concepts will be used to track the change of reaction networks under reductions. 

\subsection{Chemical reaction networks} 

\begin{definition}[Chemical reaction network]
  A {\it chemical reaction network} (CRN) $\Gamma$ 
  is a quadruple $\Gamma =(V,E,s,t)$,
  where 
  $V$ is a set of chemical species, 
  $E$ is a set of chemical reactions, 
  and $s$ and $t$ are source and target functions, 
  \begin{equation}
    s: E \to {\mathbb N}^{V}, 
    \quad 
    t: E \to {\mathbb N}^{V}, 
\end{equation}
which specify the reactants/products of a reaction. 
Here, $\mathbb N$ indicates nonnegative integers,
and the elements of $\mathbb N^V$ are maps from $V$ to $\mathbb N$. 
\end{definition}

Let us explain the definition in more detail. 
We will use the indices $i,j,k,\cdots$ for chemical species 
and $A,B,C,\cdots$ for chemical reactions. 
Given a reaction $e_A \in E$, 
we have a map, 
$s(e_A): V \to \mathbb N$, 
and 
$s(e_A) (v_i) \in \mathbb N$ for $v_i \in V$ 
indicates how many $v_i$ are needed as reactants 
for the reaction $e_A$. 
Similarly, $t(e_A)(v_i)\in \mathbb N$ 
is the number of $v_i$ created in reaction $e_A$. 
An element of $\mathbb N^{V}$ will be referred to as a {\it chemical complex}. 
The system can be an {\it open reaction network} \footnote{
 The compositional aspect of open reaction networks 
 has been studied in the language of category theory~\cite{fong2015decorated,Baez_2017,baez2020structured}. 
 Non-equilibrium thermodynamic analysis of open reaction networks with mass-action kinetics and with reversible reactions is performed in Refs.~\cite{PhysRevX.6.041064,polettini2014irreversible}. 
}, 
when there is a reaction whose 
source or target function is zero for any species (see the example reactions below). 
When $t(e_A)(v_i)=0$  for any $v_i \in V$, 
the product of reaction $e_A$ is deposited to the outer world. 
Similarly, 
a reaction with $s(e_A)(v_i)=0$  for any $v_i \in V$
is sourced from outside. 
A reaction is usually represented in the following form, 
\begin{equation}
  e_A :
 \sum_i y_{iA} v_i \to 
 \sum_i \bar y_{iA} v_i , 
\end{equation}
where $v_i \in V$, and 
$y_{iA}$ and $\bar y_{iA}$ are nonnegative integers. 
Those integers are given by the source and target functions as 
\begin{equation}\label{eq:yybar}
y_{iA} = s (e_A) (v_i) , 
 \quad 
\bar y_{iA} = t (e_A) (v_i). 
\end{equation}
The stoichiometry of the reaction is 
specified by the {\it stoichiometric matrix} $S$, 
whose components are given by 
\begin{equation}
  S_{iA} \coloneqq\bar y_{iA} - y_{iA}. 
  \label{eq:s-def} 
\end{equation}

\begin{remark}
There are several equivalent ways to formulate 
a chemical reaction network 
such as a hypergraph \cite{klamt2009hypergraphs} 
or a Petri net \cite{heiner2008petri,reddy1993petri}. 
\end{remark}
\begin{remark}
A reaction that involves at most one chemical species 
as reactants and products, 
such as $v_1 \to v_2$, is called {\it monomolecular}. 
When all the reactions in the system are monomolecular, 
the corresponding reaction network is a usual directed graph. 
In this case, the stoichiometric matrix is 
the incidence matrix of the graph. 
If we regard $\Gamma$ as a directed hypergraph, 
the stoichiometric matrix is 
the incidence matrix of a directed hypergraph. 
\end{remark}

We consider formal summations 
of species and reactions with real coefficients,
and 
consider vector spaces whose bases are chemical species/reactions.  
We denote the resulting vector spaces as 
\begin{eqnarray}
  C_0 (\Gamma) &\coloneqq &
  \big\{
  \sum_i a_i v_i
  \, |  \, 
  v_i \in V, \, a_i \in \mathbb R 
  \big\} , 
  \\
  C_1 (\Gamma) &\coloneqq & 
  \big\{
\sum_A b_A e_A 
  \, |  \,
  e_A \in E, \, b_A \in \mathbb R 
  \big\} .  
\end{eqnarray}
Elements of those spaces are referred to 
as 0-chains and 1-chains. 
Higher ($n \ge 2$) chains do not exist in the current setting. 
The stoichiometric matrix 
provides us with natural {\it boundary operators} on the spaces of chains, 
\begin{equation}
  \p_n : C_n (\Gamma) \to C_{n-1}(\Gamma) . 
\end{equation}
The action of $\p_1$ is 
defined by its action on the basis 
$e_A \in C_1(\Gamma)$ and $v_i \in C_0(\Gamma)$,
\begin{equation}
  \p_1 e_A 
  =\sum_i 
  (S^T)_{Ai} \, v_i , 
  \quad 
  \p_0 v_i = 0.  
\end{equation}
We often use the notation of linear algebra, 
where an element $\sum_i a_i v_i \in C_0 (\Gamma)$ is 
represented by the vector $\bm a = (a_1, a_2, \cdots)^T$, 
and we also write $\bm a \in C_0 (\Gamma)$. 
For $\bm b \in C_1 (\Gamma)$, the action of the boundary operator is given by the multiplication of the stoichiometric matrix, 
\begin{equation}
    \p_1: \bm b \mapsto S \bm b \in C_0 (\Gamma). 
\end{equation}

On the spaces of chains, 
let us define inner products by 
\begin{equation}
\langle e_A, e_B \rangle_1 = \delta_{AB}, 
\quad 
\langle v_i, v_j \rangle_0 = \delta_{ij}. 
\end{equation}
With these inner products, we can define 
the adjoint of the boundary operator, 
$\p_1^\dag : C_0 (\Gamma) \to C_1(\Gamma)$
such that $\langle \partial_1^\dag v_i,e_A\rangle_1=\langle v_i,\partial_1e_A\rangle_0$.
The action 
on the basis $v_i \in C_0(\Gamma)$ 
is given by 
\begin{equation}
  \p_1^\dag v_i = \sum_A S_{iA} \, e_A . 
\end{equation}
In the linear-algebra notation, 
the action of $\p^\dag_1$ is the multiplication of the transpose of $S$ to $\bm a \in C_0 (\Gamma)$, 
\begin{equation}
    \p^\dag_1 : \bm a \mapsto S^T \bm a \in C_1(\Gamma). 
\end{equation}

\begin{example}
Let us consider a reaction network $\Gamma = (\{v_1,v_2,v_3,v_4,v_5 \},\{ e_1,e_2,e_3,e_4,e_5,e_6 \} )$ 
given by the following set of chemical reactions, 
\begin{align}
    e_1 &: \text{(input)} \to v_1, \notag \\
    e_2 &: \text{(input)} \to v_2, \notag \\
    e_3 &: v_1 + v_2 \to v_3 + v_4 , \notag\\
    e_4 &: v_3 \to v_5 ,  \\
    e_5 &: v_4 \to \text{(output)}, \notag \\
    e_6 &: v_5 \to \text{(output)} .\notag  
\end{align}
The stoichiometric matrix of the network is 
\begin{equation}
S =
\begin{pmatrix} 
1&0&-1&0 &0&0 \\
0&1&-1&0 &0&0 \\
0&0&1 &-1&0&0 \\
0&0&1 &0&-1&0 \\
0&0&0 &1&0&-1 
\end{pmatrix}. 
\end{equation}
It can be drawn as 
\begin{equation}
\begin{tikzpicture}[bend angle=45] 
  \node[species] (v1) at (0,1) {$v_1$}; 
  \node[species] (v2) at (0,0) {$v_2$};
  \node[species] (v3) at (3,1) {$v_3$};
  \node[species] (v4) at (3,0) {$v_4$}; 

  \node (d1) at (-1.25,1) {}; 
  \node (d2) at (-1.25,0) {};  

  \node[species] (v5) at (4.5,1) {$v_5$}; 
  \node (d3) at (4.25,0) {}; 
  \node (d4) at (5.75,1) {}; 

  \node[reaction] (e3) at (1.5,0.5) {$e_3$}; 

   \draw [-latex,line width=0.5mm]
   (d1) edge node[below]{$e_1$} (v1);
   \draw [-latex,line width=0.5mm]
   (d2) edge node[below]{$e_2$} (v2);

  \draw [-latex,line width=0.5mm] (v1) -- (e3); 
  \draw [-latex,line width=0.5mm] (v2) -- (e3); 
   \draw [-latex,line width=0.5mm] (e3) -- (v3); 
   \draw [-latex,line width=0.5mm] (e3) -- (v4); 

   \draw [-latex,line width=0.5mm]
   (v3) edge node[below]{$e_4$} (v5);
   \draw [-latex,line width=0.5mm]
   (v4) edge node[below]{$e_5$} (d3);
   \draw [-latex,line width=0.5mm]
   (v5) edge node[below]{$e_6$} (d4);

\end{tikzpicture}
\end{equation}
We represent a monomolecular reaction by a single arrow, 
and we use a rectangle to represent a multimolecular reaction. 
In this network, $e_3$ is a multimolecular reaction 
and others are all monomolecular. 
The action of the boundary operator is, for example, 
\begin{equation}
    \p_1 e_4 = v_5 - v_3, 
    \quad 
    \p_1 e_3 = v_3 + v_4 - v_1 - v_2, 
    \quad 
    \p_1 e_1 = v_1, 
\end{equation}
and so on. 
Those are intuitively understood from the figure. 
The network is open, since we have inputs from the outside 
($e_1$ and $e_2$) and outputs to the external world ($e_5$ and $e_6$). 
For example, $s (e_1) (v_i) = 0$ for any $v_i \in V$. 
The action of $\p^\dag_1$ is 
\begin{equation}
    \p_1^\dag v_1 = e_1 - e_3, 
    \quad 
    \p_1^\dag v_3 = e_3 - e_4, 
\end{equation}
for example. 
Namely, the operator $\p_1^\dag$ measures 
the net inflow of the reactions on a vertex. 
\end{example}

The {\it chemical concentrations} and {\it reaction rates} are 
$\mathbb R$-valued linear maps 
over 0-chains and 1-chains, respectively, 
\begin{equation}
  C^n (\Gamma): C_n (\Gamma) \to \mathbb R, 
\end{equation}
for $n=0,1$. 
Given an $x \in C^0 (\Gamma)$, 
$x(v_i) \in \mathbb R$ 
represents the concentration of the chemical species $v_i$. 
Similarly, for a given $r \in C^1 (\Gamma)$, 
$r(e_A) \in \mathbb R$ represents the rate of the reaction $e_A$. 
We will also use short-hand notations $x_i \coloneqq x(v_i)$
and $r_A \coloneqq r (e_A)$. 
We will also denote 
an element as a vector
as 
$\bm x \in C^0 (\Gamma)$
and 
$\bm r \in C^1 (\Gamma)$, 
where the components of 
$\bm x$ and $\bm r$ are given by $x_i$ and $r_A$, respectively.

We define a {\it coboundary operator} in a usual way 
using the boundary operator, 
\begin{equation}\label{eq:coboundary}
(\rd_0 x) (e_A) \coloneqq 
x(\p_1 e_A) 
= 
x
\left( 
\sum_i (S^T)_{Ai} \, v_i 
\right)
= 
\sum_i (S^T)_{Ai}\, x( v_i )  ,
\end{equation}
where we have used the linearity of the map $x$. 
Thus, we can identify the coboundary operator 
that acts on the chemical concentration 
$\bm x \in C^0 (\Gamma)$ 
as the multiplication of the matrix $S^T$. 

We define the inner product 
of $n$-cochains as \footnote{
More generally, 
one may define the inner product with a weight function as 
\[
  \langle f, g \rangle_n  
  \coloneqq \sum_{ c\in C_n (\Gamma) } w(c) f(c) g(c), 
\]
where $w$ is a $\mathbb R$-valued function over $C_n(\Gamma)$. 
}
\begin{equation}
  \langle x, y \rangle_0 
  \coloneqq \sum_{ i } x(v_i) y(v_i), 
  \quad 
  \langle r,s \rangle_1
  \coloneqq \sum_{ A } r(e_A) s(e_A),
\end{equation} 
where $x,y \in C^0(\Gamma)$ and $r,s \in C^1(\Gamma)$. 
With these inner products, 
the adjoint of the coboundary operator $\rd_n$,
$\rd_n^\dag : C^{n+1}(\Gamma) \to C^n (\Gamma)$, 
is defined by 
\begin{equation}
  \langle f_{n+1}, \rd_n g_n  \rangle_{n+1} 
  = 
  \langle \rd^\dag_n f_{n+1},  g_n \rangle_{n} . 
\end{equation}
Following the definition, 
we can identify $\rd^\dag_0$ as follows,  
\begin{equation}
  \begin{split}
  \langle 
  r , \rd_0 x
  \rangle_1 
  &= 
  \sum_{A} 
  r (e_A ) \, (\rd_0 x) (e_A )  \\
  &= 
  \sum_{i,A} 
  S_{iA} r (e_A)  x(v_i)  \\ 
  &\coloneqq 
  \sum_i  
  (\rd_0^\dag r)(v_i) \, x(v_i) \\
  &= 
  \langle \rd_0^\dag r,   x  \rangle_0 \, ,
    \end{split} 
\end{equation}
where $r \in C^1(\Gamma)$ and $x \in C^0(\Gamma)$. 
Thus, the action of $\rd_0^\dag$ is given by 
\begin{equation}\label{eq:ajoint coboundary}
(\rd^\dag_0 r)(v_i) \coloneqq \sum_A S_{iA} r(e_A). 
\end{equation}
By construction, the adjoint of coboundary operator satisfies $(\rd^\dag_0 r)(v_i)= r(\partial^\dag_1v_i)$.

\subsection{Homology, cohomology, and steady states}

With the (co)chains and (co)boundary operators defined above,
we can discuss (co)homology groups. 
We have the following chain complex, 
\begin{equation}
  \xymatrix{
  0 \ar[r] &
  C_1(\Gamma) 
  \ar[r]^{\p_1} 
  &
C_0(\Gamma) 
  \ar[r]  &  0
}  . 
\label{eq:chain-complex}
\end{equation}
Noting that the action of $\p_1$ is the multiplication
of the stoichiometric matrix $S$, 
we can identify the homology groups as
\begin{eqnarray}
  H_0 (\Gamma)
&=&
C_0(\Gamma)
/ \p_1 C_1 (\Gamma) 
=
C_0(\Gamma)
/ {\rm im\,} S
= \coker S ,
\\
  H_1 (\Gamma)
&=& \ker S . 
\end{eqnarray}

\begin{remark}\label{rem:identification}
Note that $C_0(\Gamma)$ is endowed with a standard inner product, with respect to which we can take the orthogonal linear subspace $({\rm im}\,S)^\perp$. 
Moreover, the restriction of the quotient map $C_0(\Gamma) \to \coker S$ to $({\rm im}\,S)^\perp$ induces an isomorphism $({\rm im}\,S)^\perp \xrightarrow{\cong} \coker S$.
Therefore, we can always regard $\coker S$ as a linear subspace of $C_0(\Gamma)$.
Note also that the orthogonal subspace $({\rm im}\,S)^\perp$ 
is the same as 
the kernel of the transpose of $S$, $\ker S^T$. 
Combined with the above observation, this implies that we can always identify $\coker S$ with $\ker S^T \subset C_0(\Gamma)$.
\end{remark}

Similarly, with the coboundary operator $\rd_0$, 
we can define a complex of cochains as 
\begin{equation}
  \xymatrix{
  0 \ar[r] &
  C^0(\Gamma) 
  \ar[r]^{\rd_0} 
  &
C^1(\Gamma) 
  \ar[r]  &  0
}  . 
\label{eq:gamma-comples}
\end{equation}
The associated cohomology groups are 
\begin{eqnarray}
  H^0 (\Gamma)
&=&
\{
\bm d \in C^0(\Gamma) \,|\, S^T \bm d =0
\}
=
({\rm im\,} S)^\perp
\cong 
C_0 (\Gamma)/{\rm im\,}S
=\coker S ,
\\
  H^1 (\Gamma)
&=& 
C^1 (\Gamma)
/ 
\rd_0 C^0(\Gamma)
=
C^1 (\Gamma)
/ {\rm im\, }S^T 
\cong 
({\rm im\, }S^T )^\perp
=
\ker S ,
\end{eqnarray}
where $(-)^\perp$ denotes taking orthogonal spaces with respect to the standard inner product on $C_0(\Gamma)$ and $C^1(\Gamma)$.

An Euler number for this complex can be defined as 
\begin{equation}
  \chi (\Gamma) 
  \coloneqq
  |H_0(\Gamma)| - |H_1(\Gamma)| 
  = 
  | \coker S| - |\ker S|,  
\end{equation}
where $|W|$ indicates the dimension of the vector space $W$. 

Several remarks on the homology and cohomology groups are in order: 
\begin{remark}
Since we consider the $\mathbb R$ coefficients, 
the homology and cohomology groups are the same, 
$H_n (\Gamma) \cong H^n (\Gamma)$ for $n=0,1$. 
\end{remark}
\begin{remark}
In the chemistry literature, 
the elements of $H_1 (\Gamma)$ are referred to as cycles, 
and this is consistent with the mathematical terminology. 
\end{remark} 
\begin{remark}
When the network is monomolecular and the corresponding network is a directed graph, the dimension $|H_0(\Gamma)|$ is the number of connected components. 
\end{remark}
\begin{remark}
  Similarly to the homology groups of topological spaces, 
Laplace operators can be defined 
and we can perform Hodge decomposition of $C^1 (\Gamma)$. 
See Appendix \ref{sec:app-hodge}. 
\end{remark}

The cohomology groups defined above are closely related 
to the steady states of a reaction network as we see below. 
Let us consider the time evolution of 
spatially homogeneous chemical concentrations. 
The change of the chemical concentration 
is driven by the reactions. 
The time derivative of the concentration 
of species $v_i$ is given by 
the divergence of the reaction rate, 
\begin{equation}
  \frac{d}{dt} x (v_i) = (\rd^\dag_0 r) (v_i) , 
\end{equation}
which is more explicitly written as 
\begin{equation}
  \frac{d}{dt} x_i (t) = \sum_A S_{iA} \, r_A . 
  \label{eq:rate}
\end{equation}
To solve the rate equations, 
we have to specify {\it kinetics} of chemical reactions, 
such as the mass-action kinetics
and the Michaelis-Menten kinetics. 
A reaction's kinetics gives the reaction rate $r_A$ as a function of 
its substrate concentrations (i.e., the concentrations of species with $y_{iA}>0$) and parameters, 
$r_A = r_A (\bm x; k_A)$, where  $k_A$ represents  any one of the parameters for the $A$-th reaction; for example, in  the Michaelis-Menten kinetics, $k_A$ represents the Michaelis constant or the maximum rate. 

The elements of $H^0 (\Gamma)$ and $H^1 (\Gamma)$ \footnote{
Although the natural choice is to consider $x$ and $r$ as the elements of cohomology groups, 
we can equivalently consider them as elements of homology groups, 
since they are isomorphic in the current setting. 
}
characterize the steady states of chemical reaction networks. 
The rate equation (\ref{eq:rate}) at the steady state reads 
\begin{equation}
(\rd^\dag_0 r ) (v_i) = 0 , 
\text{   or equivalently   } 
\sum_A S_{iA} \, r_A = 0, 
\end{equation} 
which means that 
the steady-state reaction rate is an element of the 
kernel of $S$, $\bm r \in \ker S \cong H^1 (\Gamma)$. 
The cokernel of $S$ is related to conserved quantities of the system. 
Given $\bm d \in \coker S \cong H^0 (\Gamma)$, 
that satisfies $\sum_i d_i S_{iA}  = 0$, 
we have 
\begin{equation}
  \frac{d}{dt} 
 \langle d, x \rangle_0 
=
  \frac{d}{dt} (\sum_i d_i x_i) 
  = 
  \sum_{i,A} d_i S_{iA} r_A 
  = 0. 
\end{equation}
Thus, 
the linear combination $\sum_i d_i x_i$ is independent of time 
and hence is conserved. 
For this reason, we refer to the elements of $\coker S$ as {\it conserved charges} \footnote{
``Conserved moiety'' may be more chemistry-oriented terminology. 
}.
To find the steady-state solutions, 
we have to specify the value of all the conserved charges. 
A steady state is specified 
by an element of $H^0 (\Gamma)$ and $H^1(\Gamma)$, 
\begin{equation}
  \sum_{\bar \alpha} \ell_{\bar \alpha }\, \bm d^{\bar \alpha} 
  \in H^0 (\Gamma), 
  \quad 
  \sum_\alpha \mu_\alpha (\bm k, \bm \ell) \, \bm c^\alpha \in H^1 (\Gamma), 
\end{equation}
where
$\{\bm d^{\bar\alpha}\}$
and $\{\bm c^{\alpha}\}$ 
are basis vectors of $H^0(\Gamma)$ and $H^1(\Gamma)$, respectively. 
The coefficients $\mu_\alpha (\bm k, \bm \ell)$ depend 
on the parameters $\bm k$ and $\bm \ell$. 

\begin{example}\label{ex:v4e5}
We consider a network $\Gamma = (\{v_1,v_2,v_3,v_4 \},\{e_1,e_2,e_3,e_4,e_5 \})$ 
with the following reactions, 
\begin{align}
    e_1 &: \text{(input)} \to v_1, \notag \\
    e_2 &: v_1 \to v_2, \notag \\
    e_3 &: v_2 \to \text{(output)} ,\\
    e_4 &: v_1+v_2 \to v_3+v_4 , \notag \\
    e_5 &: v_3+v_4 \to v_1+v_2. \notag
\end{align}
The network structure can be drawn as 
\begin{equation}
    \begin{tikzpicture}[bend angle=45] 
    \node[species] (v1) at (0,0) {$v_1$}; 
    \node[species] (v2) at (2,0) {$v_2$};
 
    \node (d1) at (-1.25,0) {};
    \node (d2) at (3.25,0) {};
    
    \node[reaction] (e4) at (0, -1.65) {$e_4$}; 

    \node[reaction] (e5) at (2,-1.65) {$e_5$}; 

    \node[species] (v3) at (0,-3.25) {$v_3$}; 
    \node[species] (v4) at (2,-3.25) {$v_4$};

    \draw[-latex, line width=0.5mm] (d1) edge node[below] {$e_1$} (v1); 
    \draw[-latex, line width=0.5mm] (v1) edge node[below] {$e_2$} (v2); 
    \draw[-latex, line width=0.5mm] (v2) edge node[below] {$e_3$} (d2); 
  
    \draw[-latex, line width=0.5mm] (v1) -- (e4); 
    \draw[-latex, line width=0.5mm] (v2) -- (e4); 
    \draw[-latex, line width=0.5mm] (e4) -- (v3); 
    \draw[-latex, line width=0.5mm] (e4) -- (v4); 
 
    \draw[-latex, line width=0.5mm] (v3) -- (e5); 
    \draw[-latex, line width=0.5mm] (v4) -- (e5); 
    \draw[-latex, line width=0.5mm] (e5) -- (v1); 
    \draw[-latex, line width=0.5mm] (e5) -- (v2); 
 
\end{tikzpicture}
\end{equation}
We here take the mass-action kinetics, 
and the equations of motion are written as 
\begin{equation}
    \frac{d}{dt}
    \begin{pmatrix}
    x_1 \\
    x_2 \\
    x_3 \\
    x_4 
    \end{pmatrix}
    = 
    \begin{pmatrix}
      1 & -1 & 0 & -1 & 1 \\
      0 & 1 & -1 & -1 & 1 \\
      0 & 0 & 0 & 1 & -1 \\
     0 & 0 & 0 & 1 & -1 
    \end{pmatrix}
    \begin{pmatrix}
      r_1 \\
      r_2 \\
      r_3 \\
      r_4 \\
      r_5 
    \end{pmatrix}, 
    \quad 
    \begin{pmatrix}
      r_1 \\
      r_2 \\
      r_3 \\
      r_4 \\
      r_5 
    \end{pmatrix}
    = 
    \begin{pmatrix}
      k_1 \\
      k_2 x_1 \\
      k_3 x_2 \\
      k_4 x_1 x_2 \\
      k_5 x_3 x_4 
    \end{pmatrix}, 
\end{equation}
where $x_i = x(v_i)$ and $r_A = r (e_A)$ are the concentration 
and reaction rate for the species $v_i$ and reaction $e_A$, respectively. 
The kernel and cokernel of the stoichiometric matrix are given by 
\begin{align}
    \ker S &= {\rm span\,} \{ 
    \begin{pmatrix}
      1 & 1 & 1 & 0 & 0 
    \end{pmatrix}^T,
    \begin{pmatrix}
      0 & 0 & 0 & 1 & 1 
    \end{pmatrix}^T 
    \},  \\
\coker S &= {\rm span\,} \{ 
    \begin{pmatrix}
     0 & 0 & 1 & -1 
    \end{pmatrix}^T     
    \} ,
\end{align}
where ${\rm span\,}\{ \bm v_1, \bm v_2, \cdots \}$ 
indicates the vector space spanned by 
vectors $\bm v_1, \bm v_2, \cdots$. 
The cokernel is one-dimensional and 
the system has one conserved charge. 
To find the steady states, we need to specify the value 
of the charge as 
\begin{equation}
    \ell = x_3 - x_4 . 
\end{equation}
The steady-state reaction rates and concentrations are 
\begin{align}
    \bar{\bm r}
    &= k_1 
    \begin{pmatrix}
      1 & 1 & 1 & 0 & 0 
    \end{pmatrix}^T
    +
    \frac{k_4 k_1^2}{k_2 k_3} 
    \begin{pmatrix}
      0 & 0 & 0 & 1 & 1 
    \end{pmatrix}^T ,\\ 
    \bar {\bm x} 
    &= 
    \begin{pmatrix}
      \frac{k_1}{k_2} & \frac{k_1}{ k_3}
      & \frac{1}{2} 
      \left(
       \ell + \sqrt{\ell^2 + 4 K}
      \right)
      & 
      & \frac{1}{2} 
      \left(
      - \ell + \sqrt{\ell^2 + 4 K}
      \right)
    \end{pmatrix}^T, 
\end{align}
where we set $K \coloneqq k_4 k_1^2 / k_2 k_3$.
The vector $\bar {\bm r}$ is spanned by the basis vectors of $\ker S$ 
and their coefficients are $\mu_\alpha$. 

\end{example}

\subsection{Subnetworks}\label{sec:sub-n}

Let us consider a subset of chemicals and reactions, 
$\gamma \subset \Gamma$, 
which we specify by $\gamma = (V_\gamma, E_\gamma)$
with $V_\gamma \subset V$ and $E_\gamma \subset E$. 
Correspondingly, we have a submatrix 
$S_\gamma$ of the stoichiometric matrix $S$, 
whose components are given by 
\begin{equation}
  (S_\gamma)_{iA} = S_{iA} , 
\end{equation}
where the indices are restricted to those of the subnetwork, 
$v_i \in V_\gamma, e_A \in E_\gamma$. 
We denote 
the space of relative chains by 
\begin{equation}
 C_n (\gamma)
  \coloneqq C_n(\Gamma) / C_n (\Gamma \setminus \gamma) ,
\end{equation}
where $\Gamma \setminus \gamma \coloneqq (V \setminus V_\gamma, E \setminus E_\gamma)$ 
is the complement of the subnetwork $\gamma$. 
The homology and cohomology groups for the subnetwork 
can be defined similarly. 
The chain complex for a subnetwork $\gamma$ is 
\begin{equation}
 \xymatrix{
  0 \ar[r] &
  C_1(\gamma) 
  \ar[r]^{ \p_1}
  &
C_0(\gamma) 
  \ar[r]  &  0
} , 
\label{eq:rel-hom} 
\end{equation}
where 
the action of the boundary operator $\p_1$ 
on the basis of $C_1(\gamma)$ is defined with 
the partial stoichiometric matrix $S_\gamma$, 
\begin{equation}
    \p_1 e_A =\sum_i (S^T_\gamma)_{Ai}\, v_i . 
\end{equation}
The associated homologies with the complex (\ref{eq:rel-hom}) are 
\begin{eqnarray}
  H_0 (\gamma)
  &=&
  C_0(\gamma) / \p_1 C_1 (\gamma) 
  =
  C_0(\gamma) / {\rm im\,} S_{\gamma} 
  = 
  \coker S_\gamma, 
\\
  H_1 (\gamma)
  &=&
    \ker S_\gamma . 
\end{eqnarray}
The Euler number for a subnetwork is given by 
\begin{equation}
  \chi(\gamma) 
 \coloneqq
 | H_0(\gamma)| -| H_1(\gamma)| . 
\end{equation}
Note that 
\begin{equation}
  \chi(\gamma) 
  = 
 | H_0(\gamma)| -| H_1(\gamma)|
  = 
  | C_0(\gamma)| - |C_1(\gamma)|
  = |V_\gamma| - |E_\gamma|  . 
  \label{eq:euler-sub}
\end{equation}

The value of the concentrations and 
reaction fluxes
inside a subset $\gamma$ 
are given by 
$\mathbb R$-valued functions over 
the space of chemicals and reactions, 
\begin{equation}
  C^n (\gamma): C_n(\gamma) \to \mathbb R . 
\end{equation}
The cohomology for subnetworks 
can be defined similarly to the homology.

\subsection{Mayer-Vietoris exact sequence}

In this subsection, we give a long exact sequence of homology groups that connects local and global information. 
Suppose that there are two 
subnetworks 
$\gamma_1, \gamma_2 \subset \Gamma$, 
which consist of 
$\gamma_1 = (V_{\gamma_1}, E_{\gamma_1})$
and 
$\gamma_2 = (V_{\gamma_2}, E_{\gamma_2})$. 
We can consider the intersection and union of the subnetworks, 
\begin{equation}
    \gamma_1 \cap \gamma_2 
    \coloneqq 
    (V_{\gamma_1} \cap V_{\gamma_2},E_{\gamma_1} \cap E_{\gamma_2} ), 
    \quad 
    \gamma_1 \cup \gamma_2 
    \coloneqq 
    (V_{\gamma_1} \cup V_{\gamma_2},E_{\gamma_1} \cup E_{\gamma_2} ). 
\end{equation}
The exact sequence \eqref{eq:MV} below explains the relationship among cohomology groups of $\gamma_1 \cup \gamma_2$, $\gamma_1$, $\gamma_2$ and $\gamma_1 \cap \gamma_2$.
Regarding the family $\{\gamma_1, \gamma_2\}$ as a `covering' of $\gamma_1 \cup \gamma_2$, we can think of Eq.~\eqref{eq:MV} as an analogue of the Mayer-Vietoris sequence associated with an open covering of a topological space. 
Following the usual technique in topology, we will derive the long exact sequence from a short exact sequence of chain complexes.
We have the following short exact sequence of chain complexes, 
\begin{equation}
  \xymatrix{
    &
    0 
    \ar[d]
    &
    0 \ar[d]
    & 
    0 \ar[d]
    \\
   0 \ar[r] &
   C_1(\gamma_1 \cap \gamma_2 )
   \ar[r]^{f_1}
   \ar[d]^{\p_1} 
   &
   C_1(\gamma_1) \oplus C_1(\gamma_2)
   \ar[r]^{g_1}
   \ar[d]^{\p_1}
   & 
   C_1(\gamma_1 \cup \gamma_2 ) 
   \ar[r]
   \ar[d]^{\p_1}
   & 0
   \\
0
 \ar[r] 
&  
   C_0(\gamma_1 \cap \gamma_2 )
   \ar[r]^{f_0}
   \ar[d]   
   & 
   C_0(\gamma_1) \oplus C_0(\gamma_2) 
   \ar[r]^{g_0}
   \ar[d] 
   &
   C_0(\gamma_1 \cup \gamma_2 )
   \ar[r]
   \ar[d]
   &
   0
   \\
   & 
   0 
   & 0 
   & 0 
 }
 \label{eq:mv-ses}
\end{equation}
where the horizontal maps are given by 
\begin{equation}
  f_n : c \mapsto (c, -c), 
  \quad 
  g_n : (c_1,c_2) \mapsto c_1 + c_2 . 
\end{equation}
By applying the snake lemma to Eq.~\eqref{eq:mv-ses}, we obtain
\begin{equation}\label{eq:MV}
  \xymatrix{
   0 \ar[r] &
   H_1(\gamma_1 \cap \gamma_2 )
   \ar[r]
   &
   H_1(\gamma_1) \oplus H_1(\gamma_2)
   \ar[r] 
   &
   H_1(\gamma_1 \cup \gamma_2 ) \\
   \ar[r] &  
   H_0(\gamma_1 \cap \gamma_2 )
   \ar[r]
   & 
   H_0(\gamma_1) \oplus H_1(\gamma_2)
   \ar[r] 
   &
   H_0(\gamma_1 \cup \gamma_2 )
   \ar[r]
   &
   0.
 } 
\end{equation}
In general, if there is an exact sequence of finite-dimensional vector spaces, the alternating sum of the dimensions of them is equal to zero.
Therefore, the exact sequence \eqref{eq:MV} implies 
\begin{equation}
    \chi(\gamma_1 \cup \gamma_2)  
    = 
    \chi(\gamma_1 )
    +  
    \chi(\gamma_2 )
    -
    \chi(\gamma_1 \cap \gamma_2) . 
\end{equation}

\section{Law of localization} \label{sec:lol}

A sensitivity analysis studies the response of the system to the perturbations of reaction parameters or initial conditions (conserved charges). 
In the context of metabolic networks, a theoretical framework called  the metabolic control analysis has been developed  \cite{kacser1973control,fell1992metabolic,szathmary1993deleterious,maclean2010predicting,bagheri2004evolution}. 
Under the mass-action framework, biologically insightful results have been obtained regarding the sensitivity to conserved charges \cite{shinar2010structural,eloundou2016network} as well as stability properties of stable states  \cite{feinberg1995existence,feinberg1995multiple,craciun2005multiple,craciun2006multiple}, although the mass-action law is not necessarily appropriate for some biological systems.
Among the studies on sensitivity analysis, the structural sensitivity analysis \cite{MOCHIZUKI2015189,doi:10.1002/mma.3436,doi:10.1002/mma.4668} aims at constraining the response of reaction systems from the network structure alone. 

In this section, we first review 
the structural sensitivity analysis and the law of localization  \cite{PhysRevLett.117.048101,PhysRevE.96.022322, PhysRevE.98.012417}. 
For a given subnetwork, 
we assign a nonnegative integer, which we call the influence index. 
The influence index is determined from the topology of the subnetwork, 
and plays a decisive role in structural sensitivity.
When the influence index is zero, 
the perturbation of the parameters and conserved charges 
inside the subnetwork does not affect the rest of the network. 
Such a structure is called a buffering structure. 
In Sec.~\ref{sec:submodularity}, we prove that the influence index 
is submodular as a function over subnetworks. 
As a corollary of this property, 
we show that buffering structures are closed under intersection and union.

\subsection{Structural sensitivity analysis}

At the steady state, 
the reaction rates and the chemical concentrations satisfy 
\begin{align}
\sum_A  S_{iA}   r_A ( \bm x (\bm k, \bm \ell), k_A ) &= 0, \label{eq:sr} \\
 \sum_i d^{\bar \alpha}_i x_i (\bm k, \bm \ell)
 &= \ell^{\bar\alpha},  \label{eq:l-dx} 
\end{align} 
where 
$\{\bm d^{\bar\alpha}\}$ is a basis of $\coker S$
and 
the second equation specifies the values of conserved charges. 
Considerable effort has been devoted to 
the study of the existence or uniqueness of steady states 
under the mass-action kinetics \cite{feinberg2019foundations}. 
In the current analysis, we assume the existence of a steady state, 
and we focus on how it is perturbed under the change of parameters. 
The steady-state values of 
the concentrations and reaction rates are determined
by the values of rate parameters and conserved charges, 
$\{k_A, \ell^{\bar\alpha}\}$.
The reaction rates $r_A (\bm x (\bm k, \bm \ell), k_A)$ 
have explicit dependence on $k_A$, 
and also dependence on $\bm k$ and $\bm \ell$ through $x_i (\bm k, \bm \ell)$. 
Equation (\ref{eq:sr}) means that the reaction rates are in the kernel of $S$ and  can be expanded 
using a basis $\{\bm c^\alpha\}$ of $\ker S$ as 
\begin{equation}
  r_A ( \bm x (\bm k, \bm \ell), k_A ) 
  =- \sum_\alpha \mu_\alpha (\bm k, \bm \ell) c^\alpha_A . 
  \label{eq:r-mu-c}
\end{equation}
We are interested in the sensitivity of the reaction 
rates and concentrations under the perturbation of the parameters, 
\begin{equation}
  k_A \mapsto k_A + \delta k_A, 
  \quad 
  \ell^{\bar\alpha} \mapsto 
  \ell^{\bar\alpha} + \delta   \ell^{\bar\alpha}. 
\end{equation}
By taking the derivative 
of Eqs.~(\ref{eq:l-dx}) and (\ref{eq:r-mu-c})
with respect to $k_B$ and $\ell^{\bar\beta}$, 
we obtain the following equations, 
\begin{align}
 \sum_i \frac{\p r_A}{\p x_i} 
  \frac{\p x_i}{\p k_B} 
  + 
  \frac{\p r_A}{\p k_B}
  &= - \sum_\alpha \frac{\p \mu_\alpha }{\p k_B} c^\alpha_A  ,\\
  \sum_i \frac{ \p r_A }  {\p x_i } 
 \frac{\p x_i}{\p \ell^{\bar\alpha }} 
 &= - 
 \sum_\alpha  \frac{ \p \mu_\alpha }  {\p \ell^{\bar\alpha }} c^\alpha_A , \\
\sum_i d_i^{\bar\alpha} \frac{\p x_i}{\p k_A} &= 0,  
 \\
\sum_i  d_i^{\bar\alpha} \frac{\p x_i}{\p \ell^{\bar\beta}}
 &=  \delta^{\bar\alpha \bar\beta} . 
\end{align}
Note that $r_A ( \bm x(\bm k, \bm \ell), k_A)$ 
depends explicitly on $k_A$ 
and also depends implicitly on 
$\bm k$ and $\bm \ell$ through $\bm x$. 
The equations can be compactly written in the matrix form, 
\begin{equation}
  A \, 
  \begin{pmatrix}
 \p_B x_i \\
  \p_B \mu_\alpha     
  \end{pmatrix}
  = - 
  \begin{pmatrix}
  \p_B  r_A \\
    \bm 0
\end{pmatrix}
, 
\quad 
A \, 
\begin{pmatrix}
\p_{\bar\beta} x_i \\
\p_{\bar\beta} \mu_\alpha     
\end{pmatrix}
= 
\begin{pmatrix}
 \bm 0 \\
 \delta^{\bar\alpha \bar\beta }
\end{pmatrix}, 
\label{eq:sensitivity-1}
\end{equation}
where 
$\p_B \coloneqq \partial /\partial k^B$ 
,
$\p_{\bar\beta} \coloneqq \partial /\partial \ell^{\bar\beta}$, 
and 
we have introduced a partitioned square matrix, 
\begin{equation}
  A 
  \coloneqq 
  \begin{pmatrix}
    \p_i r_A 
    & c^\alpha_A \\
    d^{\bar\alpha}_i & \bm 0 
  \end{pmatrix},
\end{equation}
 where the upper-left block is an $|A|\times |i|$ matrix whose $(A,i)$-th element is given by $\frac{\p r_A}{\p x_i}$ evaluated at the steady state, the upper-right one  an $|A| \times |\alpha|$ matrix consisting of the basis $\{\bm c^\alpha\}$ of $\ker S$, the lower-left one $d^{\bar\alpha}_i$ an $|\bar\alpha|\times |i|$ matrix consisting of the basis $\{\bm d^{\bar \alpha}\}$ of $\coker S$, and the lower-right one the $|\bar \alpha|\times |\alpha|$ zero matrix. Here, we use the notation that index $i$ for chemicals runs from $1$ to $|i|$.  The matrix $A$ is square due to the identity,
\begin{equation}
  |i| + |\alpha| = |A| + |\bar\alpha|. 
\end{equation}

One can see from Eq.~(\ref{eq:sensitivity-1}) that 
the response to the change of the parameter is determined 
by the inverse of the matrix $A$,
\begin{equation}
  \begin{pmatrix}
 \p_B x_i \\
  \p_B \mu_\alpha     
  \end{pmatrix}
  = - 
  A^{-1}
  \begin{pmatrix}
  \p_B  r_A \\
    \bm 0
\end{pmatrix}
, 
\quad 
\begin{pmatrix}
\p_{\bar\beta} x_i \\
\p_{\bar\beta} \mu_\alpha     
\end{pmatrix}
=  A^{-1}
\begin{pmatrix}
 \bm 0 \\
 \delta^{\bar\alpha \bar\beta }
\end{pmatrix}.  
\label{eq:sensitivity-2}
\end{equation}
We refer to $A$ as the {\it A-matrix} (``A'' indicates that it is an \underline{a}ugmented matrix). 
Its inverse, $- A^{-1}$ determines the sensitivity of the system and is called the {\it sensitivity matrix}. 
If we partition $A^{-1}$ as 
\begin{equation}
  A^{-1} = 
  \begin{pmatrix}
    (A^{-1})_{iA} & (A^{-1})_{i \bar\alpha  } \\
    (A^{-1})_{ \alpha A } & (A^{-1})_{ \alpha \bar\alpha }  
  \end{pmatrix}, 
\end{equation}
and noting that $\p_B r_A$
is a diagonal matrix, 
$\p_B r_A \propto \delta_{BA}$, the responses of steady-state  concentrations and reaction rates [or equivalently, the coefficients $\mu_\alpha$ in Eq.~\eqref{eq:r-mu-c}] to the perturbations of $k_A$ and $\ell^{\bar\alpha}$ are given by 
\begin{equation}
  \p_A x_i \propto  (A^{-1})_{iA} , 
  \quad 
  \p_{\bar\alpha} x_i \propto  (A^{-1})_{i \bar\alpha } , 
  \quad 
  \p_{\bar\alpha} \mu_{\alpha} \propto (A^{-1})_{\alpha \bar\alpha },
  \quad 
  \p_{A} \mu_{\alpha} \propto (A^{-1})_{\alpha A} . 
  \label{eq:response}
\end{equation}
In this paper, we consider the following class of chemical reaction systems: 
\begin{definition}[Regularity of a chemical reaction network with kinetics] 
A chemical reaction network with kinetics 
is called {\it regular}, 
if it admits a stable steady state 
and the associated $A$-matrix is invertible. 
\end{definition}
Note that whether a reaction system is regular or not 
depends on the choice of kinetics. 
Throughout the paper, 
we assume the regularity unless otherwise stated 
so that $A$ is invertible and the response of the system is well-defined. 
The regularity implies the asymptotic stability of the steady state, 
through the relation between $\det A$ and the determinant of the Jacobian~\cite{PhysRevE.98.012417}.

\subsection{Law of localization}\label{sec:lol-lol}

\begin{definition}[Output-completeness] \label{def:oc} 
When a subnetwork $\gamma = (V_\gamma, E_\gamma)$ 
satisfies the condition that 
$E_\gamma$ includes all the chemical reactions 
affected by $V_\gamma$, 
$\gamma$ is called {\it output-complete}. 
\end{definition}

\begin{definition}[Influence index] 
For an output-complete subnetwork $\gamma$, 
the {\it influence index} is defined by 
\begin{equation}
  \lambda(\gamma)
  \coloneqq 
-  |V_\gamma |
 + 
  |E_\gamma |
  - 
  |(\ker S)_{{\rm supp\,} \gamma }|
  + 
  | P^0_\gamma (\coker S) | . 
  \label{eq:index-def} 
\end{equation}
\end{definition}
The definitions of the spaces 
that appear in the influence index 
are given as follows: 
\begin{eqnarray}
 (\ker S)_{{\rm supp}\,\gamma}
 &\coloneqq &
 \left\{
  \bm c 
  \, \middle| \,
  \bm c \in \ker S, 
  P^1_\gamma \bm c = \bm c 
  \right\}, 
   \\
   P^0_\gamma (\coker S) 
   &\coloneqq &
   \left\{
   P^0_\gamma \bm d 
   \, \middle| \,
   \bm d \in \coker S 
   \right\}, 
\end{eqnarray}  
where $S$ is the stoichiometric matrix, 
$P^0_\gamma$ and $P^1_\gamma$
are the projection matrices to $\gamma$
in the space of chemical species and reactions, respectively. 
Namely, $(\ker S)_{{\rm supp}\,\gamma}$
is the space of vectors of $\ker S$ 
supported inside $\gamma$, 
and $P^0_\gamma (\coker S)$
is the projection of $\coker S$ to $\gamma$. 
Here, recall from Remark \ref{rem:identification} that we regard $\coker S$ as a subspace of $C_0(\Gamma)$ via the identification $\coker S \cong ({\rm im}\,S)^\perp$.
We will use similar identifications throughout this paper.

\begin{remark}\label{rem:lambda-ge-0}
The influence index is nonnegative, $\lambda(\gamma) \ge 0$, 
for a regular chemical reaction network. 
It will be shown in the proof of Theorem~\ref{th:lol}. 
\end{remark}

\begin{theorem}[Law of localization]  \label{th:lol}
Let $\gamma$ be an output-complete subnetwork of 
a regular chemical reaction network $\Gamma$. 
When $\gamma$ is a buffering structure, $\lambda(\gamma)=0$, 
chemical concentrations and reaction rates 
outside $\gamma$ do not change 
under the perturbation of rate parameters or conserved charges 
inside $\gamma$. 
\end{theorem}

\begin{definition}[Buffering structures] 
For a given chemical reaction network $\Gamma$, 
an output-complete subnetwork $\gamma$  
with the vanishing influence index, $\lambda(\gamma)=0$, 
is called a {\it buffering structure}. 
\end{definition}

\begin{example}
The influence index of the empty subnetwork is zero. 
The index of the whole network $\Gamma$ is also zero, 
\begin{equation}
  \begin{split}
    \lambda(\Gamma)
    &= 
  -  |C_0 (\Gamma)|
    + 
    |C_1 (\Gamma)|
    - 
    |\ker S|
    +
    |\coker S|
    \\
    &= 
    |H_0(\Gamma)|
    -
    |H_1(\Gamma)| 
   - (|C_0 (\Gamma)|
    - 
    |C_1 (\Gamma)|
   )
    \\
    &= 0 . 
  \end{split}
\end{equation}
This is natural in the sense that 
there is no ``outside'' of the whole network. 
\end{example}

\begin{example}
Let us take the same network as Example~\ref{ex:v4e5}. 
The A-matrix of this system is 
\begin{equation}
    A 
    =
    \left(
    \begin{array}{cccc|cc}
      0 & 0 & 0 & 0             &1&0 \\
      r_{2,1} & 0 & 0 & 0       &1&0    \\
    0&r_{3,2} & 0 & 0            &1&0\\
    r_{4,1} & r_{4,2} & 0 & 0   &0&1 \\
    0 & 0 & r_{5,3} & r_{5,4}   &0&1 \\ \hline 
    0 & 0 & 1 & -1 & 0 & 0 
    \end{array}
    \right) , 
\end{equation}
where $r_{A,i} \coloneqq \p r_A / \p x_i$ and it is evaluated at the steady state. 
With this matrix $A$, 
the responses of the concentration and reaction rates 
to the change of parameters and the value of conserved charges can be obtained by Eq.~(\ref{eq:sensitivity-2}). 
The subnetwork 
$\gamma_1 = (\{v_3,v_4\},\{e_5\})$ 
is output-complete 
and is a buffering structure, since 
$\lambda(\gamma_1) = -2 + 1 - 0 + 1 = 0$. 
The output-complete subnetwork 
$\gamma_2 = (\{v_3,v_4\},\{e_4,e_5\})$ 
is also a buffering structure, 
$\lambda(\gamma_2) = -2 + 2 -1 +1 =0$, 
which contains a cycle supported on $\gamma_2$. 
This explains the fact that $x_1$ and $x_2$ 
do not depend on the value of conserved charge $\ell = x_3 - x_4$. 
The subnetwork 
$\gamma_3 = (\{v_1,v_3,v_4\},\{e_2,e_4,e_5\})$
is also a buffering structure, 
$\lambda(\gamma_3) = -3 + 3 - 1 + 1 = 0$, 
and hence 
$x_2$ does not depend on $k_2,k_4,k_5,$ and $\ell$. 
\end{example}

\begin{proof} 
The law of localization follows from the structure of the matrix $A$. 
Given an output-complete subnetwork $\gamma$, 
we can bring the rows and columns associated with $\gamma$ 
in the way shown in Fig.~\ref{fig:mat-a-lol}. 
All the component of the lower-left part is zero, 
because 
the reaction rate $r_A$ outside $\gamma$
does not depend on the chemical species in $\gamma$
(since $\gamma$ is output-complete), 
and those cycles are supported in $\gamma$. 
\begin{figure}[tb]
  \centering
  \includegraphics[keepaspectratio, scale=0.3]{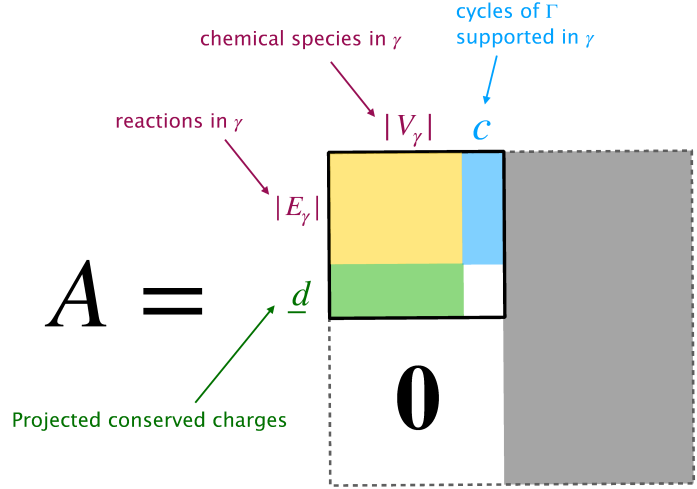} 
  \caption{
Structure of the A-matrix. 
The numbers $c$ and $\underline{d}$ are given by 
$c = |(\ker S)_{{\rm supp}\, \gamma}|$ 
and $\underline{d} = |P^0_\gamma (\coker S)|$. 
  }
  \label{fig:mat-a-lol}
\end{figure}
The index $\lambda(\gamma)$ measures 
how far the black rectangle on the upper-left corner is 
from a square matrix. 
The numbers $c$ and $\underline{d}$ in Fig.~\ref{fig:mat-a-lol} 
are given by 
$c = |(\ker S)_{{\rm supp}\, \gamma}|$ 
and $\underline{d} = |P^0_\gamma (\coker S)|$, 
which appear in Eq.~(\ref{eq:index-def}). 
Because of the assumption of regularity, we have $\det A \neq 0$, 
and the black rectangle on the upper-left corner 
should be vertically long 
(if it is horizontally long, the determinant vanishes), 
which is equivalent to the condition $\lambda(\gamma) \ge 0$. 

When $\lambda(\gamma) =0$, 
the black box in the upper-left corner is a square matrix. 
Then, $A^{-1}$ inherits the same structure, 
\begin{equation}
  A^{-1}
  = 
  \begin{pmatrix}
    * & * \\
    \bm 0 & * 
  \end{pmatrix}.
\end{equation}
Namely, if we denote the generic index of $(A^{-1})$ as $\mu,\nu, \cdots$ and 
write the index inside and outside $\gamma$ 
as $\mu^\star$ and $\mu'$, respectively, 
we have $(A^{-1})_{\mu' \nu^\star} = 0$. 
Because of this structure, 
\begin{equation}
 \p_{A^\star} x_{i'} \propto  (A^{-1})_{i' A^\star} = 0 , 
\end{equation}
which means that the concentrations out of $\gamma$ 
do not depend on the parameter $k_B$ inside $\gamma$. 
Consequently, we have 
\begin{equation}
  \p_{A^\star} r_{A'} \propto 
  \sum_{i'} \p_{i'} r_{A'} \p_{A^\star} x_{i'} = 0 , 
\end{equation}
where we used the fact that $r_{A'}$ only depends on 
the concentrations outside $\gamma$ because of the output-completeness. 
The same is true for the perturbation of the conserved charge, 
\begin{equation}
  \p_{{\bar\alpha}^\star} x_{i'} 
  \propto  (A^{-1})_{i' {\bar\alpha}^\star} = 0 , 
\quad 
\p_{{\bar\alpha}^\star} r_{A'} 
\propto 
\sum_{i'} \p_{i'} r_{A'} \p_{{\bar\alpha}^\star} x_{i'}
= 0 . 
\end{equation} 
\end{proof}

\subsection{Submodularity of the influence index}\label{sec:submodularity} 

The influence index $\lambda (\gamma)$ can be regarded as a function 
over subnetworks. 
We here show that the influence index 
satisfies an inequality. 
As a corollary, we show that the buffering structures 
are closed under union and intersection. 
This fact is useful in enumerating buffering structures in large reaction networks.

We first note that: 
\begin{itemize}
  \item Given output-complete subnetworks $\gamma_1, \gamma_2 \subset \Gamma$, 
  the union and intersection, $\gamma_1 \cup \gamma_2$ and 
$\gamma_1 \cap \gamma_2$, 
are also output-complete. 
This follows from the definition of output-completeness. 
\item A function $f(\gamma)$ 
over a set is called {\it submodular}, when it satisfies 
\begin{equation}
  f(\gamma_1 \cup \gamma_2)
  \le 
 f(\gamma_1 ) + f(\gamma_2 )
  - f(\gamma_1 \cap \gamma_2) . 
\end{equation}
When $\le$ is replaced with $\ge$, 
the function satisfying the replaced equation is 
called {\it supermodular}. 
\end{itemize}

\begin{theorem}
Let $\gamma_1, \gamma_2 \subset \Gamma$
be output-complete subnetworks. 
The influence index satisfies 
\begin{equation}
  \lambda (\gamma_1 \cup \gamma_2)
   \le 
\lambda (\gamma_1 )
+
\lambda (\gamma_2 )
 - 
 \lambda (\gamma_1  \cap \gamma_2 ) . 
 \label{eq:lm-ineq}
\end{equation}
Namely, $\lambda(\gamma)$ is a submodular function 
over output-complete subnetworks. 
\end{theorem}

\begin{proof}
We show that 
\begin{equation}
   \lambda(\gamma)
  =
- |V_\gamma|
+ |E_\gamma| 
-
 |(\ker S)_{{\rm supp}\,\gamma}|
+ |P^0_\gamma (\coker S)| 
 \label{eq:lambda-chi}
\end{equation}
is submodular. 
Recall that $\chi(\gamma)
=  |V_\gamma|- |E_\gamma|
=  |H_0 (\gamma)| - |H_1 (\gamma)| 
$ is the Euler number for subnetwork $\gamma = (V_\gamma, E_\gamma)$. 
We note that $\chi(\gamma)$ is a modular function, 
meaning that it satisfies 
\begin{equation}
\chi(\gamma_1 \cup \gamma_2)  
= 
\chi(\gamma_1 )
+  
\chi(\gamma_2 )
-
\chi(\gamma_1 \cap \gamma_2) , 
\end{equation}
which is derived from the Mayer-Vietoris exact sequence~\eqref{eq:MV}. 
Thus, it suffices to show that the last 
two terms on the right-hand side (RHS) of Eq.~(\ref{eq:lambda-chi})
are submodular. 
In fact, we show that each of them is submodular.

Let us first look at 
$|P^0_\gamma (\coker S)|$. 
If denote $W \coloneqq \coker S$, 
the submodularity of $|P^0_\gamma (\coker S)|$ reads 
\begin{equation}
 |P^0_{\gamma_1 \cup \gamma_2} W|
 \le 
 |P^0_{\gamma_1} W |
 + 
 |P^0_{\gamma_2} W |
 - 
 |P^0_{\gamma_1 \cap \gamma_2} W| . 
 \label{eq:p-w-submod}
\end{equation}
We prove this equation just after this proof. 
Thus, we have shown the submodularity of $|P^0_\gamma (\coker S)|$. 

Next, we show 
that $ |(\ker S)_{{\rm supp}\,\gamma}|$ is supermodular. 
Consider the following vector space, 
\begin{equation}
    Z \coloneqq (\ker S)_{{\rm supp}\,\gamma_1} + (\ker S)_{{\rm supp}\,\gamma_2} . 
\end{equation}
Its dimension is given by 
\begin{equation}
\begin{split}
    |Z| 
    &=
    |(\ker S)_{{\rm supp}\,\gamma_1}| + |(\ker S)_{{\rm supp}\,\gamma_2} |
    - 
    |(\ker S)_{{\rm supp}\,\gamma_1} \cap (\ker S)_{{\rm supp}\,\gamma_2} | \\
    &= 
    |(\ker S)_{{\rm supp}\,\gamma_1}| + |(\ker S)_{{\rm supp}\,\gamma_2} |
  - |(\ker S)_{{\rm supp}\,\gamma_1 \cap \gamma_2} |. 
\end{split}
\end{equation}
Since any element of $Z$ is supported in $\gamma_1 \cup \gamma_2$, we have 
$(\ker S)_{{\rm supp}\, \gamma_1 \cup \gamma_2} \supset  Z$, 
which implies $|(\ker S)_{{\rm supp}\, \gamma_1 \cup \gamma_2}| \ge |Z|$. 
Thus, we have shown that $|(\ker S)_{{\rm supp}\,\gamma}|$ is a supermodular function, 
and $- |(\ker S)_{{\rm supp}\,\gamma}|$ is submodular. 

Therefore, $|P^0_\gamma (\coker S)|$ and 
$- |(\ker S)_{{\rm supp}\,\gamma}|$ are both submodular function, and we obtain the claim. 
\end{proof}

\noindent
{\it Proof of Eq.~(\ref{eq:p-w-submod})}.  
Let us pick a basis of the space $W$ as 
$\{ \bm w_1, \cdots, \bm w_n \}$
and we denote the basis as a matrix, $w_{ij}$. 
The set of possible row and column 
indices are denoted as $V$ and $N$, respectively. 
For a subset of indices $v \subset V$, 
we denote the corresponding submatrix of $w_{ij}$ as 
$w[v,N]$. 
With this notation, 
the dimension of a projected subspace of $W$ is written as 
$|P^0_\gamma W| = {\rm rank\,} w[v_{\gamma}, N]$ 
for a subnetwork $\gamma = (v_\gamma, e_\gamma)$. 

Let us pick two subnetworks $\gamma_1$ and $\gamma_2$, 
and denote the sets of chemical species 
by $v_1$ and $v_2$, respectively. 
We consider the submatrix $w [v_1 \cap v_2, N]$. 
We can pick a row basis as 
$\{ \bm a_i^T \,|\, i \in \alpha_{1 \cap 2} \}$, 
where $\alpha_{1 \cap 2} \subset v_1 \cap v_2$.
Here, 
$|\alpha_{1\cap 2}| = {\rm rank\,} w[v_1 \cap v_2, N]$. 
We can form a row basis of $w[v_1, N]$
by adding row vectors from $w[v_1 \setminus v_2, N]$
to $\alpha_{1\cap 2}$. 
Let $\alpha_1$ be the picked indices, 
then 
\begin{equation}
 |\alpha_{1\cap 2}| + |\alpha_1| =
 {\rm rank\, } w[v_1, N] . 
\end{equation}
We can further pick row vectors from 
$w[v_2 \setminus v_1, N]$ 
and form a basis of $w[v_1 \cup v_2 , N]$. 
Let us denote the added indices as $\alpha_2$, then 
\begin{equation}
  |\alpha_{1\cap 2}|
  + 
  |\alpha_{1}|
  + 
  |\alpha_{2}| 
  = {\rm rank\, } w[v_1 \cup v_2, N]. 
\end{equation}
Since the vectors specified by the indices 
$\alpha_{1\cap 2} \cup \alpha_2$ 
are linearly independent and $\alpha_{1\cap 2} \cup \alpha_2 \subset v_2$,
we have 
$
|\alpha_{1\cap 2} |+ | \alpha_2 | 
\le
{\rm rank\,} w[v_2, N]
$. 
This can be written as 
\begin{equation}
{\rm rank\,} w[v_1\cup v_2, N]
\le 
{\rm rank\,} w[v_1, N]
+ {\rm rank\,} w[v_2, N]
- {\rm rank\,} w[v_1\cap v_2, N] . 
\end{equation}
This is equivalent to Eq.~(\ref{eq:p-w-submod}). 
\begin{flushright}
 $\Box$  
\end{flushright}

\begin{corollary}\label{cor:intersection-union}
Let $\Gamma$ be a regular chemical reaction network. 
The union and the intersection of two buffering structures inside $\Gamma$ 
are also buffering structures. 
\end{corollary} 
\begin{proof}
Suppose $\gamma_1$ and $\gamma_2$ are buffering structures inside $\Gamma$. 
Then $\lambda(\gamma_1)=\lambda(\gamma_2)=0$. 
From the submodularity of the influence index, 
we have 
\begin{equation}
  \lambda(\gamma_1 \cup \gamma_2) 
  + 
  \lambda(\gamma_1 \cap \gamma_2) 
  \le 0.
\end{equation} 
Since influence indices are nonnegative for a regular chemical reaction network, 
we have 
$
\lambda(\gamma_1 \cup \gamma_2) 
= 
\lambda(\gamma_1 \cap \gamma_2) 
=0
$. Thus, we obtain the claim. 
\end{proof}

\section{Reduction of chemical reaction networks}\label{sec:reduction}

Generically, a reduction is a process to reduce the number of degrees of freedom 
while keeping some features of the original system. 
Let us here introduce a reduction method which 
consists of the following two steps, 
\begin{itemize}
    \item[(1)] Identify a subnetwork to be reduced. 
    \item[(2)] Perform the reduction for given a subnetwork. 
\end{itemize}
As a result, we obtain a new reaction network 
with fewer chemical species and reactions, 
\begin{equation}
  \Gamma \longrightarrow \Gamma' . 
\end{equation}
The reduced network is 
characterized by a new stoichiometric matrix, 
\begin{equation}
  S \longrightarrow S'.  
\end{equation}
Crucial points are, how to identify a subnetwork 
to be eliminated, 
and how to obtain the new stoichiometric matrix $S'$, 
which determines the structure of the reduced network. 
In this section, we mainly discuss step (2). 
We will discuss more on the choice of a subnetwork in 
Sec.~\eqref{sec:red-buff}.

\subsection{Reduction procedure}

Let us here illustrate 
a method of reduction based on the network topology. 
We denote the whole reaction network by $\Gamma = (V, E)$,
where 
$V$ and $E$ are the sets 
of chemical species and reactions, respectively. 
We choose a subnetwork $\gamma = (V_\gamma, E_\gamma)$,
where $V_\gamma \subset V$ and $E_\gamma \subset E$, 
and eliminate the degrees of freedom inside $\gamma$. 
We refer to the chemical species 
and reactions inside $\gamma$ 
as {\it internal}, and 
those in $\Gamma \setminus \gamma$ as {\it boundary}. 
For the given subnetwork $\gamma$, 
we separate the chemical concentrations and reaction rates as 
\begin{equation}
    \bm x = 
    \begin{pmatrix}
      \bm x_1 \\
      \bm x_2
    \end{pmatrix}, 
    \quad 
    \bm r = 
    \begin{pmatrix}
      \bm r_1 \\
      \bm r_2
    \end{pmatrix}, 
\end{equation}
where $1$ and $2$ correspond to 
internal and boundary degrees of freedom, respectively. 
Accordingly, the stoichiometric matrix $S$ 
can be partitioned as 
\begin{equation}
  S = 
\begin{pmatrix}
  S_{11} & S_{12} \\
  S_{21} & S_{22} 
\end{pmatrix}. 
\label{eq:s-sepa}
\end{equation}
Note that the submatrix $S_{11}$ is the same matrix as $S_\gamma$ 
that appeared in Sec.~\ref{sec:sub-n}. Hereafter we use $S_{11}$ for notational convenience. 
With the separation of internal and boundary 
degrees of freedom, 
the rate equations of the whole reaction system is written as 
\begin{equation}
  \frac{d}{dt}
  \begin{pmatrix}
    \bm x_1 \\
    \bm x_2
  \end{pmatrix}
  = 
  \begin{pmatrix}
    S_{11} & S_{12} \\
    S_{21} & S_{22} 
  \end{pmatrix} 
  \begin{pmatrix}
    \bm r_1 \\
    \bm r_2
  \end{pmatrix}
  = 
  \begin{pmatrix}
    S_{11} \bm r_1 + S_{12} \bm r_2 \\    
    S_{21} \bm r_1 + S_{22} \bm r_2 
  \end{pmatrix}.  
  \label{eq:rate-eq-block}
\end{equation}
While 
the internal reaction rates 
$\bm r_1 = \bm r_1 (\bm x_1, \bm x_2)$ 
in general depend on both of the internal and boundary 
chemical concentrations, 
when $\gamma$ is chosen to be output-complete, 
the boundary reaction rates $\bm r_2 = \bm r_2 (\bm x_2)$ 
do not depend on the internal chemical concentrations $\bm x_1$. 
The first equation 
of Eq.~(\ref{eq:rate-eq-block}) can be solved for 
$\bm r_1$ as 
\begin{equation}
  \bm r_1 = S_{11}^+ \frac{d}{dt} \bm x_1- S_{11}^+ S_{12} \bm r_2 + \bm c_{11}, 
\end{equation}
where $S_{11}^+$ is the Moore-Penrose inverse of $S_{11}$, 
and $\bm c_{11} \in \ker S_{11}$. 
Substituting this to the second equation of Eq.~(\ref{eq:rate-eq-block}), we get 
\begin{equation}
  \frac{d}{dt}
  \left(
 {\bm x}_2 - S_{21} S^+_{11} {\bm x}_1
 \right)
 = 
(S_{22} - S_{21} S_{11}^+ S_{12} ) \bm r_2 
 +
S_{21} \bm c_{11} . 
 \label{eq:rate-red-1}
\end{equation}
When the following condition is satisfied, 
\begin{equation}
  \ker S_{11} \subset \ker S_{21}  , 
  \label{eq:s11-s21-cond}
\end{equation}
$S_{21} \bm c_{11}=\bm 0$ 
and the second term of the RHS of Eq.~(\ref{eq:rate-red-1}) vanishes.\footnote{
As we discuss later, this condition is 
the same as the absence of emergent cycles in $\gamma$. 
See the text around Eq.~(\ref{eq:sc1-eq-0}). 
} 
Then, the rate equation is written as 
\begin{equation}
  \frac{d}{dt}
  \left(
 {\bm x}_2 - S_{21} S^+_{11} {\bm x}_1
 \right)
 = 
S' \bm r_2 , 
\end{equation}
where $S'$ is the generalized Schur complement, 
\begin{eqnarray}
S' = S / S_{11} \coloneqq S_{22}  - S_{21} S_{11}^+ S_{12}. 
\label{eq:def-sp} 
\end{eqnarray}
As long as steady states are concerned, 
the subnetwork $(\bm x_2, \bm r_2)$ 
satisfies the rate equation 
whose stoichiometric matrix is $S'$. 
This motivates us to consider the subnetwork 
$(\bm x_2, \bm r_2)$
whose rate equation is given by 
\begin{equation}
  \frac{d}{dt} {\bm x}_2 
   = S' \bm r_2 (\bm x_2).  
\end{equation}  
Based on the considerations above, 
we define the reduction 
of a reaction system in the following way: 
\begin{definition}[Reduction]
Let $\Gamma = (V, E)$ be a chemical reaction network 
with stoichiometric matrix $S$ 
and 
$\gamma =(V_\gamma, E_\gamma)$ be an output-complete subnetwork
whose stoichiometric matrix is denoted by $S_{11}$. 
We define a reduced network 
$\Gamma' = (V \setminus V_\gamma, E \setminus E_\gamma)$ 
obtained by eliminating $\gamma$ from $\Gamma$, 
by a stoichiometric matrix $S'$
given by the generalized Schur complement~\eqref{eq:def-sp}. 
%
We denote the resultant 
reaction network by $\Gamma' = \Gamma / \gamma$. 
Accordingly, 
the chemical concentrations 
and reaction rates of the reduced system $(\bm x', \bm r')$ 
are obtained from the original ones $(\bm x, \bm r)$ as 
\begin{align}
  \bm x = 
  \begin{pmatrix}
    \bm x_1 \\
    \bm x_2
  \end{pmatrix}
  &\longrightarrow
  \bm x' = \bm x_2 ,\\
  \bm r (\bm x) = 
  \begin{pmatrix}
    \bm r_1 (\bm x_1, \bm x_2)\\
    \bm r_2 (\bm x_2)
  \end{pmatrix}
  &\longrightarrow
  \bm r' (\bm x') = \bm r_2 (\bm x_2), 
\end{align}
and the rate equation of the reduced system is given by 
\begin{equation}
  \frac{d}{dt} {\bm x}' 
   = S' \bm r' (\bm x'). 
\end{equation}  
\end{definition}

\begin{remark}
The structure of the reduced network is determined 
by the generalized Schur complement~\eqref{eq:def-sp} of the stoichiometric matrix. 
The second term in Eq.~\eqref{eq:def-sp} 
represents the rewiring of the network 
associated with the elimination of $\gamma$. 
\end{remark}

\begin{remark}
The reduced system can be always defined 
if $\gamma$ is output-complete; otherwise, the reduction is ill-defined since the reduced system would depend on $\bm x_1$ through $\bm r_2$. 
We emphasize that the output-completeness is a topological condition determined by the stoichiometry 
and the details of the reactions, namely the kinetics, are irrelevant. 
Thus, the reduction is applicable to any kind of kinetics. 
How the reduced system is related to the original system 
depends on further nature of $\gamma$. 
In the following sections, we will discuss more on 
the features of the subnetworks that behave nicely 
under reductions. 
In Sec.~\ref{sec:red-buff-th}, 
we prove that, when 
$\gamma$ has a vanishing influence index (see Sec.~\ref{sec:lol}), 
which is determined by the network topology, 
the steady state of the reduced system 
is assured to be the same as the steady state of the original system. 
\end{remark}

\begin{remark}
In Sec.~\ref{sec:red-mor}, 
we show that the reduction we introduced here can be regarded as a 
morphism of reaction networks 
that `shrinks' a subnetwork to a point, followed by the removal of degenerate (chemically meaningless) reactions. 
\end{remark}

\begin{remark}
We note that the elements of the matrix $S'$ are rational, 
since the Moore-Penrose inverse of an integral matrix is rational \cite{KOECHER1985187}. 
The matrix $S'$ can be always transformed into an integral matrix 
by columnwise rescaling of $S'$ together with the rescaling of reaction rates. 
\end{remark}

\begin{remark}
The stoichiometric matrix given by the generalized 
Schur complement has appeared previously 
in flux balance analysis \cite{van1994linear,klamt2002calculability}. 
The current method is different from the ones discussed for reaction networks with the mass-action kinetics 
in detailed balanced \cite{van2013mathematical} 
and complex balanced \cite{rao2013graph} situations, 
where the Schur complementation is performed for the weighted Laplacian similarly to the Kron reduction of electrical circuits \cite{kron1939tensorbook,8347206,6316101}. 
In the current formulation, the Schur complementation is performed for the stoichiometric matrix. 
\end{remark}

\subsection{Simple examples of reduction}

To illustrate the reduction procedure, 
here we discuss simple examples. 
In Sec.~\ref{sec:ex-ecoli}, 
we discuss the reduction of the metabolic pathway of {\it E.~coli} as a more realistic example. 

\begin{example}\label{ex:red-v2e3}
We consider a monomolecular reaction network that consists of 
$(V,E)=(\{v_1,v_2\}, \{ e_1,e_2,e_2\})$. 
We take a subnetwork $\gamma = (\{v_1 \}, \{ e_2 \})$ 
to be reduced. 
Under the reduction, the stoichiometric matrix changes as 
\begin{equation}
  \begin{tikzpicture}
    \node at (0, 0) {
      $
      S =
      \begin{blockarray}{cccc}
        && \\
        \begin{block}{c(ccc)}
         {\color{mydarkred}v_1} \quad\, &-1&1& 0\\
         v_2\quad\, &1& 0 &-1 \\
        \end{block} 
         & {\color{mydarkred}e_2} & e_1 & e_3 
      \end{blockarray}
          $
      }; 
      \draw[mydarkred, dashed,line width=1] 
      (-0.05, 0.2) rectangle (0.5,0.7);
      \node at (0.1, 1) {\color{mydarkred}$S_{11}$} ;

   \draw [mybrightblue,ultra thick,-latex]   (2.3,0) -- (3.3, 0); 

    \node at (5, 0) {
      $
      S' =
      \begin{blockarray}{ccc}
        && \\
        \begin{block}{c(cc)}
         v_2 \, \, \, & 1 &-1 \\
        \end{block} 
         & e_1 & e_3 
      \end{blockarray}
          $
      }; 

\end{tikzpicture} 
\end{equation}
where we have brought the reduced part to the upper-left part. 
The reduction looks like 
\begin{equation}
  \begin{tikzpicture}

    \node[species] (x) at (0,0) {$v_1$}; 
    \node[species] (y) at (1.5,0) {$v_2$}; 
 
    \node (D1) at (-1.25,0) {}; 
    \node (D2) at (2.75,0) {}; 
    
   \draw [-latex,line width=0.5mm] (D1) edge node[below]{$e_1$} (x);
   \draw [-latex,line width=0.5mm] (x) edge node[below]{$e_2$} (y);
   \draw [-latex,line width=0.5mm] (y) edge node[below]{$e_3$} (D2);

   \draw[mydarkred, dashed,line width=1] (-0.3,-1) rectangle (1.2,1);
   \node at (0, 0.7) {  \scalebox{1.1} {\color{mydarkred} $\gamma$} };
   
 \draw [mybrightblue,ultra thick,-latex]  (3.25,0) -- (4.25,0); 

  \node[species] (y2) at (6,0) {$v_2$};
  
  \node (D12) at (4.75,0) {};
  \node (D22) at (7.25,0) {};  
  
   \draw [-latex,line width=0.5mm] (D12) edge node[below]{$e_1$} (y2);
   \draw [-latex,line width=0.5mm] (y2) edge node[below]{$e_3$} (D22);

\end{tikzpicture}
\end{equation}
The original rate equation is 
\begin{equation}
  \frac{d}{dt}
  \begin{pmatrix}
    x_1 \\
    x_2 
  \end{pmatrix}
  = 
  \begin{pmatrix}
    -1 & 1&0 \\
    1 & 0 &-1 \\
  \end{pmatrix}
  \begin{pmatrix}
    r_2(x_1) \\
    r_1 \\
    r_3 (x_2)
  \end{pmatrix}, 
\end{equation}
where $x_1 = x (v_1)$, $r_2 = r (e_2)$ and so on. 
If we eliminate $r_2 (x_1)$, 
\begin{equation}
  \frac{d}{dt}
\left(  x_2  + x_1   \right)
  = 
  \begin{pmatrix}
    1 & -1 
  \end{pmatrix} 
 \begin{pmatrix}
   r_1 \\
   r_3 (x_2) 
 \end{pmatrix}. 
\end{equation}
The reduced equation of motion is obtained by replacing 
$x_2  + x_1$ with $x_2$ on the left-hand side. 

To compute the steady-state solutions, 
let us for example employ the mass-action kinetics, 
\begin{equation}
\begin{pmatrix}
      r_2 (x_1)\\
      r_1 \\
      r_3(x_2)\\
    \end{pmatrix}
    = 
    \begin{pmatrix}
      k_2 x_1 \\
      k_1 \\
      k_3 x_2 \\
\end{pmatrix}. 
\end{equation}
The steady-state reaction rates and concentrations are given by 
\begin{equation}
    \begin{pmatrix}
      \bar r_2 \\
      \bar r_1 \\
      \bar r_3\\
\end{pmatrix}
= 
k_1 
\begin{pmatrix}
      1 \\
      1 \\ 
      1 
\end{pmatrix}, 
\quad 
\begin{pmatrix}
  \bar x_1 \\
  \bar x_2 
\end{pmatrix}
= 
\begin{pmatrix}
  k_1 / k_2 \\
  k_1 / k_3 
\end{pmatrix}.
\label{eq:ex-sol-orig}
\end{equation}
The steady-state solutions of the reduced system are 
\begin{equation}
    \begin{pmatrix}
      \bar r_1 \\
      \bar r_3 
\end{pmatrix}
= 
k_1 
\begin{pmatrix}
      1 \\
      1 
\end{pmatrix}, 
\quad 
\bar x_2 = \frac{k_1}{k_3}.
\end{equation}
Note that this solution of the reduced system 
is exactly the same as the solution (\ref{eq:ex-sol-orig}) of the original system for the boundary concentrations and rates. 
Indeed, this is a special property of buffering structures. 
In this example, the subnetwork has a vanishing influence index, $\lambda(\gamma) = -1+1-0+0=0$, and hence is a buffering structure. 
Generically, 
when the reduced subnetwork is a buffering structure, 
the steady-state solution of the reduced system 
is the same as the original system,
and this is the content of Theorem~\ref{th:reduction}. 
Although we used the mass-action kinetics in this example, 
the theorem applies to any kind of kinetics. 
We give a proof of the theorem in Sec.~\ref{sec:red-buff-th}.

\end{example}

\begin{example} \label{ex:v2e3}
$(V,E) =(\{v_1,v_2,v_3\}, \{ e_1,e_2,e_3\})$. 
The stoichiometric matrices of the original and reduced system 
are given by 
\begin{center}
  \begin{tikzpicture}
    \node at (0, 0) {
      $
      S =
      \begin{blockarray}{cccc}
        && \\
        \begin{block}{c(ccc)}
         {\color{mydarkred}v_1} \quad\, &-1&0 &1\\
         v_2\quad\, &1&-1 &0\\
         v_3\quad\, &0&1 &-1\\
        \end{block} 
         & {\color{mydarkred}e_1} & e_2 & e_3 
      \end{blockarray}
      $
      }; 
      \draw[mydarkred, dashed,line width=1] 
      (-0.2,0.4) rectangle (0.5,0.9);

   \draw [mybrightblue,ultra thick,-latex]   (2.3,0) -- (3.3, 0); 

    \node at (5.3, 0) {
      $
      S' =
      \begin{blockarray}{ccc}
        && \\
        \begin{block}{c(cc)}
         v_2\quad\, &-1 &1\\
         v_3\quad\, &1 &-1\\
        \end{block} 
         & e_2 & e_3 
      \end{blockarray}
      $
      };
      
      \node at (0.2, 1.3) {\color{mydarkred}$S_{11}$} ;
  \end{tikzpicture}    
\end{center}
where we reduced the subnetwork $\gamma = (\{ v_1\}, \{ e_1\})$. 
The reduction is visually expressed as 
\begin{equation}
  \begin{tikzpicture}

    \node[species] (v1) at (0,0) {$v_1$}; 
    \node[species] (v2) at (1,1.73) {$v_2$}; 
    \node[species] (v3) at (2,0) {$v_3$}; 
  
    \draw [-latex,line width=0.5mm] (v1) edge node[left]{$e_1$} (v2);
    \draw [-latex,line width=0.5mm] (v2) edge node[right]{$e_2$} (v3);
    \draw [-latex,line width=0.5mm] (v3) edge node[below]{$e_3$} (v1); 

   \node at (-0.5, 1) {  \scalebox{1.1} {\color{mydarkred} $\gamma$} };
   
   \draw[rotate around={-30:(0,0)}, color=mydarkred, dashed, line width=1]
    (-0.12,0.43) ellipse (0.5 and 1.2); 
   
 \draw [mybrightblue,ultra thick,-latex]   (3,0.8) -- (4,0.8); 

    \node[species] (rv2) at (5,1.73) {$v_2$}; 
    \node[species] (rv3) at (6,0) {$v_3$}; 

    \draw [-latex,line width=0.5mm,out=-30,in=90] 
    (rv2) edge node[right]{$e_2$} (rv3); 

    \draw [-latex,line width=0.5mm,out=150,in=-90] 
    (rv3) edge node[left]{$e_3$} (rv2); 

\end{tikzpicture}
\end{equation}
Suppose that we take the mass-action kinetics, 
\begin{equation}
\begin{pmatrix}
      r_1(x_1) \\
      r_2(x_2) \\
      r_3(x_3)\\
    \end{pmatrix}
    = 
    \begin{pmatrix}
      k_1 x_1 \\
      k_2 x_2 \\
      k_3 x_3 \\
\end{pmatrix}. 
\end{equation}
The system has one conserved charge and we specify the value as 
$\ell = x_1 + x_2 + x_3.$ 
The steady-state reaction rates of the original system are 
\begin{equation}
    \begin{pmatrix}
      \bar r_1 \\
      \bar r_2 \\
      \bar r_3 
    \end{pmatrix}
    = \ell K 
    \begin{pmatrix} 1 \\
    1\\
    1
    \end{pmatrix} , 
\end{equation}
where $K$ is defined by 
$
\frac{1}{K} \coloneqq \frac{1}{k_1} + \frac{1}{k_2} + \frac{1}{k_3}  . 
$
In the reduced system, $\ell' = x_2 + x_3$ is a conserved charge.
The steady-state rates in the reduced system are 
\begin{equation}
    \begin{pmatrix}
      \bar r_2 \\
      \bar r_3 
   \end{pmatrix}
    = \ell' K' 
    \begin{pmatrix} 
    1 \\
    1
    \end{pmatrix} , 
\end{equation}
where 
$
    \frac{1}{K'} \coloneqq \frac{1}{k_2} + \frac{1}{k_3}  . 
$
Note that, if we want to have the same steady-state in the reduced system 
as the one in the original system, 
we have to choose the parameters so that $\ell K = \ell' K'$. 
This is in contrast to Example~\ref{ex:red-v2e3}, 
where no fine-tuning of the parameters is needed. 
The difference is attributed to the fact that 
the subnetwork $\gamma$ is not a buffering structure 
and the index is nonzero, $\lambda(\gamma) = -1 + 1 - 0 + 1 = 1$.

\end{example}

\begin{example}
$(V,E)=(\{v_1,v_2,v_3,v_4\}, \{ e_1,e_2,e_3,e_4,e_5,e_6\})$. 
The stoichiometric matrix changes under reduction as 
\begin{equation}
  \begin{tikzpicture}
    \node at (0, 0) {
      $
      S =
      \begin{blockarray}{ccccccc}
        && \\
        \begin{block}{c(cccccc)}
     {\color{mydarkred}v_1} \quad\, &     -1 & 0 & 1 & 0 & 1 & 0 \\
 {\color{mydarkred}v_2} \quad\, & 1 & -1 & 0 & 0 & 0 & 0 \\
v_3\quad\, & 0 & 1 & 0 & -1 & 0 & -1 \\
 v_4\quad\, &0 & 0 & 0 & 1 & -1 & 0 \\
        \end{block} 
         & {\color{mydarkred}e_2} & {\color{mydarkred}e_3}& e_1 &e_4&e_5&e_6
      \end{blockarray}  
          $
      }; 
      \draw[mydarkred, dashed,line width=1] 
      (-1,0.2) rectangle (0.3,1);
      \node at (-0.3, 1.25) {\color{mydarkred}$S_{11}$} ;

  \draw [mybrightblue,ultra thick,-latex]   (3,0) -- (4, 0); 

  \node at (6.5, 0) {
      $
      S' =
      \begin{blockarray}{ccccc}
        && \\
        \begin{block}{c(cccc)}
v_3\quad\, & 1 & -1 & 1 & -1\\
 v_4\quad\, &0 & 1 & -1 & 0 \\
        \end{block} 
         &  e_1 &e_4&e_5&e_6
      \end{blockarray}
          $ 
      }; 

  \end{tikzpicture} 
\end{equation}
where we chose the subnetwork $\gamma = (\{v_1,v_2\},\{e_2,e_3\})$ to be reduced. 
The reduction is visually expressed as 
\begin{equation}\begin{tikzpicture} 
    \node[species] (v1) at (1.25,0) {$v_1$}; 
    \node[species] (v2) at (2.5,0) {$v_2$};
    \node[species] (v3) at (3.75,0) {$v_3$};
    \node[species] (v4) at (2.5,1.5) {$v_4$}; 
    \node (d1) at (0,0) {}; 
    \node (d2) at (5,0) {}; 

    \node at (1.2, -0.7) {  \scalebox{1.1} {\color{mydarkred} $\gamma$} };
 
    \draw [-latex,draw,line width=0.5mm] (d1) edge node[below]{$e_1$} (v1);
    \draw [-latex,draw,line width=0.5mm] (v1) edge node[below]{$e_2$} (v2);
    \draw [-latex,draw,line width=0.5mm] (v2) edge node[below]{$e_3$} (v3);
    \draw [-latex,draw,line width=0.5mm] (v3) edge node[above right]{$e_4$} (v4);
    \draw [-latex,draw,line width=0.5mm] (v4) edge node[above left]{$e_5$} (v1);
    \draw [-latex,draw,line width=0.5mm] (v3) edge node[below]{$e_6$} (d2);

   \draw[mydarkred, dashed,line width=1] (0.9,-1.) rectangle (3.45,0.3);
       
   \draw [mybrightblue,ultra thick,-latex]
   (5.5,0) -- (6.5,0); 
   
    \node[species] (red_v4) at (8.25,1.5) {$v_4$}; 
    \node[species] (red_v3) at (8.25,0) {$v_3$};
    \node (red_d1) at (7,0) {}; 
    \node (red_d2) at (9.5,0) {}; 

    \draw [-latex,line width=0.5mm] (red_d1) edge node[below]{$e_1$} (red_v3);
    \draw [-latex,line width=0.5mm,out=60,in=-60] (red_v3) edge node[right]{$e_4$} (red_v4);
    \draw [-latex,line width=0.5mm,out=240,in=120] (red_v4) edge node[left]{$e_5$} (red_v3);
    \draw [-latex,line width=0.5mm] (red_v3) edge node[below]{$e_6$} (red_d2);

\end{tikzpicture}
\end{equation}
The subnetwork is a buffering structure, $\lambda(\gamma) = -2 + 2 - 0 + 0 = 0$. 

\end{example}

\begin{example}
$(V,E)=(\{v_1,\ldots, v_9\}, \{ e_1,\ldots, e_{13}\})$
with the following stoichiometric matrix, 
\begin{center}
  \begin{tikzpicture}
    \node at (0, 0) {
      $
      S =
      \begin{blockarray}{ccccccccccccccc}
        && \\
        \begin{block}{c(cccccccccccccc)}
   {\color{mydarkred}v_3} \quad\, &  -1 & -1 & 0 & 0 & 0 & 1 & 0 & 0 & 0 & 1 & 0 & 0 & 0 & 0\\
  {\color{mydarkred}v_4} \quad\, &   1 & 0 & -1 & 0 & 0 & 0 & 0 & 0 & 0 & 0 & 1 & 0 & 0 & 0 \\
  {\color{mydarkred}v_5} \quad\, & 0 & 1 & 0 & 1 & -1 & 0 & 0 & 0 & 0 & 0 & 0 & 0 & 0 & 0 \\
  {\color{mydarkred}v_7} \quad\, &  0 & 0 & 0 & 0 & 1 & -1 & -1 & 0 & 0 & 0 & 0 & 0 & 0 & 0 \\

  {\color{mydarkred}v_8} \quad\, & 0 & 0 & 0 & 0 & 0 & 0 & -1 & 0 & 0 & 0 & 0 & 1 & 0 & 0 \\
v_1\quad\, &  0 & 0 & 0 & 0 & 0 & 0 & 0 & 1 & 0 & -1 & 0 & 0 & 0 & 0 \\
v_2\quad\, & 0 & 0 & 0 & 0 & 0 & 0 & 0 & 0 & 1 & 0 & -1 & 0 & 0 & 0 \\

v_6\quad\, & 0 & 0 & 1 & -1 & 0 & 0 & 0 & 0 & 0 & 0 & 0 & -1 & -1 & 0 \\

v_9\quad\, & 0 & 0 & 0 & 0 & 0 & 0 & 1 & 0 & 0 & 0 & 0 & 0 & 0 & -1 \\
        \end{block} 
         & {\color{mydarkred}e_5} & {\color{mydarkred}e_6}&  {\color{mydarkred}e_7} & {\color{mydarkred}e_8}& {\color{mydarkred}e_9}& {\color{mydarkred}e_{11}}&  {\color{mydarkred}e_{12}} & e_1 & e_2  & e_3 & e_4 & e_{10} &  e_{13} & e_{14}
      \end{blockarray} 
          $
      }; 
      \draw[mydarkred, dashed,line width=1] 
      (-3.25,0) rectangle (0.9,2);
      \node at (-1, 2.3) {\color{mydarkred}$S_{11}$} ;
  \end{tikzpicture}  .
\end{center}
We choose the subnetwork $\gamma = (\{v_3,v_4,v_5,v_7,v_8 \},\{e_5,e_6,e_7,e_8,e_9,e_{11},e_{12} \})$ to be reduced. The reduced subnetwork is given by 
\begin{center}
  \begin{tikzpicture}
    \node at (0, 0) {
      $
      S' =
      \begin{blockarray}{cccccccc}
        && \\
        \begin{block}{c(ccccccc)}
v_1\quad\, &1 & 0 & -1 & 0 & 0 & 0 & 0 \\
v_2\quad\,&0 & 1 & 0 & -1 & 0 & 0 & 0 \\
v_6\quad\, & 0 & 0 & 1 & 1 & -2 & -1 & 0 \\
v_9\quad\, &  0 & 0 & 0 & 0 & 1 & 0 & -1 \\
        \end{block} 
         &e_1 & e_2  & e_3 & e_4 & e_{10}& e_{13} & e_{14}
      \end{blockarray} 
          $
      }; 
     \end{tikzpicture}.  
\end{center}
The subnetwork $\gamma$ is a buffering structure: $\lambda(\gamma) = -5 + 7 - 2 + 0 = 0$. 
The reduction is visually expressed as
\begin{equation}\begin{tikzpicture} 

    \node[species] (v1) at (0,0) {$v_1$}; 
    \node[species] (v2) at (2,0) {$v_2$};
    \node[species] (v3) at (0,-1.25) {$v_3$};
    \node[species] (v4) at (2,-1.25) {$v_4$}; 
    \node[species] (v5) at (0,-2.5) {$v_5$}; 
    \node[species] (v6) at (2,-2.5) {$v_6$};
    \node[species] (v7) at (0,-3.75) {$v_7$};
    \node[species] (v8) at (2,-3.75) {$v_8$}; 
    \node[reaction] (e12) at (1,-4.5) {$e_{12}$}; 
    \node[species] (v9) at (1,-5.5) {$v_9$}; 
    
    \node (d1) at (0,1.125) {};
    \node (d2) at (2,1.125) {};
    \node (d3) at (3.125,-2.5) {};
    \node (d4) at (1, -6.625) {};

    \draw [-latex,line width=0.5mm] (d1) edge node[left]{$e_1$} (v1); 
    \draw [-latex,line width=0.5mm] (d2) edge node[left]{$e_2$} (v2); 
    \draw [-latex,line width=0.5mm] (v1) edge node[left]{$e_3$} (v3); 
    \draw [-latex,line width=0.5mm] (v2) edge node[left]{$e_4$} (v4); 
    \draw [-latex,line width=0.5mm] (v3) edge node[below]{$e_5$} (v4); 
    \draw [-latex,line width=0.5mm] (v3) edge node[left]{$e_6$} (v5); 
    \draw [-latex,line width=0.5mm] (v4) edge node[left]{$e_7$} (v6); 
    \draw [-latex,line width=0.5mm] (v6) edge node[below]{$e_8$} (v5); 
    \draw [-latex,line width=0.5mm] (v5) edge node[left]{$e_9$} (v7); 
    \draw [-latex,line width=0.5mm] (v6) edge node[right]{$e_{10}$} (v8); 
    \draw [-latex,line width=0.5mm,out=180,in=180] (v7) edge node[left]{$e_{11}$} (v3); 
    \draw [-latex,line width=0.5mm] (v6) edge node[above]{$e_{13}$} (d3); 
    \draw [-latex,line width=0.5mm] (v9) edge node[left]{$e_{14}$} (d4); 
    \draw [-latex,line width=0.5mm] (v7) -- (e12); 
    \draw [-latex,line width=0.5mm] (v8) -- (e12); 
    \draw [-latex,line width=0.5mm] (e12) -- (v9);

   \draw [mybrightblue,ultra thick,-latex]  (3.2,-3) -- (4.2,-3); 
  
   \draw [mydarkred,dashed,line width=1]
  (-2,-5.2) -- (2.3,-5.2); 
   \draw [mydarkred,dashed,line width=1]
  (-2,-0.9) -- (2.3,-0.9); 
    \draw [mydarkred,dashed,line width=1]
  (-2,-0.9) -- (-2,-5.2); 
   \draw [mydarkred,dashed,line width=1]
  (2.3,-0.9) -- (2.3,-2.2); 
    \draw [mydarkred,dashed,line width=1]
  (1.6,-2.2) -- (2.3,-2.2); 
   \draw [mydarkred,dashed,line width=1]
  (1.6,-2.2) -- (1.6,-3.4); 
   \draw [mydarkred,dashed,line width=1]
  (2.3,-3.4) -- (1.6,-3.4); 
   \draw [mydarkred,dashed,line width=1]
  (2.3,-3.4) -- (2.3,-5.2);

    \node at (-1.5, -4.5) {  \scalebox{1.1} {\color{mydarkred} $\gamma$} };
  
    \node[species] (rv1) at (5,0) {$v_1$}; 
    \node[species] (rv2) at (7,0) {$v_2$};
    \node[species] (rv6) at (7,-2.5) {$v_6$};
    \node[species] (rv9) at (6,-5.5) {$v_9$}; 

    \node (rd1) at (5,1.125) {};
    \node (rd2) at (7,1.125) {};
    \node (rd3) at (8.125,-2.5) {};
    \node (rd4) at (6, -6.625) {};

    \draw [-latex,line width=0.5mm] (rd1) edge node[left]{$e_1$} (rv1); 
    \draw [-latex,line width=0.5mm] (rd2) edge node[left]{$e_2$} (rv2); 
    \draw [-latex,line width=0.5mm] (rv1) edge node[left]{$e_3$} (rv6); 
    \draw [-latex,line width=0.5mm] (rv2) edge node[left]{$e_4$} (rv6); 
    \draw [-latex,line width=0.5mm] (rv6) edge node[left]{$e_{10}$} (rv9); 
    \draw [-latex,line width=0.5mm] (rv9) edge node[left]{$e_{14}$} (rd4); 
    \draw [-latex,line width=0.5mm] (rv6) edge node[above]{$e_{13}$} (rd3); 

\end{tikzpicture}
\end{equation}
 We note that, under the reduction, the stoichiometries for reactions $e_3,\ e_4, \ e_{10}$ are changed from the original ones. In particular, 
 $e_{10}$, which is originally  monomolecular,  becomes non-monomolecular, $2 v_6 \rightarrow v_9$. To reproduce steady states of the original system, the rate $r_{10}$ is required to be the same as before the reduction; for example, in the mass-action kinetics, $r_{10}$ is given by $r_{10}(x_6)=k_{10} x_6$ rather than by $r_{10}(x_6)=k_{10} x_6^2$,  even after the reduction. 

\end{example}

\subsection{Reduction as a morphism of chemical reaction networks} \label{sec:red-mor}

The structure of a reduced network is characterized 
by the generalized Schur complement \eqref{eq:def-sp}. 
Here, let us show that this form arises if we consider 
a map between chemical reactions that shrink a subnetwork to a point. 
The morphisms of chemical reaction networks have been discussed, for example, in Ref.~\cite{cardelli2014morphisms}. 
Let us prepare some terminologies. 
\begin{definition}(Degenerate reactions)
  A reaction $e \in E$ is said to be 
  {\it degenerate in stoichiometry} 
  if $s(e)(v_i) = t(e)(v_i)$ for any $v_i \in V$. 
\end{definition}
A degenerate reaction is a trivial reaction since it does not change anything, and the removal of degenerate reactions 
does not affect the chemical properties of the reaction network. 
A degenerate reaction is represented as a $0-$column in the stoichiometric matrix. 

Let us slightly extend the definition of CRNs
for technical reasons. 
\begin{definition}[Generalized CRNs] 
  A {\it generalized chemical reaction network} $\Gamma$ 
  is a quadruple $\Gamma =(V,E,s,t)$,
  where 
  $V$ is a set of chemical species, 
  $E$ is a set of chemical reactions, 
  and $s$ and $t$ are source and target functions, 
  \begin{equation}
    s: E \to {\mathbb R}^{V}, 
    \quad 
    t: E \to {\mathbb R}^{V}. 
\end{equation}
\end{definition}
Compared with the previous definition of a CRN, 
$\mathbb N$ is replaced with real numbers, $\mathbb R$. 
We also call an element of ${\mathbb R}^V$ 
as a chemical complex. 
In the remainder of this paper, 
we will mean a generalized CRN 
when we write a CRN. 
This extension is needed because the reductions we consider do not necessarily preserve the integrality of the source and target functions. 
However, we note that the integrality 
can be always recovered by reactionwise rescaling, 
if the original $s$ and $t$ functions are valued in integers. 
\begin{definition}[CRN morphisms] 
A {\it CRN morphism} 
$\varphi$ 
from 
$\Gamma = (V, E,s,t)$ 
to 
$\Gamma' = (V', E',s',t')$
is a pair of maps,  
$(\varphi_0, \varphi_1)$, where 
\begin{eqnarray}
  \varphi_0 &:& {\mathbb R}^V \to {\mathbb R}^{V'} ,
  \\
  \varphi_1 &:& E \to E' , 
\end{eqnarray}
which we call a chemical complex map 
and a reaction map, respectively, 
such that the following diagrams commute, 
\begin{equation}
    \xymatrix{ 
      E
      \ar[r]^{ s } 
      \ar[d]^{ \varphi_1 } 
      & 
      {\mathbb R}^V
      \ar[d]^{ \varphi_0 } 
      & 
      E 
      \ar[r]^{t} 
      \ar[d]^{ \varphi_1 } 
      & 
      {\mathbb R}^V 
      \ar[d]^{ \varphi_0 } 
      \\
      E'
      \ar[r]^{ s' }
      & 
      {\mathbb R}^{V'}
            & 
      E' 
      \ar[r]^{t'} 
      & 
      {\mathbb R}^{V'} . 
    }  
\end{equation}  
\end{definition}



We introduce the matrix representation
of a chemical complex map and a reaction map, 
\begin{align}
  \varphi_0 (v_{i}) &=\sum_{i'} (\varphi_0)_{i i'} \, v_{i'} ,  \\
  \varphi_1 (e_{A}) &=\sum_{A'} (\varphi_1)_{A A'} \, e_{A'}. 
\end{align}
On the spaces of chains, 
a CRN morphism induces the following commutative diagram,
\begin{equation}
  \xymatrix{ 
    C_1(\Gamma )
    \ar[d]^{ \varphi_1 } 
    \ar[r]^{ \p_1 } 
    & 
    C_0(\Gamma ) 
    \ar[d]^{ \varphi_0 }  
    \\
    C_1(\Gamma' )
    \ar[r]^{ \p'_1 } 
    & 
    C_0(\Gamma' ) 
    }  
\end{equation}
where $\p'_1$ is the boundary operator on 
$C_1(\Gamma')$. 
Namely, 
\begin{equation}
    \varphi_0 \circ \p_1
      = 
    \p'_1 \circ \varphi_1 . 
    \label{eq:crn-hom-com}
\end{equation}
In terms of the matrix components, 
\begin{equation}
\sum_{i}
  S_{iA}  
  \, 
  (\varphi_0)_{i i' }
  = 
\sum_{A'}
  (\varphi_1)_{A A'} 
  \, 
  S'_{i'A'} . 
\end{equation}
We write this relation in the matrix form, 
\begin{equation}
  \varphi_0^T S = S' \varphi_1^T . 
  \label{eq:commutativity-mat}
\end{equation}

Now we are ready to discuss a morphism that corresponds to the reduction: 
\begin{definition}[Reduction morphisms] 
  We define a {\it reduction morphism} 
  from $\Gamma$ to $\Gamma'$, 
associated with a subnetwork $\gamma \subset \Gamma$, 
as a CRN morphism satisfying the following properties: 
\begin{enumerate} 
  \item The chemical complexes and reactions in 
$\Gamma \setminus \gamma$ are unchanged. 
  \item 
  All the chemical complexes in $\gamma$ are collapsed 
  into one chemical complex $\bar c$ in  $\Gamma \setminus \gamma$, 
  in such a way that image of all the reactions in $\gamma$ are degenerate in stoichiometry. 
\end{enumerate}
\end{definition}
%


Let us here show that a reduction morphism 
gives rise to the reduced stoichiometric matrix given by the generalized Schur complement \eqref{eq:def-sp}. 
We consider the matrix representation of a reduction morphism. 
From property 2 of reduction morphisms, 
the chemical complex map and the reaction map are both identity 
on $v_i \in V_{\Gamma \setminus \gamma}$ and 
$e_A \in E_{\Gamma \setminus \gamma}$, 
\begin{equation}
    \varphi_0 |_{\Gamma \setminus \gamma} = \bm 1, 
    \quad 
    \varphi_1 |_{\Gamma \setminus \gamma} = \bm 1. 
\end{equation}
Furthermore, 
we can always set $\varphi_1 |_{\gamma} = \bm 1$ without affecting the chemical properties, since degenerate reactions do nothing chemically. 
By arranging the rows and columns, 
the species and reaction maps 
of a reduction morphism can be written in the following form, %
\begin{equation}
  \varphi_0 
  = 
  \begin{pmatrix}
    F^T \\ 
    \bm 1 
  \end{pmatrix}, 
\quad 
\varphi_1 = \bm 1 , 
\label{eq:phi-0-1}
\end{equation}
where $F$ is some matrix (see examples later in this section). 
The explicit form of $F$ does not matter here\footnote{ 
In the case of a directed graph (i.e. a monomolecular reaction network), $F$ has one row whose elements are all $1$ and 
other components are all zero, 
\begin{equation}
  F = 
  \begin{pmatrix}
    0 & \cdots & 0 \\
    \vdots & & \vdots \\
    1 & \cdots & 1 \\
    \vdots & & \vdots \\
    0 & \cdots & 0 
  \end{pmatrix}. 
\end{equation}
}. 
By plugging Eq.~\eqref{eq:phi-0-1}
into the commutativity condition \eqref{eq:commutativity-mat}, 
we find that $S'$ is written as 
\begin{equation}
  S' = 
    \begin{pmatrix}
    F & \bm 1 
  \end{pmatrix}
  \begin{pmatrix}
    S_{11} & S_{12} \\
    S_{21} & S_{22}
  \end{pmatrix}
 = 
  \begin{pmatrix}
    F S_{11} + S_{21} & \,\,  F S_{12} + S_{22}
  \end{pmatrix}. 
\end{equation}
From the condition that 
the image of the reactions in $\gamma$ 
under a reduction morphism 
is degenerate reactions, we have 
\begin{equation}
  F S_{11} + S_{21} = \bm 0. 
  \label{eq:f-s11-s21}
\end{equation}
This condition implies 
\begin{equation}
\ker S_{11} \subset \ker S_{21}. 
\label{eq:kers11-in-kers21}
\end{equation}
This is because, if $\bm c \in \ker S_{11}$, 
Eq.~\eqref{eq:f-s11-s21} implies 
$S_{21} \bm c = \bm 0$, we have Eq.~\eqref{eq:kers11-in-kers21}. 
A generic solution to Eq.~\eqref{eq:f-s11-s21} for $F$ 
is written as 
\begin{equation}
  F = - S_{21} S_{11}^+  + D, 
\end{equation}
where $D$ is a matrix satisfying $D S_{11} = 0$. 
The stoichiometric matrix $S'$ can be now written as 
\begin{equation}
  S' =
\begin{pmatrix}    
  \bm 0 & \,\, S_{22} - S_{21} S_{11}^+ S_{12}  + D S_{12} 
\end{pmatrix}. 
\label{eq:sp-0-sc}
\end{equation}
The term $D S_{12}$ vanishes if and only if 
\begin{equation}
  \coker S_{11} \subset \coker S_{12} . 
  \label{cokers11-in-cokers12}
\end{equation}
Note that the combination $S_{21} S_{11}^+ S_{12}$ does not depend on the choice of the pseudo-inverse, as long as 
Eqs.~\eqref{eq:kers11-in-kers21} and \eqref{cokers11-in-cokers12} are satisfied. 
After removing degenerate reactions
from Eq.~\eqref{eq:sp-0-sc}, 
which does not change the chemical property of the system, 
we arrive at the 
generalized Schur complement \eqref{eq:def-sp} that we introduced earlier. 
In this way, a CRN morphism that shrinks 
a subnetwork gives rise to the reduced stoichiometric matrix given by the Schur complement. 

What we have just shown can be summarized 
as the following statement: 
\theorem{
Under a reduction morphism associated with $\gamma \subset \Gamma$, 
the stoichiometric matrix of $\Gamma'$ 
can be written uniquely (up to the changes of rows and columns) in the form 
\begin{equation}
  S' =
\begin{pmatrix}
  \bm 0 &\,\, S_{22} - S_{21} S_{11}^+ S_{12}   
\end{pmatrix},
\label{eq:sp-0-schur}
\end{equation}
if and only if the following conditions are satisfied
for the subnetwork $\gamma$, \footnote{
We note that the condition \eqref{eq:kers11-in-kers21} 
is equivalent to the absence of 
``emergent cycles'', 
which can be written as $\widetilde c(\gamma) =0$
in the notation of Sec.~\ref{sec:red-buff}. 
We show this equivalence
in Sec.~V below Eq.~\eqref{eq:com-2}. 
The condition \eqref{cokers11-in-cokers12}
implies the absence of ``emergent conserved charges,''
which can be written as $\widetilde d (\gamma)=0$, 
but the converse is not true. 
We discuss more on the meaning of 
emergent cycles and conserved charges 
in Appendix~\ref{sec:app-cycles}. 
}
\begin{equation}
    \ker S_{11} \subset \ker S_{21}, 
    \quad 
    \coker S_{11} \subset \coker S_{12}. 
    \label{eq:no-emergent-cond}
\end{equation}
}

Conversely, for a given output-complete 
subnetwork such that  Eq.~\eqref{eq:no-emergent-cond} is satisfied, 
we can construct the following  reduction map by 
\begin{equation}
  \varphi_0 
  = 
  \begin{pmatrix}
    (- S_{21} S^{+}_{11} + D)^T 
    \\ 
    \bm 1 
  \end{pmatrix}, 
\quad 
\varphi_1 = \bm 1 , 
\label{eq:phi0-phi1-red}
\end{equation}
where $D$ is a matrix satisfying $D S_{11} = \bm 0$. 
The commutativity condition reads 
\begin{equation}
    S' = 
    \begin{pmatrix}
      S_{21} (\bm 1 - S_{11}^+ S_{11}) 
      & 
      \,\, 
      S_{22} - S_{21} S^{+}_{11} S_{12} + D S_{12} 
    \end{pmatrix}. 
\end{equation}
Since $(\bm 1 - S_{11}^+ S_{11})$ is a projection matrix 
to $\ker S_{11}$
and 
we have $\ker S_{11} \subset \ker S_{21}$ by assumption, 
the matrix $S_{21} (\bm 1 - S_{11}^+ S_{11})$ is a zero matrix. Furthermore, $D S_{12} = \bm 0$ by the assumption 
$\coker S_{11} \subset \coker S_{12}$. 
Thus, we arrive at the reduced stoichiometric matrix of the form \eqref{eq:sp-0-schur}. 
\\

Below, let us illustrate reduction morphisms in simple examples. 

\begin{example}
Let us consider the following closed directed graph. 
We consider the morphism, 
which can be pictorially represented as 
\begin{equation}
\begin{tikzpicture}
  
    \node[species] (D1) at (-1.5,0) {$v_1$} ;
    \node[species] (x)  at (   0,0) {$v_2$} ;
    \node[species] (y)  at ( 1.5,0) {$v_3$} ;
    \node[species] (D2) at (   3,0) {$v_4$} ;
  
    \draw[-latex, line width=0.5mm] (D1) edge node[below] {$e_1$} (x);
    \draw[-latex, line width=0.5mm] (x)  edge node[below] {$e_2$} (y);
    \draw[-latex, line width=0.5mm] (y)  edge node[below] {$e_3$} (D2);

     \draw[mydarkred, dashed,line width=1] (-0.3,-1) rectangle (1.2,1);
     \node at (0, 0.7) {  \scalebox{1.1} {\color{mydarkred} $\gamma$} };
     
     \draw [mybrightblue,ultra thick,-latex] (4.00,0) -- (5.00,0); 

    \node[species] (y1) at   (6,0) {$v_1'$}; 
    \node[species] (y2) at (7.5,0) {$v_3'$}; 
    \node[species] (y3) at   (9,0) {$v_4'$}; 
    
    \draw[-latex, line width=0.5mm] (y1) edge node[below] {$e_1'$} (y2);
    \draw[-latex, line width=0.5mm] (y2) edge node[below] {$e_3'$} (y3);
    \draw[-latex,line width=0.5mm]
    (y2) to [out=135,in=45,looseness=5] node[above] {$e_2'$} (y2);
    
  \end{tikzpicture}
\end{equation}
The reduction shrinks the vertices in $\gamma$ 
to a single complex, $\bar c=v_3'$. 
The species and reactions are mapped as 
\begin{align}
  \varphi_0(v_1) &=v_1',  
  \quad \varphi_0(v_2) =v_3' , 
  \quad \varphi_0(v_3) =v_3' ,
  \quad \varphi_0(v_4) =v_4' ,
  \\
  \varphi_1(e_1) &=e_1' ,
  \quad \varphi_1(e_2) =e_2' ,
  \quad \varphi_1(e_3) =e_3' . 
\end{align}
In the matrix form, 
\begin{equation}
  \varphi_0
  =
  \begin{blockarray}{cccc}
    && \\
    \begin{block}{c(ccc)}
     {\color{mydarkred}v_2} \quad\, &0&1&0\\
     v_1\quad\, &1&0 &0\\
     v_3\quad\, &0&1 &0\\
     v_4\quad\, &0&0 &1\\
    \end{block} 
     & {v_1'} & v_3' & v_4' 
  \end{blockarray}
  \quad , 
  \quad 
  \varphi_1 = \bm 1_3 . 
\end{equation}
Using the consistency condition \eqref{eq:commutativity-mat}, 
the stoichiometric matrix of $\Gamma'$ can be written as 
\begin{equation}
S' =   
\begin{pmatrix}
  0&1&0&0 \\
  1&0&1&0 \\
  0&0&0&1 
\end{pmatrix}
\begin{pmatrix}
  -1 & 1 & 0 \\
  0 & -1 & 0 \\
  1 & 0 & -1 \\
  0 & 0 & 1 
\end{pmatrix}
\bm 1_3
=
\begin{blockarray}{cccc}
  && \\
  \begin{block}{c(ccc)}
   v_1' \quad\, &0 &-1&0\\
   v_3' \quad\, &0 &1&-1\\
   v_4' \quad\, &0 &0&1\\
  \end{block} 
   & e_2' & e_1' &e_3'
\end{blockarray} \,\,\,\,. 
\end{equation}
This is indeed of the form~\eqref{eq:sp-0-schur}. 


\end{example}

\begin{example}
$\Gamma = (\{ v_1,v_2,v_3\}, \{ e_1, e_2 \})$. 
We consider the reduction of $\gamma = (\{v_1 \}, \{ e_1\})$.
The corresponding reduction morphism is visualized as 
\begin{equation}
\begin{tikzpicture}[bend angle=45] 

    \draw[mydarkred, dashed,line width=1] (-0.5,0) rectangle (2.5,2);
    \node at (0, 0.2) {  \scalebox{1.1} {\color{mydarkred} $\gamma$} };

    \node[species] (v1) at (0,1) {$v_1$}; 
    \node[species] (v2) at (3,2) {$v_2$};
    \node[species] (v3) at (3,0) {$v_3$}; 

    \node[reaction] (e1) at (1.5,1) {$e_1$}; 

    \draw[-latex,line width=0.5mm] (v1) -- (e1);
    \draw[-latex,line width=0.5mm] (e1) -- (v2);
    \draw[-latex,line width=0.5mm] (e1) -- (v3);
    
    \draw[-latex,line width=0.5mm] (v2) edge node[right] {$e_2$} (v3); 
    \draw [mybrightblue,ultra thick,-latex] (4,1) -- (5,1); 

    \node[species] (v2p) at (7.5,2) {$v'_2$};
    \node[species] (v3p) at (7.5,0) {$v'_3$}; 

    \node[reaction] (e1p) at (6,1) {$e_1'$}; 

    \draw[-latex,line width=0.5mm] (v2p) edge node[right] {$e_2'$} (v3p); 

    \draw[-latex,line width=0.5mm,out=-120,in=10] (v2p) edge node[right] {} (e1p); 
    \draw[-latex,line width=0.5mm,out=120,in=-10] (v3p) edge node[right] {} (e1p); 

    \draw[-latex,line width=0.5mm,out=80,in=180] (e1p) edge node[right] {} (v2p); 
    \draw[-latex,line width=0.5mm,out=-80,in=180] (e1p) edge node[right] {} (v3p); 

\end{tikzpicture}
\end{equation}
The chemical complex map and the reaction map are given by 
\begin{align}
    \varphi_0 (v_1) &= v'_2 + v'_3 , 
    \quad 
    \varphi_0 (v_2) = v'_2 , 
    \quad 
    \varphi_0 (v_3) = v'_3  , \\
    \varphi_1 (e_1) &= e_1', \quad 
    \varphi_1 (e_2) = e_2'. 
\end{align}
The action $\varphi_0(v_1)$ is determined 
so that the image of $e_1$ be degenerate in stoichiometry. 
The image of $e_1$ is the following degenerate reaction, 
\begin{equation}
    e_1' : v'_2 + v'_3 \to v'_2 + v'_3 . 
\end{equation}
In the matrix form, 
\begin{equation}
  \varphi_0
  =
  \begin{blockarray}{ccc}
    && \\
    \begin{block}{c(cc)}
     {\color{mydarkred}v_1} \quad\, &1&1\\
     v_2\quad\, &1&0 \\
     v_3\quad\, &0&1 \\
    \end{block} 
     & {v_2'} & v_3'  
\end{blockarray}\,\,\, , 
\quad 
\varphi_1 = \bm 1_2. 
\end{equation}
The stoichiometric matrix of $\Gamma'$ is written as 
\begin{equation}
    S' = 
    \begin{pmatrix}
      1 & 1 & 0 \\
      1 & 0 & 1 
    \end{pmatrix}
    \begin{pmatrix}
      -1 & 0 \\
      1 & -1 \\
      1 & 1 
    \end{pmatrix}
    = 
    \begin{pmatrix}
      0 & -1 \\
      0 & 1 
    \end{pmatrix} . 
\end{equation}

\end{example}

\section{Reduction and buffering structures}\label{sec:red-buff}

We here explore the close connection between 
the structural sensitivity analysis and 
the reduction method we introduced in the previous section. 
The structural sensitivity analysis works as a guide 
to identify 'unimportant' subnetworks. 
In this section, we present a key result of this paper: 
we show that, when a subnetwork is a buffering structure, 
the reduced network has exactly the same steady-state solution 
as the original reaction network. 
The proof will be completed in 
Sec.~\ref{sec:red-buff-th}. 

The structure of this section is as follows:
In Sec.~\ref{sec:red-buff-decom}, 
we show that 
the influence index allows for a decomposition 
in terms of the numbers of cycles and conserved charges. 
In Sec.~\ref{sec:red-buff-les}, 
we construct a short exact sequence of the chain complexes for a subnetwork $\gamma \subset \Gamma$, under some conditions. 
This short exact sequence automatically derives a long exact sequence of homology groups. 
Using this exact sequence, we can describe the relationship among cycles and conserved charges of $\gamma$, $\Gamma$ and $\Gamma'$.
In Sec.~\ref{sec:red-buff-th}, 
we show the main result, that is, that the steady state of the reduced network is the same as the one of the original network, under some conditions.
In the proof, the long exact sequence prepared in subsection B plays an important role.
In Sec.~\ref{sec:order}, 
we study the situation where we have nested subnetworks $\gamma' \subset \gamma \subset \Gamma$.
In this case, we have a subnetwork $\gamma/ \gamma' \subset \Gamma/\gamma'$.
We will show that the reduced network $\Gamma/\gamma$ is the same as $(\Gamma/\gamma')/(\gamma/\gamma')$.
This ensures that the eventual network does not depend on the ordering of the reductions.

\subsection{Decomposition of the influence index} \label{sec:red-buff-decom}

As we detailed in Sec.~\ref{sec:crn}, 
steady-state properties are captured by 
cycles and conserved charges, 
which are the elements of homology groups. 
In this subsection, 
we study their meaning in more detail, 
and 
discuss the relation between 
the influence index $\lambda (\gamma)$ 
and cycles/conserved charges in $\gamma$, $\Gamma$, and $\Gamma'$. 
We introduce a decomposition of the influence index
in terms of the spaces of cycles/conserved charges of certain classes. 

We first note that the index can be written as 
\begin{equation}
  \begin{split}
  \lambda(\gamma)
  &=
  - |V_\gamma| + |E_\gamma| - |(\ker S)_{{\rm supp}\, \gamma}| 
  + |P^0_\gamma (\coker S)| \\
  & = 
  |\ker S_{11}| - |(\ker S)_{{\rm supp}\, \gamma}| 
  + |P^0_\gamma (\coker S)|
  - |\coker S_{11}| ,
\end{split}
\label{eq:lambda-decom-4}
\end{equation}
where we used Eq.~(\ref{eq:euler-sub}).
With the first two terms, we define 
\begin{equation}
  \widetilde c (\gamma) \coloneqq |\ker S_{11}| - |(\ker S)_{{\rm supp}\, \gamma}| . 
  \label{eq:c-til-def}
\end{equation}
The number $\widetilde c (\gamma)$ is a nonnegative integer, 
because there is an injective map 
from $(\ker S)_{{\rm supp}\, \gamma}$ to $\ker S_{11}$. 
Indeed, an element of $(\ker S)_{{\rm supp}\, \gamma}$ is written as 
$
\bm c = 
\begin{pmatrix}
  \bm c_1 \\
  \bm 0
\end{pmatrix}
$
satisfying the condition 
\begin{equation}
  \begin{pmatrix}
    S_{11} & S_{12}\\
    S_{21} & S_{22}
  \end{pmatrix}
  \begin{pmatrix}
    \bm c_1 \\
    \bm 0
  \end{pmatrix}
  =
  \begin{pmatrix}
    S_{11}\bm c_1 \\
    S_{21}\bm c_1 
  \end{pmatrix}
  = \bm 0. 
  \label{eq:sc1}
\end{equation}
Consider an injective map $\bm c \mapsto \bm c_1$. 
Equation~(\ref{eq:sc1}) indicates 
that the image of this map is always included in $\ker S_{11}$, 
\begin{equation}
  (\ker S)_{{\rm supp}\, \gamma} 
  \ni \bm c \mapsto \bm c_1 
  \in \ker S_{11} . 
\end{equation}
Thus, we have $\widetilde c(\gamma)\ge 0$. 

Now let us turn to the latter two terms in Eq.~(\ref{eq:lambda-decom-4}). 
Note that \footnote{
  For a vector space $V$ and 
  a projection matrix $P$, 
  \begin{equation}
    |P V| = 
    |\{ P \bm v \,|\, \bm v \in V  \}|
    = 
    |V / ( {\rm im}\, \bar P \cap V )|
    =
    |V| - |  {\rm im}\, \bar P \cap V |, 
  \end{equation}
  where $\bar P \coloneqq 1 - P$. 
} 
\begin{equation}
  |P^0_\gamma (\coker S)| 
  = 
  |\coker S|
  - 
  |{\rm im}\, \bar P^0_\gamma  \cap \coker S|, 
  \label{eq:p-cok-s}
\end{equation}
where $\bar P^0_\gamma \coloneqq 1 - P^0_\gamma$. 
The second term 
of the RHS of Eq.~(\ref{eq:p-cok-s}) is the number of 
the conserved charges of $\Gamma$ supported 
in $\Gamma \setminus \gamma$, 
\begin{equation}
\bar d' (\gamma) \coloneqq  |{\rm im}\, \bar P^0_\gamma  \cap \coker S| 
=
|\bar D'(\gamma)|, 
\end{equation}
where the space $\bar D'(\gamma)$ is given by 
\begin{equation}
    \bar D'(\gamma)
    \coloneqq 
    \left\{ 
      \begin{pmatrix}
        \bm d_1 \\
        \bm d_2
      \end{pmatrix}
      \in \coker S 
      \, \middle| \, \bm d_1 = \bm 0 
      \right\} . \label{eq:Dbp-gamma}
\end{equation}
We divide the space $\coker S$ 
according to the following distinctions: 
\begin{itemize}
    \item Projection to $\gamma$ is also a conserved charge in $\gamma$. 
    \item Projection to $\gamma$ is not a conserved charge in $\gamma$. 
\end{itemize}
Correspondingly to the two distinctions above, 
we introduce the following spaces, \footnote{
Note that we can regard the element of $D(\gamma)$ as a vector in $\coker S$ 
by the isomorphism 
$X(\gamma) / \bar D' (\gamma) \cong X(\gamma) \cap [\bar D'(\gamma)]^\perp$. 
The isomorphism in Eq.~\eqref{eq:Dp-gamma}  
can be derived as follows: 
\[
(\coker S) / D(\gamma)
=
(\coker S) /
\left( X(\gamma) \cap [\bar D'(\gamma)]^\perp \right) 
\cong 
(\coker S) \cap \left( X(\gamma) \cap [\bar D'(\gamma)]^\perp \right)^\perp 
= 
(\coker S) \cap \left( [X(\gamma)]^\perp + \bar D'(\gamma) \right) 
= 
(\coker S) /  X(\gamma) \oplus \bar D'(\gamma) ,
\] 
where we used the relations
$V / W \cong V \cap W^\perp$, 
$(V \cap W)^\perp = V^\perp + W^\perp$ 
for vector spaces $W \subset V$, 
and 
$[X(\gamma)]^\perp + \bar D' (\gamma) 
= [X(\gamma)]^\perp \oplus \bar D' (\gamma)$ 
since $[X(\gamma)]^\perp \cap \bar D'(\gamma) = \bm 0$. 
}
\begin{align}
  D(\gamma)
  &\coloneqq X(\gamma)  / \bar D' (\gamma) , 
    \label{eq:D-gamma} \\
  D'(\gamma)
  &\coloneqq 
   \coker S / D(\gamma) 
   \cong 
    (\coker S) / X(\gamma)  \oplus \bar D' (\gamma) 
 ,\label{eq:Dp-gamma} 
\end{align}
where we defined 
\begin{equation}
X (\gamma) 
\coloneqq 
  \left\{ 
    \begin{pmatrix}
      \bm d_1 \\
      \bm d_2
    \end{pmatrix}
    \in \coker S 
    \, \middle| \,
    \bm d_1 \in \coker S_{11}
    \right\}. 
\end{equation}
Namely, we have the following decomposition of $\coker S$, 
\begin{equation}
    \coker S 
    \cong D(\gamma) \oplus D'(\gamma) 
    \cong D(\gamma) \oplus (\coker S)/X(\gamma) 
    \oplus \bar D'(\gamma) . 
    \label{eq:cokers-decom}
\end{equation}
The dimension of $\coker S$ is written as 
\begin{equation}
  |\coker S| = d(\gamma) + d'(\gamma)  , 
\end{equation}
where 
$d(\gamma)\coloneqq |D(\gamma)|$, and $d'(\gamma)\coloneqq |D'(\gamma)|$. 
We now have the expression 
\begin{equation}
  |P^0_\gamma (\coker S)| 
  = d (\gamma)
  + d'(\gamma) - \bar d'(\gamma). 
  \label{eq:p-cok-s-d}
\end{equation}

To rewrite $|\coker S_{11}|$, we introduce the following spaces, 
%
\begin{align}
  D_{11}(\gamma)
  &\coloneqq 
  \left\{ 
    \bm d_1 \in \coker S_{11} 
  \, \middle| \,
  {}^{\exists} \bm d_2 {\text{ such that }} 
  \begin{pmatrix}
    \bm d_1 \\
    \bm d_2
  \end{pmatrix}
  \in \coker S 
  \right\} ,
   \\
  \widetilde D(\gamma)
  &\coloneqq \coker S_{11} / D_{11}(\gamma) 
   . 
\end{align}
The elements of $D_{11}(\gamma)$ are conserved charges in $\gamma$ 
that can be extended to a global conserved charge, 
while those in $\widetilde D(\gamma)$ are emergent conserved charges
that are only conserved in the subnetwork $\gamma$.

Observe that $D (\gamma) \cong D_{11}(\gamma)$. 
Indeed, there is a surjection 
\begin{equation}
X(\gamma) \ni \begin{pmatrix}
  \bm d_1  \\
  \bm d_2
\end{pmatrix} 
\to \bm d_1 \in  D_{11}(\gamma) . 
\end{equation}
The kernel of this map is $\bar D' (\gamma)$, 
and the induced map $D(\gamma) = X(\gamma) / \bar D' (\gamma) \to D_{11}(\gamma)$ 
is an isomorphism. 
%
%
%
Thus, 
$|D_{11}(\gamma)| = |D(\gamma)| = d (\gamma)$ and 
we have the decomposition, 
\begin{equation}
  |\coker S_{11}| = d(\gamma) + \widetilde d (\gamma) , 
  \label{eq:d-cok-s11}
\end{equation}
where $\widetilde d(\gamma) \coloneqq |\widetilde D(\gamma)|$ is the number of charges 
that cannot be obtained as the projections of conserved charges in $\Gamma$. 

Combining Eqs.~(\ref{eq:c-til-def}), (\ref{eq:p-cok-s-d}), and (\ref{eq:d-cok-s11}), 
we find that the influence index is written as  
\begin{equation}
  \lambda (\gamma)
  = \widetilde c (\gamma) + d_l (\gamma)  - \widetilde d (\gamma) ,
  \label{eq:lambda-decom}
\end{equation}
where we defined 
\begin{equation}
d_l (\gamma)
\coloneqq
d' (\gamma) - \bar d' (\gamma)
=
|(\coker S) / X(\gamma)| 
.   
\end{equation}
The decomposition (\ref{eq:lambda-decom}) is the central result of this subsection. 
Each term of Eq.~(\ref{eq:lambda-decom}) allows for 
the following intuitive interpretations: 
\begin{itemize}
\item The first term 
$\widetilde c(\gamma)= |\ker S_{11} / (\ker S)_{{\rm supp}\, \gamma}|$ 
represents the number of {\it emergent cycles} in $\gamma$. 
Namely, $\widetilde c(\gamma)$ is the number of cycles in $\gamma$, which are not cycles in $\Gamma$. 
\item The second term 
$d_l (\gamma)=|(\coker S) / X(\gamma)| $
is the dimension of the space of {\it lost conserved charges} 
by focusing on $\gamma$, namely those that are conserved in $\Gamma$ 
but their projection to $\gamma$ are not. 
%
%
\item The third term $\widetilde d(\gamma) = |\widetilde D(\gamma)|$ 
is the number of {\it emergent conserved charges} in $\gamma$. 
It is the number of conserved charge in $\gamma$
that cannot be extended to conserved charges in $\Gamma$. 
The meaning becomes evident if we note 
that $\widetilde D(\gamma)$ is isomorphic to 
the space 
that consists of $\bm d_1 \in \coker S_{11}$ 
that are orthogonal to the vectors that can be extended to 
conserved charges in $\coker S$. 
\end{itemize}
For more detailed explanations with examples, 
see Appendix~\ref{sec:app-interpretations}. 
In Appendix~\ref{sec:app-embedding}, 
we show that the decomposition \eqref{eq:lambda-decom}
can be visually understood from the structure of A-matrices.

An element of $D'(\gamma)$ can be regarded as a conserved charge in $\Gamma'$ via an injective map $\bar \varphi_0 : D'(\gamma) \to \coker S'$, which we will construct as follows.
We define 
$\bar \varphi_0$  
on each component of $D'(\gamma) = (\coker S)/X(\gamma) \oplus \bar{D}'(\gamma)$.
The map $\bar{D}'(\gamma) \to \coker S'$ is given by 
$
\begin{pmatrix}
  \bm 0 \\
  \bm d_2
\end{pmatrix}
\mapsto {\bm d_2},
$ 
which is obviously injective, and is well-defined since ${\bm d_2}$ belongs to $\coker S'$ by 
\begin{equation}
    \bm d_2^T S' 
    = 
    \bm d_2^T (S_{22} -  S_{21} S_{11}^+ S_{12} )
  = \bm 0.
\end{equation}
Note that the second equality follows from
$\bm d_2^T S_{22} = \bm 0$
and 
$\bm d_2^T S_{21} = \bm 0$, 
which hold by the assumption $\begin{pmatrix}
  \bm 0 \\
  \bm d_2
\end{pmatrix} \in \bar{D}'(\gamma)$.
Next, we construct an injection $(\coker S) / X(\gamma) \to \coker S'$.
For 
$
\left[ 
\begin{pmatrix}
  \bm d_1 \\
  \bm d_2
\end{pmatrix}
\right]
\in( \coker S)/X(\gamma)
$, 
we can always choose a representative 
$
\begin{pmatrix}
  \bm d_1 \\
  \bm d_2
\end{pmatrix}
$
such that $\bm d_1 \in (\coker S_{11})^\perp$
and 
$
\begin{pmatrix}
  \bm d_1 \\
  \bm d_2
\end{pmatrix}
\in [\bar D'(\gamma)]^\perp . 
$
Using $\bm d^T S = \bm 0$, 
\begin{equation}
  \begin{split}
  \bm d_2^T S' 
  &= 
  \bm d_2^T S_{22}
  - 
  \bm d_2^T S_{21} S_{11}^+ S_{12}  \\
  &=
  - \bm d_1^T S_{12}
  +
  \bm d_1^T S_{11} S_{11}^+ S_{12}  \\
  &= 
  - 
  \bm d_1^T 
  (1 - S_{11} S_{11}^+)
  S_{12} . 
\end{split}
\label{eq:d2sp}
\end{equation}
Since $(1-S_{11} S_{11}^+)$ is a projection matrix to $\coker S_{11}$, 
we have 
$\bm d_1^T (1-S_{11} S_{11}^+) = \bm 0$, 
and thus $\bm d_2 \in \coker S'$. 
This defines an injective map 
$(\coker S)/X(\gamma) \to \coker S'$. 

Thus, we have obtained a map $\bar \varphi_0 : D'(\gamma) = (\coker S)/X(\gamma) \oplus \bar{D}'(\gamma) \to \coker S'$. 
To see the injectivity of $\bar \varphi_0$, since it is injective on each component, it suffices to show that the intersection of the images of $(\coker S)/X(\gamma)$ and $\bar D'(\gamma)$ 
by $\bar \varphi_0$ is zero.
To show this, let 
us pick an arbitrary element 
$\bm d_2 \in \bar \varphi_0 ( (\coker S)/X(\gamma) ) \cap \bar\varphi_0 (\bar D'(\gamma))$. 
It suffices to show that $\bm d_2 = \bm 0$.
Since $\bm d_2$ comes from $(\coker S) /  X(\gamma)$ by assumption, 
there is an element 
$
\begin{pmatrix}
  \bm d_1 \\
  \bm d_2
\end{pmatrix}
\in \coker S
$
such that $\bm d_1 \in (\coker S_{11})^\perp$. 
Since $\bm d_2$ is also in the image of $\bar D'(\gamma)$, 
we have $\bm d_1=\bm 0$, 
and 
$\begin{pmatrix}
  \bm 0 \\
  \bm d_2
\end{pmatrix} \in \bar D'(\gamma)^\perp$. 
This means that $\bm d_2 = \bm 0$ as desired.

We also define $\bar \varphi_0$ for the elements of 
$D(\gamma)$ by $\bar \varphi_0 |_{D(\gamma)}=\bm 0$. 
Hence, $\bar \varphi_0$ is now defined as a map 
from $\coker S$ to $\coker S'$,
and its kernel and coimage are given by 
$\ker \bar \varphi_0 = D(\gamma)$ and 
${\rm coim}\, \bar \varphi_0 = D'(\gamma)$.

In general, the conserved charges in $\Gamma'$ 
consists of those obtained from the conserved charges of $\Gamma$ 
and emergent ones, 
\begin{equation}
 | \coker S' |  = d' (\gamma )+ \widetilde d'(\gamma) , 
\end{equation}
where 
$\widetilde d' (\gamma) \coloneqq | (\coker S') / {\rm im\,} \bar\varphi_0 |$
indicates the number of emergent conserved charges in $\Gamma'$.

\subsection{Long exact sequence of a pair of chemical reaction networks} \label{sec:red-buff-les}

The reduction of a reaction network naturally induces 
the reduction of (co)homology groups, 
which are the steady-state characteristics of reaction networks. 
Suppose that we have a reaction network $\Gamma$, 
and choose a subnetwork $\gamma \subset \Gamma$, 
and reduce it to obtain $\Gamma' = \Gamma/ \gamma$. 
The inter-relations of homologies of $\gamma$, $\Gamma$, 
and $\Gamma'$, 
can be systematically treated 
using a long exact sequence 
for a pair of chemical reaction networks, 
which we define momentarily. 
We consider the following short exact sequence of chain complexes, 
\begin{equation}
  \xymatrix{
    &
    0 
    \ar[d]
    &
    0 \ar[d]
    & 
    0 \ar[d]
    \\
   0 \ar[r] &
   C_1 (\gamma)
   \ar[r]^{\psi_1 }
   \ar[d]^{\p_\gamma } 
   &
   C_1 (\Gamma) 
   \ar[r]^{\varphi_1}
   \ar[d]^{\p }
   & 
   C_1(\Gamma' )
   \ar[r]
   \ar[d]^{\p' }
   & 0
   \\
0
 \ar[r] 
&  
C_0(\gamma)
   \ar[r]^{\psi_0 }
   \ar[d]   
   & 
C_0 (\Gamma) 
  \ar[r]^{\varphi_0}
  \ar[d] 
   &
C_0 (\Gamma' )
   \ar[r]
   \ar[d]
   &
   0
   \\
   & 
   0 
   & 0 
   & 0 
 }
 \label{eq:diag-b}
\end{equation}
where 
the space of chains in $\Gamma'$ is given by 
$C_n(\Gamma') \coloneqq C_n(\Gamma) / C_n (\gamma)$. 
In the linear-algebra notations, 
the boundary maps are given by the following 
multiplications of matrices on vectors, 
\begin{equation}
  \p_\gamma: \bm c_1 \mapsto S_{11} \bm c_1, 
  \quad 
  \p: \bm c = 
  \begin{pmatrix}
    \bm c_1 \\
    \bm c_2
  \end{pmatrix}
  \mapsto 
  S \bm c, 
  \quad 
  \p': \bm c_2 \mapsto S' \bm c_2. 
\end{equation}
We define the horizontal maps by 
\begin{equation}
  \psi_1: 
  \bm c_1 \mapsto 
  \begin{pmatrix}
    \bm c_1 \\ 
    \bm 0
  \end{pmatrix}, 
  \quad 
  \varphi_1: 
  \begin{pmatrix}
    \bm c_1 \\ 
    \bm c_2
  \end{pmatrix}
  \mapsto \bm c_2 , 
\end{equation}
\begin{equation}
  \psi_0: 
  \bm d_1 \mapsto 
  \begin{pmatrix}
    \bm d_1 \\ 
    S_{21} S^+_{11} \bm d_1
  \end{pmatrix}, 
  \quad 
  \varphi_0: 
  \begin{pmatrix}
    \bm d_1 \\ 
    \bm d_2
  \end{pmatrix}
  \mapsto \bm d_2 - S_{21}S^+_{11} \bm d_1 . 
\end{equation}
The exactness of the rows of Eq.~(\ref{eq:diag-b}) can be checked easily. 
Note that $\varphi$ is the reduction morphism \eqref{eq:phi0-phi1-red} 
followed by the removal of degenerate reactions. 
One can check that 
the diagram (\ref{eq:diag-b}) commutes when the following condition is satisfied: 
\begin{equation}
  S_{21} (1 - S_{11}^+ S_{11}) \bm c_1 = \bm 0, 
  \label{eq:com-1}
\end{equation}
where $\bm c_1 \in C_1(\gamma)$. 
The matrix $ (1 - S_{11}^+ S_{11})$ 
is the projection matrix to $\ker S_{11}$, 
and Eq.~(\ref{eq:com-1}) is equivalent to 
\begin{equation}
  \ker S_{11} \subset \ker S_{21}  . 
  \label{eq:com-2}
\end{equation}
This condition is the same as the condition 
that an arbitrary term in Eq.~(\ref{eq:rate-red-1}) vanishes.

The condition (\ref{eq:com-2}) 
has a natural interpretation in terms of cycles: 
 Eq.~(\ref{eq:com-2}) is equivalent to $\widetilde c(\gamma) = 0$, 
namely the absence of emergent cycles, 
which can be checked as follows. 
When $\widetilde c (\gamma) = |\ker S_{11} / (\ker S)_{{\rm supp}\, \gamma}| = 0$, 
any $\bm c_1\in \ker S_{11}$ is a cycle in $\Gamma$ 
by an inclusion to $C_1(\Gamma)$. 
Thus, $\bm c_1$ satisfies 
\begin{equation}
  \begin{pmatrix}
    S_{11} & S_{12} \\
    S_{21} & S_{22}
  \end{pmatrix}
  \begin{pmatrix}
    \bm c_1 \\
    \bm 0
  \end{pmatrix}
  = \bm 0. 
  \label{eq:sc1-eq-0}
\end{equation}
This implies $S_{21} \bm c_1 = \bm 0$ 
and we have $\ker S_{11} \subset \ker S_{21}$. 
Conversely, when $\ker S_{11} \subset \ker S_{21}$
is true, 
the map 
$
\ker S_{11} \ni \bm c_1 \mapsto \begin{pmatrix}
  \bm c_1 \\
  \bm 0
\end{pmatrix}
\in (\ker S)_{{\rm supp}\, \gamma} 
$
is a bijection.
This implies $\widetilde c (\gamma) = 0$. 
Thus, we have shown that the diagram (\ref{eq:diag-b}) commutes 
if and only if $\gamma$ has no emergent cycle.

Applying the snake lemma to Eq.~(\ref{eq:diag-b}), 
we obtain a long exact sequence, \\
\begin{equation}
  \xymatrix{
   0 \ar[r] &
   H_1(\gamma) 
   \ar[r]^{\psi_1} 
   &
   H_1(\Gamma) 
   \ar[r]^{\varphi_1}
   &
   H_1(\Gamma')
   \ar[r]^{\delta_1\quad } &  
   H_0(\gamma) 
   \ar[r]^{\bar{\psi}_0}
   & 
   H_0(\Gamma) 
   \ar[r]^{\bar{\varphi}_0} 
   &
   H_0 (\Gamma') 
   \ar[r]
   &
   0 , 
 }  
 \label{eq:les}
\end{equation}
\\ 
where $\bar \psi_0$ and $\bar \varphi_0$ 
are induced maps of $\psi_0$ and $\varphi_0$. 
The map $\delta_1: H_1(\Gamma') \to H_0(\gamma)$ 
is called the connecting map. 
For a given $\bm c_2 \in H_1 (\Gamma')$, 
the connecting map is given by \footnote{
The connecting map is identified as follows. 
An element $\bm c_2 \in H_1(\Gamma')$, 
can be included in $C_1(\Gamma')$. 
$\varphi_1$ is surjective and there exists 
$
\bm c = 
\begin{pmatrix}
  \bm c_1 \\
  \bm c_2 
\end{pmatrix}$ 
such that $\varphi_1 (\bm c) = \bm c_2$. 
From the commutativity of the diagram (\ref{eq:diag-b}), 
we have $\varphi_0 (S\bm c) = S' \bm c_2 = \bm 0$. 
From the exactness of the row of Eq.~(\ref{eq:diag-b}), 
there exists $\bm d_1 \in C_0(\gamma)$ 
such that $\psi_0 (\bm d_1) = S \bm c$. 
We obtain $[\bm d_1] \in H_0(\gamma)$ by identifying the differences in ${\rm im\,} S_{11}$. 
More explicitly, 
$[\bm d_1] = [S_{11} \bm c_1 + S_{12} \bm c_2]
= [S_{12} \bm c_2]$. 
The mapping $\bm c_2 \mapsto [S_{12} \bm c_2]$ is the connecting map. 
The well-definedness of the map 
(indifference to the choice of $\bm c_1$) is obvious in this expression. 
} 
\begin{equation}
\delta_1: 
\bm c_2 \mapsto 
[S_{12} \bm c_2] 
\in H_0 (\gamma) = \coker S_{11}
, 
\end{equation}
where $[...]$ means to identify the differences in ${\rm im\,} S_{11}$.

Let us look at the consequences of the long exact sequence (\ref{eq:les}). 
Suppose that 
we choose $\gamma$ so that its homology groups are trivial, 
\begin{equation}
  H_1(\gamma) \cong \ker S_{11} \cong \bm 0, 
  \quad 
  H_0(\gamma) \cong \coker S_{11} \cong \bm 0. 
\end{equation}
Then, we have the isomorphisms, 
\begin{equation}
  H_1(\Gamma) 
  \cong 
  H_1(\Gamma'), 
  \quad 
  H_0(\Gamma) 
  \cong 
  H_0(\Gamma') , 
\end{equation}
equivalently, 
\begin{equation}
  \ker S
  \cong 
  \ker S',
  \quad 
  \coker S
  \cong 
  \coker S' .  
\end{equation}
Thus, the spaces of cycles and conserved charges 
before and after the reduction are isomorphic 
when $\gamma$ has trivial homologies. 
Example \ref{ex:v2e3} in Sec.~\ref{sec:reduction} corresponds to this situation, 
where the partial stoichiometric matrix is given by 
$S_{11} = \begin{pmatrix} -1 \end{pmatrix}$, 
whose kernel and cokernel are trivial.

The exact sequence applies as long as the commutativity condition, 
$\ker S_{11} \subset \ker S_{21}$, 
is satisfied, and we can consider more general cases 
with $\ker S_{11} \neq \bm 0$. 
If the connecting map $\delta_1: H_1(\Gamma') \to H_0(\gamma)$ is 
a zero map, 
the long exact sequence (\ref{eq:les}) results in the following two exact sequences, 
\begin{eqnarray}
&&  \xymatrix{
   0 \ar[r] &
   H_1(\gamma) 
   \ar[r]
   &
   H_1(\Gamma) 
   \ar[r] 
   &
   H_1(\Gamma')
  \ar[r] 
  & 0
 }  , \\  
 &&
 \xymatrix{
  0 \ar[r] &
  H_0(\gamma) 
  \ar[r]
  &
  H_0(\Gamma) 
  \ar[r] 
  &
  H_0(\Gamma') 
 \ar[r] 
 & 0
}  . 
\end{eqnarray}
This implies the isomorphisms, 
\begin{equation}
  \ker S / \ker S_{11} \cong \ker S', 
  \quad 
  \coker S / \coker S_{11} \cong \coker S'.
  \label{eq:iso-2}
\end{equation}
Note that $\ker S_{11}$ consists of only locally supported global cycles, 
due to the assumption $\ker S_{11} \subset \ker S_{21}$. 
The isomorphisms (\ref{eq:iso-2}) represent 
equivalence of chemical reaction networks 
up to locally supported global cycles
and locally supported global conserved charges 
(emergent conserved charges are also absent when $\delta_1$ is a zero map, as we see below). 

Let us examine the condition when the connecting map is a zero map. 
$\delta_1$ is a zero map if
\begin{equation}
  S_{12} \bm c_2 \in {\rm im\,} S_{11}
  = 
  (\coker S_{11})^\perp 
  ,
  \label{eq:b12-c2}
\end{equation}
for any $\bm c_2 \in H_1(\Gamma')$. 
Below we show that, 
if every conserved charge in $\gamma$
is obtained by the projection of a global conserved charge in $\Gamma$ (namely, when there is no emergent conserved charge), 
the connecting map $\delta_1$ is a zero map. 
For a given $\bm d_1 \in \coker S_{11}$, there exists 
an element of $\coker S$, $\bm d^T = (\bm d_1^T, \bm d_2^T)$. 
The condition $\bm d^T S = \bm 0$ reads 
\begin{eqnarray}
  \bm d_2^T S_{21} &=& \bm 0 , \label{eq:db1}\\
  \bm d_1^T S_{12} + \bm d_2^T S_{22} &=& \bm 0,  \label{eq:db2}
\end{eqnarray}
where we used $\bm d^T_1 S_{11}=\bm 0$.
Let us pick $\bm c_2 \in H_1(\Gamma')=\ker S'$. 
The quantity $\bm d_1^T S_{12} \bm c_2$
can be shown to vanish as follows: 
\begin{equation}
  \bm d_1^T S_{12} \bm c_2 
  \underset{{\rm (\ref{eq:db2})}}{=}
  - \bm d_2^T S_{22} \bm c_2 
  \underset{S'\bm c_2=\bm 0}{=}
  - \bm d_2^T S_{21} S^+_{11} S_{12} \bm c_2 
  \underset{\rm (\ref{eq:db1})}{=} 
  0 . 
\end{equation}
Therefore, we have shown 
$\bm d_1^T S_{12} \bm c_2 =  0$ 
for any $\bm d_1 \in \coker S_{11}$ and $\bm c_2 \in \ker S'$. 
This is equivalent to $S_{12} \bm c_2 \in (\coker S_{11})^\perp$.

The relation between the long exact sequence 
and the numbers of cycles and conserved charges of various types 
is summarized in Fig.~\ref{fig:les}. 
The vertical lines represent the spaces, 
and the kernels are shown in black. 
Since it is an exact sequence, the kernel and image coincide 
at each space, such as ${\rm im\,} \psi_1 = \ker \varphi_1$ and so on. 
The exactness is the key to the connections between cycles and conserved charges of particular types. 
Let us see an example. 
The image of $\delta_1$ is the space of emergent conserved charges,  ${\rm im \,} \delta_1 = \widetilde D(\gamma)$. 
They are emergent, because the image of $\delta_1$ is the kernel of $\bar\psi_0$, and there is no counterpart in $\Gamma$. 
The connecting map $\delta_1$ provides us with 
a one-to-one mapping between an emergent cycle in $\Gamma'$
and an emergent conserved charge in $\gamma$ 
(elements of $\ker \delta_1$ are not emergent, since they can be written as an image of $\varphi_1$ due to the exactness). 
The numbers $d(\gamma)$, $d'(\gamma)$ in Fig.~\ref{fig:les} are 
the same as the dimensions 
of the spaces (\ref{eq:D-gamma}) and (\ref{eq:Dp-gamma})
that we defined previously. 

Compare Fig.~\ref{fig:les} also with Fig.~\ref{fig:mat-a} in the Appendix~\ref{sec:app-embedding}, 
where we discuss the relation between the numbers of cycles and conserved charges and the structure of the A-matrix. 
The long exact sequence is valid when $\widetilde c(\gamma)=0$ (i.e., when the diagram (\ref{eq:diag-b}) commutes). 
This implies $\widetilde d'(\gamma)=0$ 
and there is not emergent conserved charge in $\Gamma'$, 
since $\bar \varphi_0$ is surjective. 

\begin{figure}[tbh]
  \centering
  \includegraphics[clip, trim=0cm 8cm 0cm 5cm, width=\textwidth]{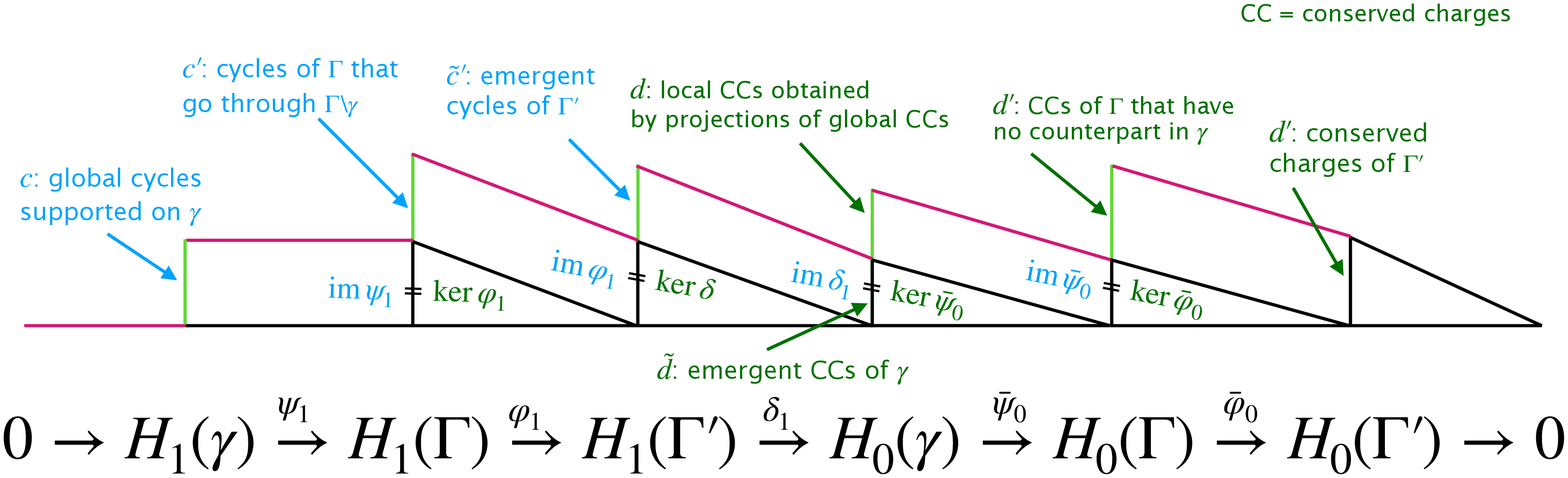} 
  \caption{
  Long exact sequence and conserved charges/cycles of various types. 
  }
  \label{fig:les}
\end{figure}

\subsection{Reduction of buffering structures}\label{sec:red-buff-th}

Here we present the main result, 
regarding the reduction of buffering structures. 
The following theorem represents a particularly nice property of 
buffering structures under reductions. 
We show that the steady-state concentrations and rates 
of the network obtained by reducing a buffering structure 
are exactly the same as those of the network before reduction, 
without any modification of parameters. 
Thus, the reduction of a buffering structure preserves the steady-state properties of the boundary degrees of freedom. 
The theorem only relies on topological information of the network and is true {\it regardless of the kinetics.} 

\begin{theorem} 
 Let $\Gamma$ be a regular chemical reaction network with kinetics $\bm r(\bm x)$ 
 and 
 let $\gamma$ be an output-complete subnetwork of $\Gamma$. 
 We assume that 
 the subnetwork $\gamma$ does not have an emergent conserved charge. 
 We consider a reduced network $\Gamma' = \Gamma/\gamma$. 
 If $\gamma$ is a buffering structure, 
 we have the isomorphisms, 
\begin{equation}
    \ker S / \ker S_{11} \cong \ker S', 
    \quad 
    \coker S / \coker S_{11} \cong \coker S', 
    \label{eq;ker-coker-mapping}
\end{equation}
Furthermore, when $(\bm r, \bm x)$ is 
steady-state reaction rates concentrations of $\Gamma$, 
whose components we separate into those in 
$\gamma$ and $\Gamma \setminus \gamma$ as 
\begin{equation}
  \bm r = \begin{pmatrix}
    \bm r_1  \\
    \bm r_2
  \end{pmatrix}, 
  \quad 
  \bm x = \begin{pmatrix}
    \bm x_1 \\
    \bm x_2
  \end{pmatrix}, 
\end{equation}
then, $(\bm r_2, \bm x_2)$ is a steady-state solution of $\Gamma'$. 
\label{th:reduction}
\end{theorem}

\begin{remark}
Let us comment on the assumption of the absence of emergent conserved charges. 
Under the assumption of the regularity, 
the appearance of emergent conserved charges in 
an output-complete subnetwork $\gamma$ is quite unlikely. 
In fact, in the case of monomolecular reaction networks, 
we can prove $\widetilde d(\gamma) = 0$ 
for a connected and output-complete subnetwork $\gamma$, 
assuming that $\Gamma$ is regular (see Appendix~\ref{sec:app-mono}), 
and this condition is redundant. 
So far, 
the examples of buffering structures 
with nonzero emergent conserved charges 
are pathological in some sense. 
Presently, we have not been able to prove the absence of 
emergent conserved charges for a generic (sound) reaction network, and thus it is assumed. 
We have more discussions on this point in Appendix~\ref{sec:app-emergent}. 
\end{remark}

\begin{remark}
We note that there is a possibility that the reduced system 
might have some solutions which are not allowed in the original system, depending on the kinetics\footnote{
We appreciate the anonymous referee 
for pointing out this possibility. 
}. 
This can occur when the reactions in a subnetwork 
have limitations in the values of reaction rates.
When such a subnetwork is removed, 
the reduced system does not have the restrictions, 
and there may appear additional solutions. 
Let us illustrate this with an example. 
We consider a reaction network $\Gamma = (\{v_1,v_2,v_3 \},\{ e_1,e_2,e_3,e_4,e_5,e_6,e_7 \} )$ 
given by the following set of reactions, 
\begin{align}
    e_1 &: v_1 \to v_2, \notag \\
    e_2 &: v_2 \to \text{(output)},  \notag \\
    e_3 &: 2 v_3 \to 3 v_3  , \notag\\
    e_4 &: v_3 \to \text{(output)} ,  \\
    e_5 &: \text{(input)} \to v_3, \notag \\
    e_6 &: v_3 \to v_1 .\notag  \\
    e_7 &: 3 v_3 \to \text{(output)}. \notag 
\end{align}
Let us here choose the kinetics as 
\begin{equation}
        \begin{pmatrix}
      r_1(x_1) \\
      r_2(x_2)  \\
      r_3(x_3) \\
      r_4(x_3) \\
      r_5 \\
      r_6(x_3) \\
      r_7 (x_3)
    \end{pmatrix}
= 
    \begin{pmatrix}
      \frac{k_1 x_1}{c_1 + x_1} \\
      k_2 x_2 \\
      k_3 (x_3)^2 \\
      k_4 x_3  \\
      k_5  \\
      k_6 x_3 \\
      k_7 (x_3)^3 
    \end{pmatrix}. 
\end{equation}
For reaction $r_1$, we adopted the Michaelis-Menten kinetics, 
and we chose the mass-action kinetics for other reactions. 
The rate equations read 
\begin{align}
    \frac{d}{dt} x_1 &= r_6(x_3) - r_1 (x_1) 
    = k_6 x_3 - \frac{k_1 x_1} {c_1 + x_1} , \label{eq:kin-x1} 
    \\
    \frac{d}{dt} x_2 
    &=   r_1 (x_1) - r_2 (x_2) 
    = \frac{k_1 x_1} {c_1 + x_1} 
    - k_2 x_2 ,
    \\
    \frac{d}{dt} x_3 &= 
    - k_7 (x_3)^3
    + k_3 (x_3)^2 - k_4 x_3 + k_5 - k_6 x_3 
    \eqqcolon 
    - k_7  (x_3 - d_1) (x_3 - d_2) (x_3 - d_3), 
    \label{eq:kin-x3}
\end{align}
where we reparametrized the equation for $\frac{d}{dt} x_3$ 
using $d_1$, $d_2$, and $d_3$
such that $d_1 < d_2 < d_3$. 
From Eq.~\eqref{eq:kin-x3}, 
we get two candidates of stable steady-state values, 
$\bar x_3 = d_1, d_3$ (note that $d_2$ is unstable). 
However, those candidates may not lead to the solutions 
of whole equations 
when the reaction rate $r_1$ has a bound, as in the current example. 
Using Eq.~\eqref{eq:kin-x1}, 
the steady-state value of $x_1$ is given by 
\begin{equation}
    \bar x_1 = \frac{c_1 k_6 \bar x_3}{k_1 - k_6 \bar x_3} . 
    \label{eq:kin-barx1}
\end{equation}
If the denominator of Eq.~\eqref{eq:kin-barx1} is negative, 
it is not a valid solution. 
Thus, depending on the values of $d_1$ and $d_3$, 
the original network may have no, or one, or two solutions. 
The subnetwork $\gamma = \{ (v_1), (e_1) \}$ is a buffering 
structure, and we can consider the corresponding reduced network $\Gamma' = \Gamma /\gamma$. 
In the reduced network, such a restriction on the values of (internal) reaction rates is invisible. 
Hence, the reduced system may admit more solutions 
that were not possible in the original system. 
\end{remark}

\begin{proof}
  The regularity of $\Gamma$ requires $\lambda(\gamma) \ge 0$ (Remark~\ref{rem:lambda-ge-0}). 
  In the absence of the emergent conserved charges, 
  we have 
\begin{equation}
  \lambda (\gamma ) = \widetilde c(\gamma) + d_l (\gamma) = 0 . 
\end{equation}
Since $\widetilde c(\gamma)$ and $d_l (\gamma)$ are 
nonnegative integers, 
we have 
$\widetilde c (\gamma) = 0$ and $d_l (\gamma) = 0$. 
Since $\widetilde c(\gamma)=0$, we can use the long exact sequence (\ref{eq:les}). 
Because there is no emergent conserved charge, 
$\widetilde d(\gamma) = 0$, by assumption, 
the connecting map $\delta_1$ is a zero map in the long exact sequence. This proves Eq.~(\ref{eq;ker-coker-mapping}).

Let us proceed to the latter part of the claim. 
The steady-state condition of $\Gamma$ is written as 
\begin{align}
  S \bm r (\bm x) &= 0, \label{eq:sr-0} \\  
  \bm d^{\bar\alpha} \cdot \bm x &= \ell^{\bar\alpha} . \label{eq:dx-l}
\end{align}
As usual, we divide the degrees of freedom to those in $\gamma$ and $\Gamma \setminus \gamma$. Then Eq.~(\ref{eq:sr-0}) is written as 
\begin{equation}
  \begin{pmatrix}
    S_{11} & S_{12} \\
    S_{21} & S_{22} 
  \end{pmatrix}
  \begin{pmatrix}
    \bm r_1 (\bm x_1, \bm x_2)\\
    \bm r_2 (\bm x_2 )
  \end{pmatrix}
  = \bm 0 . 
\end{equation}
The reactions $\bm r_2 (\bm x_2)$ depend only on $\bm x_2$, 
because $\gamma$ is chosen to be output-complete. 
The first equation 
can be solved for $\bm r_1$ as 
$
  \bm r_1 = - S_{11}^+ S_{12} \bm r_2 + \bm c_{11}, 
$ 
with  $\bm c_{11} \in \ker S_{11}$, 
and we have 
\begin{equation}
S' \bm r_2 (\bm x_2 ) 
 = - S_{21} \bm c_{11} 
 = \bm 0 , 
 \label{eq:sp-r2-0}
\end{equation}
where the last equality is due to $\widetilde c(\gamma) = 0$, 
that is equivalent to $\ker S_{11} \subset \ker S_{21}$.

Let us turn to the conserved charges. 
Recall that $d_l (\gamma)$ is written as 
$d_l (\gamma) = |(\coker S) / X(\gamma)|$. 
Because of the decomposition \eqref{eq:cokers-decom}, 
when $d_l (\gamma) = 0$, 
the space $\coker S$ is written as the direct sum 
of $D(\gamma)$ and $\bar D'(\gamma)$, 
\begin{equation}
     \coker S 
     \cong 
     D(\gamma) \oplus 
   (\coker S) / X(\gamma) \oplus \bar D' (\gamma) 
     \cong D(\gamma) \oplus \bar D' (\gamma). 
\end{equation}
Correspondingly, we can divide the basis vectors of $\coker S$ 
into two classes, 
$\{ \bm d^{\bar\alpha} \} 
= 
\{
   \bm d^{\bar\alpha_\gamma}, \bm d^{\bar \alpha' } 
\}$, 
where 
$\{ \bm d^{\bar\alpha_\gamma}\}$ 
is a basis of $D(\gamma)$, 
and 
$\{ \bm d^{\bar \alpha'} \}$ 
is a basis of $\bar D'(\gamma)$. 
The basis vectors are of the form, 
\begin{equation}
  \bm d^{\bar\alpha_\gamma}
  = 
  \begin{pmatrix}
    \bm d^{\bar\alpha_\gamma}_1 
    \\
    \bm d^{\bar\alpha_\gamma}_2
  \end{pmatrix}
  \text{ with } 
  \bm d^{\bar\alpha_\gamma}_1 \neq \bm 0, 
  \quad 
  \bm d^{\bar\alpha'}
  = 
  \begin{pmatrix}
    \bm 0 \\
    \bm d^{\bar\alpha'}_2
  \end{pmatrix}. 
\end{equation}
With this basis of $\coker S$, 
Eq.~(\ref{eq:dx-l}) is written as 
\begin{align}  
 \bm d_1^{\bar\alpha_\gamma} \cdot \bm x_1 
  + 
  \bm d_2^{\bar\alpha_\gamma} \cdot \bm x_2 
  &= \ell^{\bar\alpha_\gamma}, 
\\
\bm d_2^{\bar\alpha'} \cdot \bm x_2 &= \ell^{\bar\alpha'} . 
\label{eq:d2x2-l}
\end{align}

In fact, 
$\bm d^{\bar\alpha'}_2$ is a conserved charge in $\Gamma'$, 
$\bm d^{\bar\alpha'}_2 \in \coker S'$, 
as we see in the following. 
Since $\bm d^{\bar\alpha'} \in \coker S$ , 
it satisfies 
\begin{equation}
  \begin{pmatrix}
    \bm 0 & (\bm d^{\bar\alpha'}_2)^T
  \end{pmatrix}
\begin{pmatrix}
  S_{11} & S_{12} \\
  S_{21} & S_{22}
\end{pmatrix}
= 
\begin{pmatrix}
  (\bm d^{\bar\alpha'}_2)^T S_{21} & (\bm d^{\bar\alpha'}_2)^T S_{22}
\end{pmatrix}
= \bm 0 . 
\end{equation}
This implies that $\bm d^{\bar\alpha'}_2$ satisfies 
\begin{equation}
  (\bm d^{\bar\alpha'}_2)^T S' 
  = (\bm d^{\bar\alpha'}_2)^T ( S_{22} - S_{21} S^{+}_{11} S_{12}) 
  = \bm 0 , 
\end{equation}
hence $\bm d^{\bar\alpha'}_2 \in \coker S'$. 
Thus we have obtained an injective map, 
\begin{equation}
  \coker S \ni 
  \bm d^{\bar\alpha'} = 
  \begin{pmatrix}    
    \bm 0 \\
  \bm d^{\bar\alpha'}_2 
\end{pmatrix}
\mapsto 
\bm d^{\bar\alpha'}_2  \in \coker S'. 
  \label{eq:sd-ds}
\end{equation}
This map is nothing but the induced map $\bar \varphi_0$. 
It is important to note that, 
when $\widetilde c(\gamma)=0$, this map is a surjection, 
that is evident from the long exact sequence (\ref{eq:les})\footnote{
Another way to see this is by Eq.~(\ref{eq:c-tc=td-tdp}). 
We have $\widetilde c (\gamma) = \widetilde d (\gamma)=0$ from the assumption, 
and $\widetilde c' (\gamma)= \widetilde d (\gamma)$ holds by the connecting map $\delta_1$. 
Thus, we have $\widetilde d'(\gamma)=0$, which means that there is no emergent conserved charge in $\Gamma'$. 
}. 
The equations satisfied by 
the boundary part (denoted by 2) of the concentrations/rates of $\Gamma$ 
are Eqs.~(\ref{eq:sp-r2-0}) and (\ref{eq:d2x2-l}). 
%
Since all the conserved charges in $\Gamma'$ 
is given as a image $\bar \varphi_0$, 
we find that the set of Eqs.~(\ref{eq:sp-r2-0}) and  (\ref{eq:d2x2-l}) 
are exactly the same as 
the steady-state condition for the reduced network $\Gamma'$, 
\begin{align}
  S' \bm r' (\bm x') &= \bm 0,  \\
   \bm d_2^{\bar\alpha'} \cdot \bm x'  &= \ell^{\bar\alpha'} . 
\end{align}
where $\bm x' = \bm x_2$ and $\bm r' (\bm x') = \bm r_2 (\bm x_2)$. 
Thus, the steady-state solution of $\Gamma'$ should also be 
the steady-state solution of $\Gamma$ for the boundary degrees of freedom. 
This concludes the proof. 
\end{proof}

\subsection{Hierarchy of subnetworks}\label{sec:order}

Let us consider nested subnetworks $\gamma' \subset \gamma \subset \Gamma$. 
Given the stoichiometric matrix $S$ of the whole network, 
we denote 
the stoichiometric matrices of the subnetworks $\gamma$ 
and $\gamma'$ by $S_\gamma$ and $S_{\gamma'}$, respectively. 
The submatrices are included in the following form, 
\begin{equation}
  S = 
  \begin{pmatrix}
    S_\gamma & \ast \\
    \ast & \ast 
  \end{pmatrix}
, 
\quad   
  S_\gamma = 
  \begin{pmatrix}
    S_{\gamma'} & \ast \\
    \ast & \ast 
  \end{pmatrix},  
\end{equation}
where $\ast$ indicates an arbitrary matrix. 
Let us consider the situation where 
$\gamma$ and $\gamma'$ 
has no emergent cycle and emergent conserved charge in $\Gamma$,
namely $\widetilde c(\gamma)=\widetilde d (\gamma) = 0$, 
and $\widetilde c(\gamma') = \widetilde d (\gamma') = 0$. 
Under those assumptions, 
the quotient formula of the generalized Schur complement  \cite{carlson1974generalization} holds, 
\begin{equation}
  S/S_{\gamma} 
  =( S / S_{\gamma'}) / ( S_\gamma / S_{\gamma'}) . 
\end{equation}
This indicates the isomorphisms of homology groups, 
\begin{equation}
  H_n (\Gamma/\gamma)
  \cong 
  H_n ( (\Gamma/\gamma')/(\gamma/\gamma') ), 
  \label{eq:excision}
\end{equation}
for $n=0,1$. 
Thus, when we perform the reductions of nested subnetworks 
that have no emergent cycles and emergent conserved charges, 
the order of the reduction of them does not matter.

\section{Example of reduction: metabolic pathway of {\it E.~coli}}\label{sec:ex-ecoli}

As an application of the reduction method, 
let us examine the central metabolism of {\it E.~coli}.  
We use the stoichiometric matrix presented in Ref.~\cite{doi:10.1002/mma.3436}, which is constructed based on Ref.~\cite{ishii2007multiple} with minor modifications.  The network structure is shown in Fig.~\ref{fig:ecoli-full}, which consists of the glycolysis, the pentose phosphate pathway (PPP), and the tricarboxylic acid cycle (TCAC).
The list of the reactions for this system 
is given in Appendix~\ref{sec:ecoli-list}. 
Here, we assume that 
H\textsubscript{2}O and cofactors such as ATP and NADH are abundant and do not affect the behavior of the system. 
Buffering structures in this network have been identified 
in Ref.~\cite{PhysRevLett.117.048101}
and there are in total 17 buffering structures,
which we list in Appendix~\ref{sec:ecoli-buffering-list}. 
As we showed in Sec.~\ref{sec:submodularity}, 
the intersections or unions of 
buffering structures are also buffering structures. 
They form a hierarchy, 
and such an architecture can be regarded 
as a source of robustness against perturbations, since buffering structures work as a kind of firewalls.

Let us now perform reductions of buffering structures, 
under which the steady state is ensured to be the same as the original network as we showed in Sec.~\ref{sec:red-buff-th}. 
We denote the whole network by $\Gamma$. 
We can pick a buffering structure $\gamma_8$, which is a part of the pentose phosphate pathway (the yellow subnetwork in Fig.~\ref{fig:ecoli-full}) and given by\footnote{
In fact, $\gamma_8 = \gamma_5 \cap \gamma_{14}$, 
and if we allow taking intersections of buffering structures, 
$\gamma_8$ is redundant. 
This is consistent with Corollary~\ref{cor:intersection-union}. 
}
\begin{equation}
  \gamma_8 = 
  (
    \{
   \text{X5P, S7P, E4P}   
    \},
    \{ 17,18 ,19 ,20 ,21  \}
  ), 
\end{equation} 
and perform a reduction to obtain $\Gamma_1 \coloneqq \Gamma /\gamma_8$. 
The stoichiometric matrix of the reduced reaction network 
can be computed by Eq.~(\ref{eq:def-sp}). 
The resulting network is shown in Fig.~\ref{fig:ecoli-reduced-1}. 
The reduction procedure induces rewiring of the reactions, which are colored in magenta in Fig.~\ref{fig:ecoli}. 
Reactions 15 and 22 are rewired, and 22 is now a degenerate reaction. 
The fraction $1/2$ shown at reaction 15 indicates the weight 
of the species. 
Those reconnections including the change of weights 
are {\it necessary} 
if we want the steady state to be the same  
as those of the original network. 
Otherwise, the steady state is changed in general. 
We can proceed further and reduce the subnetwork $( \{\text{G3P,R5P}\},\{7,22,40\})$ (colored in red and orange in Fig.~\ref{fig:ecoli-full}).  This reduction is the same as reducing $\gamma_5 \cup \gamma_{14}$ from $\Gamma$. 
The result of the reduction is shown in Fig.~\ref{fig:ecoli-reduced-2}. 
Again, rewiring occurs and 
the reactions 5,6,15, and 36 are modified from the original system. 
Finally, let us focus on the part colored in green in Fig.~\ref{fig:ecoli-full}, 
which consists of the following subsets of chemical species and reactions, 
\begin{equation}
( 
  \{ \text{G6P}, \text{F6P}, \text{F16P}, \text{6PG}, \text{Ru5P}, \text{DHAP} \},  
\{2, 3, 4, 5, 6, 13, 14, 15, 16, 36, 43\} 
 ) . 
 \label{eq:green}
\end{equation}
The complement of the subset (\ref{eq:green}) 
is given by 
$\gamma_5 \cup \gamma_7 \cup \gamma_{14}$,
which hence is a buffering structure, 
and a reduction can be performed. 
The structure of the reduced network 
$\Gamma_3 = \Gamma / (\gamma_5 \cup \gamma_7 \cup \gamma_{14})$ 
is given in Fig.~\ref{fig:ecoli-reduced-glycolysis}. 
Compared to the original network, 
we notice that the reactions 15 and 43 are rewired. 

To demonstrate our theoretical prediction, 
we numerically solve the rate equations for 
the four systems in Fig.~\ref{fig:ecoli} (the original network $\Gamma$ and the reduced ones, $\Gamma_1,\,\Gamma_2,\,\Gamma_3$), 
using the same initial condition and reaction rate constants in all of the four cases (see Appendix~\ref{sec:parameters} for details of parameter values).
The time-series of concentrations are presented in Fig.~\ref{fig:dyn}. 
After the initial transient dynamics, the original system   approaches a (stable) steady state [Fig.~\ref{fig:dyn}(a)]. We can see that the reduced systems can reproduce the steady-state concentrations that the original system eventually reaches,  although they have distinct short-time dynamics [Fig.~\ref{fig:dyn}(b-d)].

In this way,  buffering structures work as a guide 
as to how to perform the reduction 
and simplify a complex reaction network. 
As long as the reduced part is a buffering structure 
and 
we use the generalized Schur complement (\ref{eq:def-sp})
as a stoichiometric matrix of a reduced network, 
the steady-state concentrations and rates of the remaining part 
stay the same as the original ones
regardless of the details of the kinetics, 
as a consequence of Theorem \ref{th:reduction}.


\begin{figure}[htb]
\subfloat[
    \label{fig:ecoli-full} 
]{%
\includegraphics[clip, trim=5cm 0cm 5cm 0cm, width=0.45\textwidth]
   {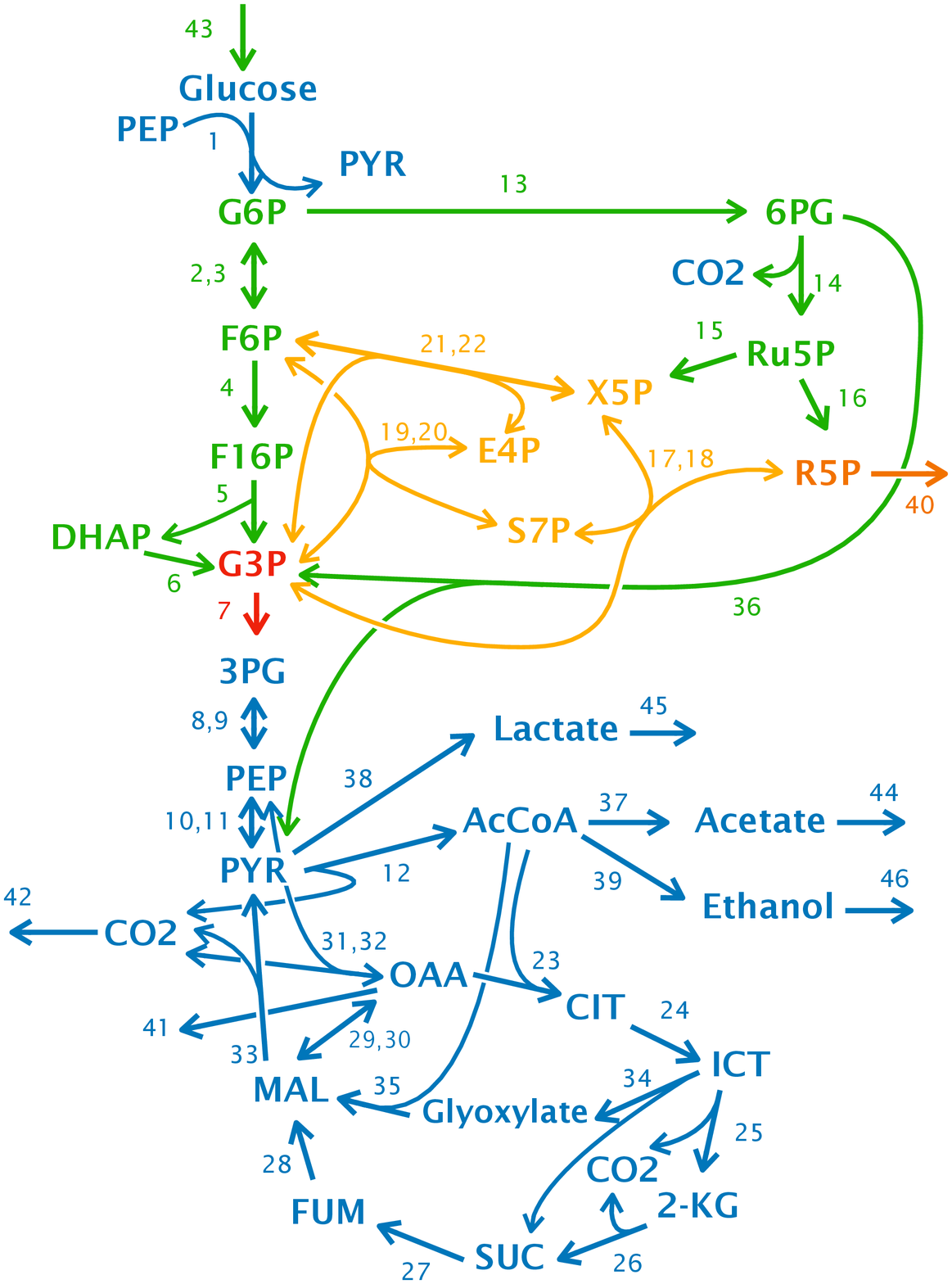}
}
\subfloat[ \label{fig:ecoli-reduced-1}]{%
\includegraphics[clip, trim=5cm 0cm 5cm 0cm, width=0.45\textwidth]
   {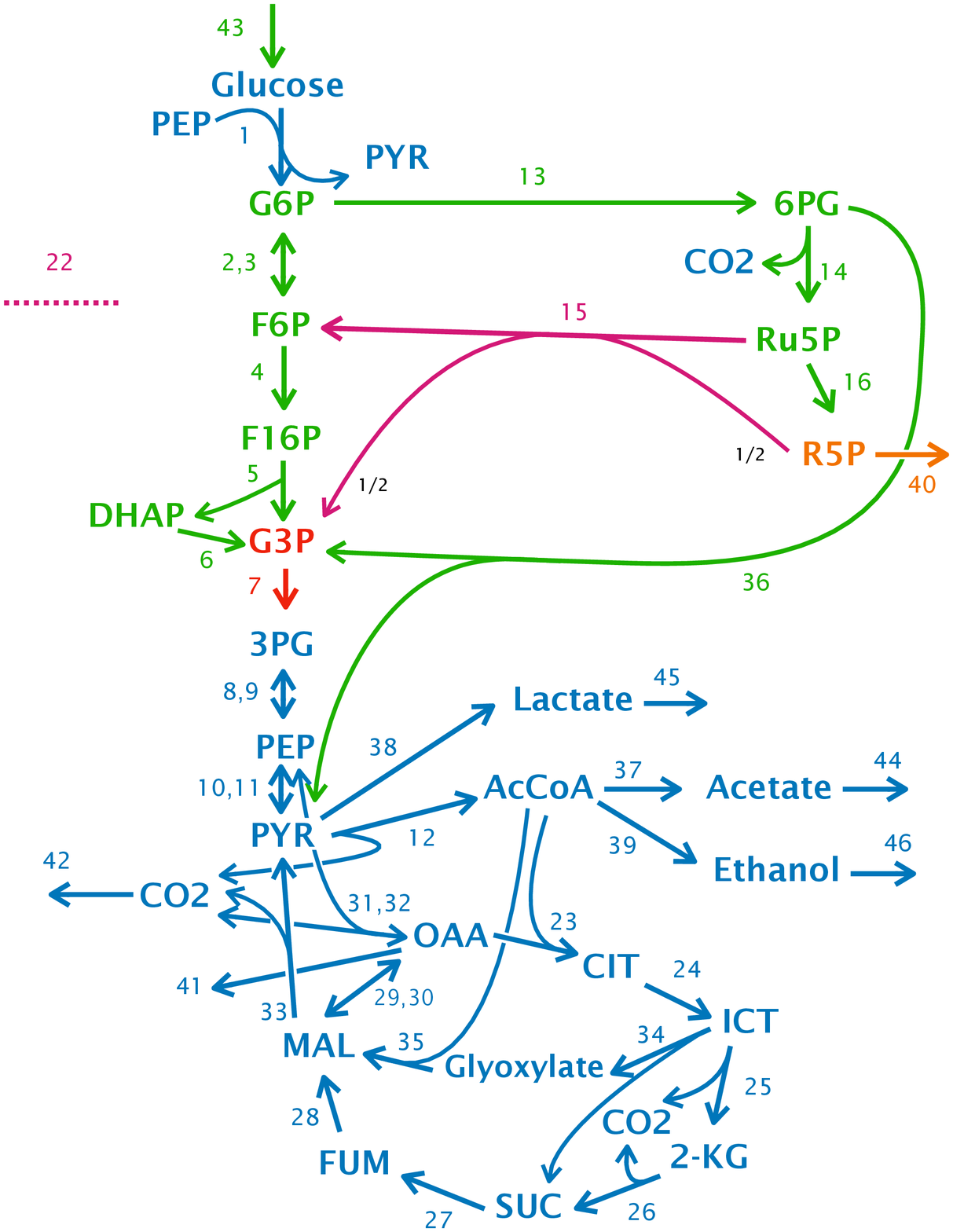}
} \\
\subfloat[ \label{fig:ecoli-reduced-2}]{%
\includegraphics[clip, trim=5cm 0cm 5cm 0cm, width=0.45\textwidth]
   {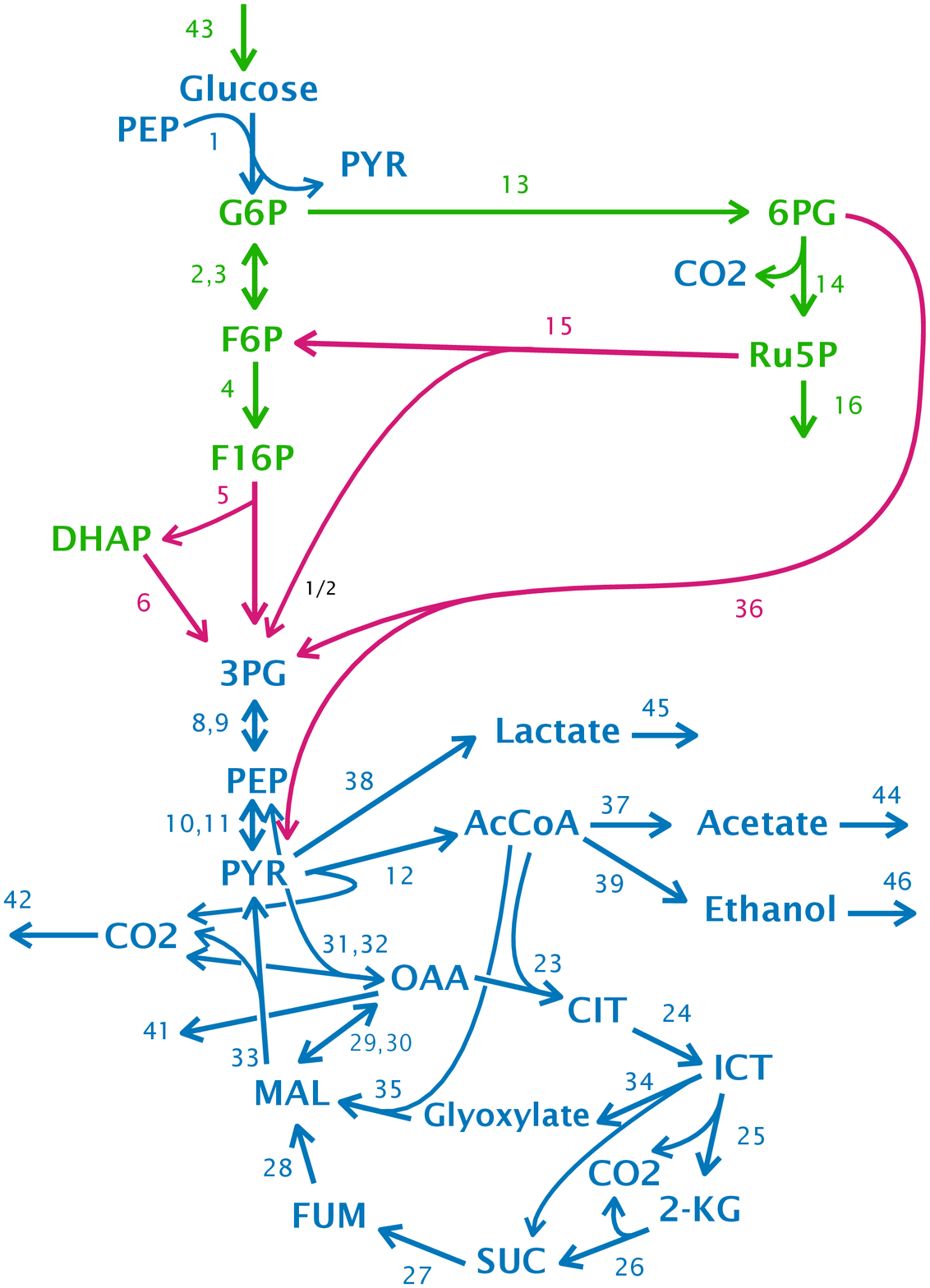}
} 
\subfloat[ \label{fig:ecoli-reduced-glycolysis} ]{%
\includegraphics[clip, trim=5cm 0cm 5cm 0cm, width=0.45\textwidth]
   {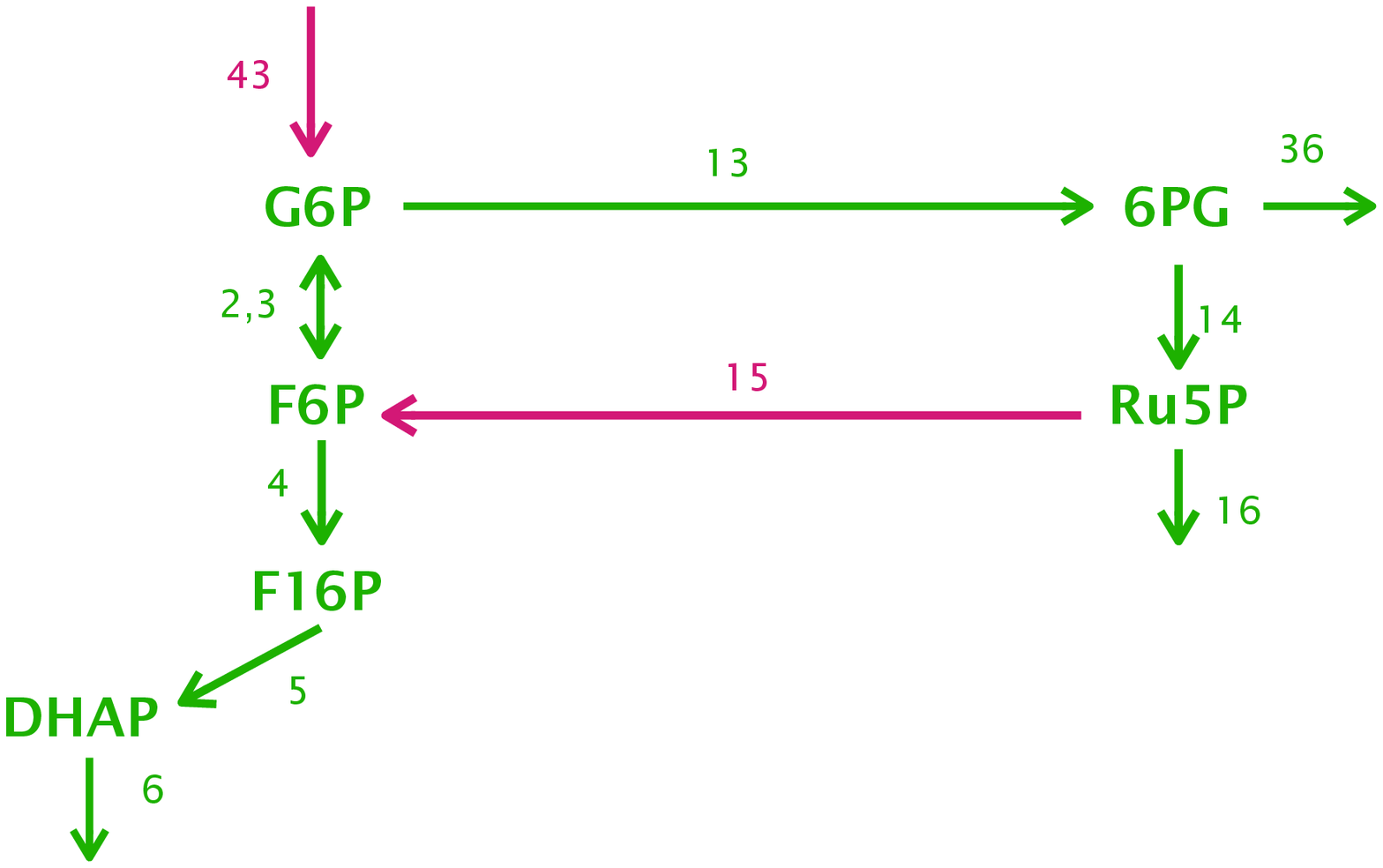}
}
\caption{
  Central metabolic pathway (a) of {\it E.~coli} 
  and reduced networks (b, c, d). 
  In the reduced networks, rewired reactions under the reductions are colored in magenta. 
  (b) Reduced network $\Gamma_1 = \Gamma/\gamma_8$, where $\gamma_8$ is colored in yellow in (a).  
    The fraction $1/2$ written in black indicates the weight in the stoichiometric matrix.
(c) $\Gamma_2 = \Gamma / (\gamma_5 \cup \gamma_{14} )$, where $\gamma_5 \cup \gamma_{14}$ is colored in yellow, red, and orange in (a).
(d) $\Gamma_3 = \Gamma/(\gamma_5 \cup \gamma_7 \cup \gamma_{14})$, where $\gamma_5 \cup \gamma_7 \cup \gamma_{14}$ is colored in  yellow, red, orange, and blue in (a).
  }
  \label{fig:ecoli} 
\end{figure}

\begin{figure}[htb]
  \centering
  \includegraphics[keepaspectratio, scale=0.7]{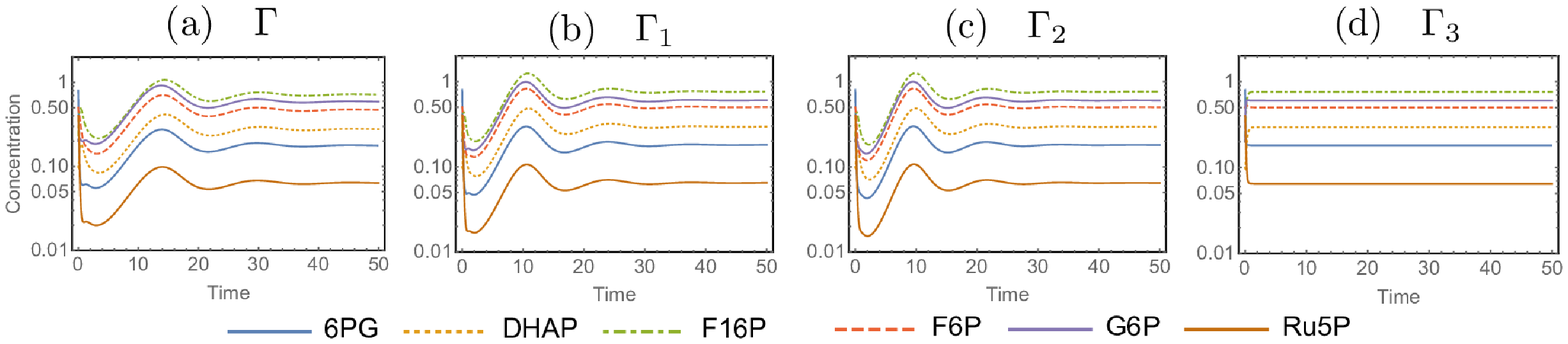} 
  \caption{Time-series of concentrations of DHAP, F16P, F6P, G6P, PG6, Ru5P computed by solving the whole system $\Gamma$ (a) and the reduced systems $\Gamma_1,\, \Gamma_2,\, \Gamma_3$ (b, c, d).  The same initial condition and reaction rate constants are used in the four cases. 
  }
  \label{fig:dyn}
 \end{figure}

\section{Summary and outlook}\label{sec:summary} 

The main focus of the present paper was 
the relationship between the structure and functions of the chemical reaction network. 
As a characterization of the structure, 
homology and cohomology groups for chemical reaction networks
were introduced, in which the actions of boundary and coboundary operators 
are determined by the stoichiometry. 
The elements of homology groups correspond to cycles and conserved charges of chemical reaction networks, 
and steady states were shown to be determined by the elements of the cohomology groups. 
In a similar way to the homology and cohomology groups of topological spaces, 
the Mayer-Vietoris sequence 
and the long exact sequence of a pair of chemical 
reaction networks were introduced,
the latter being particularly useful for studying the reduction of reaction networks.

We propose a method of reduction of chemical reaction networks. 
The reduced network is characterized by 
the stoichiometric matrix 
obtained by eliminating the chemical species and reactions 
of an output-complete subnetwork
via the Schur complementation. 
The reduction relies only on the stoichiometry, which determines the topology of the reaction networks, 
and thus {\it is applicable to any kind of kinetics}. 
This represents an advantage since in many biological systems it is difficult to experimentally determine the kinetics and parameters of the reactions. 
For tracking the change of cycles and conserved charges under the reductions, 
the tools of algebraic topology, such as the long exact sequence, have been useful. 
We have studied how the law of localization can be understood from 
this perspective. 
We showed that the influence index is expressed 
in terms of the numbers of cycles/conserved charges of particular types, as in Eq.~\eqref{eq:lambda-decom}. 
We also showed that the influence index is a submodular function over output-complete subnetworks. 
A corollary of this is that 
buffering structures are closed under
intersection and union, 
which is useful when we enumerate the buffering structures 
of a large reaction network. 
As a central result of the paper, 
we showed that buffering structures, 
which are subnetworks with vanishing influence index, 
behave nicely under the reduction. 
Namely, under the reduction of a buffering structure,
the steady state of the remaining elements of the network 
stays the same as the original network (Theorem~\ref{th:reduction}). 
The theorem justifies the intuition that 
buffering structures are regarded as `irrelevant' substructures:  
they can be safely eliminated 
through the reduction method proposed here 
without changing the long-time behavior of the system. 
The reduction procedure introduces rewiring of reactions,
which is necessary so that the steady state is not modified under the reduction.
As an application of the reduction method, 
we discussed the reduction of the central metabolic pathway of {\it E. coli} and illustrated that reactions are rewired non-trivially under the reduction. 
We also demonstrated the invariance of the steady state under the reduction of buffering structures by numerically solving the rate equations before and after the reduction\footnote{
We remark that, in our analysis of the central carbon metabolism, cofactors are not included as variables
on the assumption that they are abundant
and their concentrations are stable. 
If this is not the case, the identifications of buffering structures will be modified. 
The applicability of such assumptions should be examined 
depending on the situations one wants to consider. 
}.

Our results highlight that special care should be taken 
when simplifying a reaction network.
A naive elimination of a subnetwork not of interest would alter steady-state properties of the original system. 
 As long as the subnetwork has the vanishing influence index and reactions are rewired appropriately using the generalized Schur complement, it can be eliminated while keeping the steady state intact. 

Another significance of our method is that it allows us to identify the modules in a complex network and facilitates the biological interpretation of the whole system.
For example, the central metabolic pathway of {\it E.~coli} consists of three modules; glycolysis, TCAC, and PPP.  Interestingly, the reduced network in Fig.~\ref{fig:ecoli-reduced-glycolysis} roughly corresponds to the glycolysis. The fact of glycolysis being a reduced network may suggest that  E. coli can control the glycolysis in an isolated manner, and the expression levels of enzymes in the TCAC and the PPP do not affect the physiological states of the glycolysis.

For practical applications, 
one important issue is how to find the buffering structures efficiently in large-scale reaction networks. 
Although we defer this as a future problem, let us make some comments on this point. 
One practical way of finding  buffering structures is as follows: We first compute the sensitivity matrix $A^{-1}$ by assigning random values to $\frac{\partial r_A}{\partial x_i}$. 
From this, we can identify, for each parameter $k_A$ (and for each conserved concentration $l^{\bar \alpha}$ if exists), the subset $V_A$ of chemicals that show nonzero responses to the perturbation of $k_A$ under generic kinetics. 
The inclusion relation among  $V_A$'s indicates candidate buffering structures (see Figs.~3 and 5 in \cite{PhysRevLett.117.048101} for the illustrations). 
For example, $V_A \subset V_B$ indicates the existence of two nested buffering structures. Finally, for those candidates, we can compute the influence index and verify if they are indeed  buffering structures.

Establishing a combinatorial method for identifying buffering structures is an amusing problem. 
We believe that the basic properties of buffering structures 
that we showed in this paper would be useful for this purpose. 
For example, if a network contains many small buffering structures, we can use the reduction method repeatedly
and make the network smaller one we fine a small buffering structure. 
This procedure is possible because the order of reduction 
does not matter for the buffering structures, 
as we showed in Sec.~\ref{sec:order}. 
The submodular property of the influence index and 
the subsequent closure property of buffering structures 
under unions/intersections 
would also be useful in enumerating buffering structures.

We believe that the mathematical formulation that we used 
to characterize the topology of chemical reaction networks 
will be useful for understanding the static and dynamical properties\footnote{ 
In Ref.~\cite{cardelli2014morphisms}, morphisms of chemical reaction networks are considered and 
a condition is given as to when 
a reaction network can dynamically emulate another one. 
} of reaction systems. 
The eigenvalues of the Laplacian operators 
entail the information of the topology of the network connectivity. 
Steady states correspond to the eigenvectors with zero eigenvalues 
and they incorporate the crudest topological information of the reaction network. 
The eigenvectors with higher eigenvalues are going to be needed 
if we want to extend the reduction method to approximate the dynamics as well as the steady states.

\begin{acknowledgments}
This work is in part supported by RIKEN iTHEMS Program.
Y.~Hirono is supported by the Korean Ministry of Education, Science and Technology, Gyeongsangbuk-do and Pohang City at the Asia Pacific Center for Theoretical Physics (APCTP) and by the National Research Foundation (NRF) funded by the Ministry of Science of Korea (Grant No. 2020R1F1A1076267).
H.~Miyazaki is supported by JSPS KAKENHI Grant (19K23413) and JST CREST Grant Number JPMJCR1913.
Y.~Hidaka is supported by JSPS KAKENHI Grant Numbers~17H06462.
The authors thank Atsushi Mochizuki, Tetsuo Hatsuda, Hideaki Aoyama, Yuichi Ikeda, and Genki Ouchi for useful discussions and comments. 
Y.~Hirono is grateful to Benjamin for helpful discussions and the constant encouragement. 
\end{acknowledgments}

\appendix

\section{Laplace operators and Hodge decomposition}\label{sec:app-hodge} 

In this section, we discuss the Hodge decomposition 
and Laplace operators, which are closely related 
to the cohomology groups introduced in the main text. 

We can define Laplace operators,  
$\Delta_n: C^n(\Gamma) \to C^n(\Gamma)$, 
as 
\begin{equation}
  \Delta_0 \coloneqq \rd^\dag_0 \rd_0, 
  \quad 
  \Delta_1 \coloneqq \rd_0 \rd^\dag_0. 
\end{equation}
Recall that the coboundary operator~\eqref{eq:coboundary} and its adjoint~\eqref{eq:ajoint coboundary} are 
given by 
$(\rd_0 a_0)(e_A) =\sum_i (S^T)_{Ai} a_0(v_i)$ for $a_0 \in C^0(\Gamma)$
  and $(\rd^\dag_0 a_1)(v_i) = \sum_A S_{iA} a_1(e_A)$ 
  for $a_1 \in C^1(\Gamma)$. 
The action of the Laplacians are written 
in the matrix form as 
\begin{equation}
  (\Delta_0 a_0) (v_i)  =\sum_j (S S^T)_{ij} \, a_0(v_j) ,  
  \quad 
  (\Delta_1 a_1)(e_A)  =\sum_{B} (S^T S)_{AB}\, a_1(e_B), 
\end{equation}
for $a_0 \in C^0 (\Gamma)$ and $a_1 \in C^1 (\Gamma)$. 
Those are generalizations of the graph Laplacian 
to hypergraphs. 
The properties of hypergraph Laplacians 
were discussed recently 
in Refs.~\cite{JOST2019870,MULAS2021112372, 10.1093/comnet/cnab009}. 
When all the reactions are monomolecular, 
the Laplacian reduces to the graph Laplacian of the directed graph.

The space $C^1(\Gamma)$ admits 
the following orthogonal decomposition, 
\begin{equation}
    C^1 (\Gamma) = {\rm im\,} \rd_0 
    \oplus \ker \Delta_1 . 
\end{equation}
This is a natural generalization of the Hodge decomposition 
of flows on networks \cite{jiang2011statistical}
to the case of a hypergraph. 
Thus, given a 1-cochain $f \in C^1(\Gamma)$, 
we can decompose it in a unique way as 
\begin{equation}
  f = \rd_0 a + c , 
  \label{eq:f}
\end{equation}
where $c \in \ker \rd_0^\dag \,\cap\, \ker \rd_1$
is a harmonic cochain 
and $a \in C^0(\Gamma)$. 
This is the Hodge decomposition associated with the complex 
(\ref{eq:gamma-comples}). 
By acting $\rd^\dag_0$ on Eq.~(\ref{eq:f}), 
we have 
\begin{equation}
  \rd^\dag_0 f = \rd^\dag_0 \rd_0 a = \Delta_0 a. 
\end{equation}
We can solve this for the potential $a$ as 
\begin{equation}
  a = \Delta_0^+ \rd^\dag_0 f + a_0 . 
\end{equation}
Here, $\Delta_0^+:C^0(\Gamma) \to C^0(\Gamma)$ is 
the operator defined by 
$(\Delta^+_0 b_0)(v_i) \coloneqq\sum_j (S S^T)^+_{ij} b_0(v_j)$
for $b_0 \in C^0(\Gamma)$, 
where $M^+$ indicates the Moore-Penrose inverse of a matrix $M$, 
and $a_0 \in \ker \Delta_0$. 
The harmonic component $c$ can be obtained by 
\begin{equation}
  c = f- \rd_0 a = (1 - \rd_0 \Delta^+_0 \rd^\dag_0) f . 
  \label{eq:c-f-da}
\end{equation}
Using the properties of the Moore-Penrose inverse, 
the action of the operator that appears 
on the RHS of Eq.~(\ref{eq:c-f-da}) is written as 
\begin{equation}
 [(1 - \rd_0 \Delta^+_0 \rd^\dag_0) b_1] (e_A) 
 =\sum_B (1 - S^+ S)_{AB} \, b_1 (e_B) , 
\end{equation}
for an arbitrary $b_1 \in C^1(\Gamma)$. 
The matrix $1 - S^+ S$ is the projection matrix to 
$\ker S$. 
Thus, the harmonic component can be identified by 
the projection to $\ker S$, 
\begin{equation}
c(e_A) =\sum_B ( 1 - S^+ S )_{AB} f(e_B) . 
\end{equation}
This is consistent with the fact that $c \in H^1(\Gamma) = \ker S$. 
The potential $a$ can be obtained 
by the multiplication of the Moore-Penrose inverse of $S$
to $f$, 
\begin{equation}
  a (v_i) =\sum_A (S^T)^+_{iA} f (e_A) 
  + a_0 (v_i) , 
\end{equation}
where $a_0 \in \ker \Delta_0$.

\section{Cycles and conserved charges}\label{sec:app-cycles} 

\subsection{Interpretation of $\widetilde c(\gamma)$ and $d_l (\gamma)$, 
  and $\widetilde d(\gamma)$}\label{sec:app-interpretations}  

Let us here discuss intuitive interpretations 
of the integers appearing 
in the decomposition (\ref{eq:lambda-decom}). 
We will refer to the elements of $\ker S$ as ``global cycles''
and those of $\coker S$ as ``global conserved charges.'' 
Similarly, 
the elements $\ker S_{11}$ and 
$\coker S_{11}$ 
are referred to as ``local cycles'' and ``local conserved charges.''
With this terminology, 
the elements of $(\ker S)_{{\rm supp}\, \gamma}$ 
are called as ``locally supported global cycles.''

We first look at $\widetilde c(\gamma)$. We denote the space by 
\begin{equation}
{\rm \widetilde C}(\gamma) 
  \coloneqq
  \ker S_{11} /  
  (\ker S)_{{\rm supp}\, \gamma}. 
  \label{eq:space-gamma-c}
\end{equation}
Then, $\widetilde c(\gamma) = |{\rm \widetilde C}(\gamma)|$. 
To clarify its meaning, 
we represent the space (\ref{eq:space-gamma-c}) as follows, 
\begin{equation}
  \begin{split}
 {\rm \widetilde C}(\gamma)  
 &= 
  \{ 
  \bm c \in C_1(\Gamma)
  \,|\,
  S \bm c = \bar P^0_{\gamma} \bm v
  ,\, 
  P^1_\gamma  \bm c = \bm c
  \}
  /
  \{ 
  \bm c 
  \in C_1(\Gamma) 
  \,|\,
  S \bm c = 0, \, 
  P^1_\gamma  \bm c = \bm c 
  \}   
  \\
  &=
  \{ 
    \bm c \in C_1(\Gamma)
    \,|\,
    S \bm c = \bm v \neq \bm 0, \, 
  \bm v \in 
  C_0 (\Gamma \setminus \gamma)
    ,\, 
    P^1_\gamma  \bm c = \bm c 
    \} . 
  \end{split}
\end{equation}
An element of $ {\rm \widetilde C}(\gamma)$ 
is a local cycle that is not a global cycle, 
by which we mean that 
$\bm c \in  {\rm \widetilde C}(\gamma)$  
has its boundary in $\Gamma \setminus \gamma$. 
For example, let us take 
the subnetwork $\gamma = ( \{ v_2 \}, \{ e_1,e_2 \} )$ 
of a monomolecular reaction network, 
\begin{equation}
  \begin{tikzpicture}
  
    \node[species] (x) at (0,0) {$v_1$}; 
    \node[species] (y) at (1.5,0) {$v_2$}; 
    \node[species] (z) at (3,0) {$v_3$}; 
    \node[species] (w) at (4.5,0) {$v_4$}; 
    
    \draw [-latex,draw,line width=0.5mm] (x) edge node[below]{$e_1$} (y);
    \draw [-latex,draw,line width=0.5mm] (y) edge node[below]{$e_2$} (z);
    \draw [-latex,draw,line width=0.5mm] (z) edge node[below]{$e_3$} (w);

   \draw[mydarkred, dashed,line width=1] (0.3,-1) rectangle (2.7,1);
   \node[left] at (0.8, 0.7) {
    \scalebox{1.1} {\color{mydarkred} $\gamma$} };
  \end{tikzpicture} 
  \label{eq:network-1} 
\end{equation}
Although $e_1 + e_2 \in C_1 (\Gamma)$ has its support in $\gamma$, 
its boundary, \footnote{
Recall that the boundary of each reaction is specified by 
the stoichiometric matrix as 
$
  \p_1 e_A =\sum_i (S^T)_{Ai} v_i . 
$
}
\begin{equation}
  \p_1 (e_1 + e_2) = -v_1 + v_3 , 
\end{equation}
is outside of $\gamma$. 
The element $e_1+e_2$ is a local cycle, 
since $\p_1 (e_1+e_2)$ is zero as a relative chain in 
$C_0(\gamma) = C_0(\Gamma)/C_0(\Gamma \setminus \gamma)$. 
Note that the network (\ref{eq:network-1}) as a whole does not have a cycle and $\ker S=\bm 0$. 
Thus, we can identify 
$\bm c \in  {\rm \widetilde C}(\gamma)$ 
to be a local cycle 
whose boundary is out of $\gamma$. 
When $\bm c$ is viewed in $\Gamma$, 
it may be extended to a global cycle,
but it does not have to be. 
Considering its meaning, 
we will refer to the elements of ${\rm \widetilde C}(\gamma)$
as {\it emergent cycles}, which only appear when we focus on a subnetwork.

Let us illustrate the space $ {\rm \widetilde C}(\gamma)$ pictorially. 
The matrix $S$ works as a boundary operator on 
the space of chemical reactions.  
Thus, the kernel of $S$ are 
linear combinations of reactions without boundaries. 
Cycles and non-cycles can be drawn pictorially as 
\begin{center}
  \begin{tikzpicture}[scale=0.8]
  
    \draw[color=mydarkred, line width=1.8]
    (3, 0) -- (2.5,3.6);


     \draw[color=mydarkred, line width=1.8]
     (1, 0) -- (0,1);    

    \draw[color=mydarkred, line width=1.8]
    (1,2.2) ellipse (0.4 and 0.6); 
    
    \draw[color=myblue, line width=1.8]
    (4.5, 0) -- (4.5,1.6); 
    \filldraw[color=myblue, fill=myblue, line width=1.4]
    (4.5,1.6) circle (0.08);

    \draw[color=myblue, line width=1.8]
    (3.9, 2.2) -- (5,3.2); 
    \filldraw[color=myblue, fill=myblue, line width=1.4]
    (3.9,2.2) circle (0.08);
    \filldraw[color=myblue, fill=myblue, line width=1.4]
    (5,3.2) circle (0.08);

    \draw[black, line width=2] (0,0) rectangle (6,3.6);
  
    \draw[->,>=stealth,color=black, line width=1]
     (-1,3) -- (2.4,3);

     \draw[->,>=stealth,color=black, line width=1]
     (-1,2.9) -- (0.5,2.6);

     \draw[->,>=stealth,color=black, line width=1]
     (-1,2.8) -- (-0.1,1);

     \node[left] at (-1, 3) {
      \scalebox{1.1} {Cycles} };
  
      \draw[->,>=stealth,color=black, line width=1]
      (6.5,1.5) -- (4.7,1); 
      \draw[->,>=stealth,color=black, line width=1]
      (6.5,1.6) -- (4.7,2.6); 
 
      \node[left] at (9, 1.5) {
        \scalebox{1.05}
        {Non-cycles}
        };
  
  \end{tikzpicture}, 
\end{center} 
where the boundary of the box is identified. 

We consider an output-complete subnetwork $\gamma$. 
The space ${\ker}\, S_{11}$ 
is spanned by local cycles in $\gamma$, for example, 
\begin{center}
\begin{tikzpicture}[scale=0.8]

  \node[left] at (0, 1.8) {
    \scalebox{1.2}
    {
    $\quad \,\,\,\,{\ker}\, S_{11} = $
    }
    };

    \node[right] at (5.3, 2.4) {
      \scalebox{1.2}
      {$\gamma$} 
      };

  \draw[color=myblue, line width=1.8]
  (0.8, 0) -- (0.8,3.6);

  \draw[color=mydarkred, line width=1.8]
  (4.7,2) ellipse (0.3 and 0.5); 

  \draw[color=myblue, line width=1.8]
  (2, 0) -- (2,3.6);
  \draw[color=mydarkred, line width=1.8]
  (2, 1) -- (2,2.8);

  \draw[color=myblue, line width=1.8]
  (2.8, 0) -- (2.8,1.6);
  \filldraw[color=myblue, fill=myblue, line width=1.4]
  (2.8,1.6) circle (0.08);

  \draw[color=myblue, line width=1.8]
  (3.8, 0.5) -- (3.8,3.6);
  \draw[color=mydarkred, line width=1.8]
  (3.8, 1) -- (3.8,2.8); 
  \filldraw[color=myblue, fill=myblue, line width=1.4]
  (3.8,0.5) circle (0.08);

  \draw[black, line width=2] (0,0) rectangle (6,3.6);

  \draw[black, line width=1.5] (1.5,1) rectangle (6,2.8);

  \filldraw[color=mydarkred, fill=mydarkred, line width=1.4]
  (2,1) circle (0.08);
  \filldraw[color=mydarkred, fill=mydarkred, line width=1.4]
  (2,2.8) circle (0.08);

  \filldraw[color=mydarkred, fill=mydarkred, line width=1.4]
  (3.8,1) circle (0.08);
  \filldraw[color=mydarkred, fill=mydarkred, line width=1.4]
  (3.8,2.8) circle (0.08);

\end{tikzpicture}, 
\end{center} 
where the inner box represents a subnetwork $\gamma$ 
and the red lines constitute the basis of the space. 
Here, the symbol 
\begin{tikzpicture}
  \filldraw[color=mydarkred, fill=mydarkred, line width=1.4]
  (0,0) circle (0.08);
\end{tikzpicture}
means that the cut ends are reactions and not chemical species. 
See the following two choices for example: 
\begin{center} 
  \begin{tikzpicture}

    \node[species] (x) at (0,0) {$v_1$}; 
    \node[species] (y) at (1.5,0) {$v_2$}; 
    \node[species] (z) at (3,0) {$v_3$}; 
    \node[species] (w) at (4.5,0) {$v_4$}; 

    \draw [-latex,draw,line width=0.5mm] (x) edge node[below]{$e_1$} (y);
    \draw [-latex,draw,line width=0.5mm] (y) edge node[below]{$e_2$} (z);
    \draw [-latex,draw,line width=0.5mm] (z) edge node[below]{$e_3$} (w);

   \draw[mydarkred, dashed,line width=1] (0.3,-1) rectangle (2.7,1);
   \node[left] at (1, 0.7) {
    \scalebox{1.1} {\color{mydarkred} $\gamma_1$} };
\end{tikzpicture}
\qquad 
\begin{tikzpicture}

    \node[species] (x) at (0,0) {$v_1$}; 
    \node[species] (y) at (1.5,0) {$v_2$}; 
    \node[species] (z) at (3,0) {$v_3$}; 
    \node[species] (w) at (4.5,0) {$v_4$}; 

    \draw [-latex,draw,line width=0.5mm] (x) edge node[below]{$e_1$} (y);
    \draw [-latex,draw,line width=0.5mm] (y) edge node[below]{$e_2$} (z);
    \draw [-latex,draw,line width=0.5mm] (z) edge node[below]{$e_3$} (w);

 \draw[mydarkred, dashed,line width=1] (1.2,-1) rectangle (2.7,1);
 \node at (1.6, 0.7) {
  \scalebox{1.1} {\color{mydarkred} $\gamma_2$} };
\end{tikzpicture}
\end{center} 
For the left one, cut ends are both reactions. 
For the right one, the cut ends are a species and a reaction. 
Both ends have to be reactions so that the cut cycle can be a local cycle. 
An element of $\ker S_{11}$ may be extended to a global cycle, or it may be a part of a global noncycle.

The space $({\ker} S)_{{\rm supp}\,\gamma}$ for the same configuration takes into account only the global cycles supported on $\gamma$, 
\begin{center}
  \begin{tikzpicture}[scale=0.8]

    \node[left] at (0, 1.8) {
      \scalebox{1.2}
      {$(\ker S)_{{\rm supp}\,\gamma} = $ }
      };

    \node[right] at (5.3, 2.4) {
      \scalebox{1.2}
      {$\gamma$} 
      };

    \draw[color=myblue, line width=1.8]
    (0.8, 0) -- (0.8,3.6);
  
    \draw[color=mydarkred, line width=1.8]
    (4.7,2) ellipse (0.3 and 0.5); 
  
    \draw[color=myblue, line width=1.8]
    (2, 0) -- (2,3.6);
  
    \draw[color=myblue, line width=1.8]
    (2.8, 0) -- (2.8,1.6);
    \filldraw[color=myblue, fill=myblue, line width=1.4]
    (2.8,1.6) circle (0.08);
  
    \draw[color=myblue, line width=1.8]
    (3.8, 0.5) -- (3.8,3.6);
    \filldraw[color=myblue, fill=myblue, line width=1.4]
    (3.8,0.5) circle (0.08);
  
    \draw[black, line width=2] (0,0) rectangle (6,3.6);
  
    \draw[black, line width=1.5] (1.5,1) rectangle (6,2.8);
  
  \end{tikzpicture}. 
  \end{center} 
Therefore, the coset space is generated by the following elements, 
  \begin{center}
    \begin{tikzpicture}[scale=0.8]
    
      \node[left] at (0, 1.8) {
        \scalebox{1.2}
        {
        $\ker S_{11} / 
        ({\ker} S)_{{\rm supp}\,\gamma} = 
        $
        }
        };

        \node[right] at (5.3, 2.4) {
          \scalebox{1.2}
          {$\gamma$} 
      };
  
      \draw[color=myblue, line width=1.8]
      (0.8, 0) -- (0.8,3.6);
    
      \draw[color=myblue, line width=1.8] 
      (4.7,2) ellipse (0.3 and 0.5); 
    
      \draw[color=myblue, line width=1.8]
      (2, 0) -- (2,3.6);
      \draw[color=mydarkred, line width=1.8]
      (2, 1) -- (2,2.8);
    
      \draw[color=myblue, line width=1.8]
      (2.8, 0) -- (2.8,1.6);
      \filldraw[color=myblue, fill=myblue, line width=1.4]
      (2.8,1.6) circle (0.08);
    
      \draw[color=myblue, line width=1.8]
      (3.8, 0.5) -- (3.8,3.6);
      \draw[color=mydarkred, line width=1.8]
      (3.8, 1) -- (3.8,2.8); 
      \filldraw[color=myblue, fill=myblue, line width=1.4]
      (3.8,0.5) circle (0.08);
    
      \draw[black, line width=2] (0,0) rectangle (6,3.6);
    
      \draw[black, line width=1.5] (1.5,1) rectangle (6,2.8);

  \filldraw[color=mydarkred, fill=mydarkred, line width=1.4]
  (2,1) circle (0.08);
  \filldraw[color=mydarkred, fill=mydarkred, line width=1.4]
  (2,2.8) circle (0.08);

  \filldraw[color=mydarkred, fill=mydarkred, line width=1.4]
  (3.8,1) circle (0.08);
  \filldraw[color=mydarkred, fill=mydarkred, line width=1.4]
  (3.8,2.8) circle (0.08);

    \end{tikzpicture}. 
\end{center} 
As we see in the figure, 
the space 
$\ker S_{11} / 
(\ker S)_{{\rm supp}\,\gamma}$
consists of local cycles that are not global cycles.

We can similarly interpret conserved charges. 
%
The transpose of the stoichiometric matrix, 
$S^T$, can be regarded as a boundary operator acting on 
$C_0(\Gamma)$, 
which is the space of chemical species. 
In this sense, an element of $\coker S$ 
has no boundary, 
with respect to this boundary operator. 
We here visualize this in a similar way to the cycles, 
\begin{center}
  \begin{tikzpicture}[scale=0.8]
  
    \draw[color=mydarkred, line width=1.8]
    (3, 0) -- (2.5,3.6);

     \draw[color=mydarkred, line width=1.8]
     (1, 0) -- (0,1);    

    \draw[color=mydarkred, line width=1.8]
    (1,2.2) ellipse (0.4 and 0.6); 
    
    \draw[color=mydarkpurple, line width=1.8]
    (4.5, 0) -- (4.5,1.6); 
    \filldraw[color=mydarkpurple, fill=mydarkpurple, line width=1.4]
    (4.5,1.6) circle (0.08);

    \draw[color=mydarkpurple, line width=1.8]
    (3.9, 2.2) -- (5,3.2); 
    \filldraw[color=mydarkpurple, fill=mydarkpurple, line width=1.4]
    (3.9,2.2) circle (0.08);
    \filldraw[color=mydarkpurple, fill=mydarkpurple, line width=1.4]
    (5,3.2) circle (0.08);

    \draw[black, line width=2] (0,0) rectangle (6,3.6);
  
    \draw[->,>=stealth,color=black, line width=1]
     (-1,3) -- (2.4,3);

     \draw[->,>=stealth,color=black, line width=1]
     (-1,2.9) -- (0.5,2.6);

     \draw[->,>=stealth,color=black, line width=1]
     (-1,2.8) -- (-0.1,1);

     \node at (-3, 3) {
      \scalebox{1.1} {Conserved charges} };
  
      \draw[->,>=stealth,color=black, line width=1]
      (6.5,1.5) -- (4.7,1); 
      \draw[->,>=stealth,color=black, line width=1]
      (6.5,1.6) -- (4.7,2.6); 
 
      \node at (8, 1.5) {
        \scalebox{1.05}
        {Not conserved} 
        };
  
  \end{tikzpicture}. 
\end{center} 
Note that the boundary of the box is identified. 
The filled circles 
\begin{tikzpicture}
  \filldraw[color=mydarkpurple, fill=mydarkpurple, line width=1.4]
  (0,0) circle (0.08);  
\end{tikzpicture}
represent a source or a drain of chemical species, 
because of which the charge is not conserved.

The space $\coker S$ represents the global conserved charges, and 
$P^0_\gamma (\coker S)$ is the projection of $\coker S$ to $\gamma$, 
\begin{center}
  \begin{tikzpicture}[scale=0.8]
    
      \node[left] at (0, 1.8) {
        \scalebox{1.2}
        {
        $P^0_\gamma (\coker S) = $
        }
        };

        \node[right] at (5.3, 2.4) {
          \scalebox{1.2}
          {$\gamma$} 
      };

      \draw[color=mydarkpurple, line width=1.8]
      (0.5, 0) -- (0.5,3.6);
    
      \draw[color=mydarkpurple, line width=1.8]
      (1.3, 0) -- (1.3,3.6);
      \draw[color=mydarkred, line width=1.8]
      (1.3, 1) -- (1.3,2.8); 

      \draw[color=mydarkpurple, line width=1.8]
      (2.1, 0) -- (2.1,3.6);
      \draw[color=mydarkred, line width=1.8]
      (2.1, 1) -- (2.1,2.8);  

      \draw[color=mydarkpurple, line width=1.8]
      (2.9, 0) -- (2.9,1.6);
      \filldraw[color=mydarkpurple, fill=mydarkpurple, line width=1.4]
      (2.9,1.6) circle (0.08);
    
      \draw[color=mydarkpurple, line width=1.8]
      (3.8, 0.5) -- (3.8,3.2);
      \filldraw[color=mydarkpurple, fill=mydarkpurple, line width=1.4]
      (3.8,3.2) circle (0.08);
      \filldraw[color=mydarkpurple, fill=mydarkpurple, line width=1.4]
      (3.8,0.5) circle (0.08);

      \draw[color=mydarkred, line width=1.8] 
      (4.7,2.1) ellipse (0.5 and 0.3); 

      \draw[color=mydarkpurple, line width=1.8]
      (5, 0.5) -- (5,1.5) -- (6,1.5);

      \filldraw[color=mydarkpurple, fill=mydarkpurple, line width=1.4]
      (5,0.5) circle (0.08);
      
      \draw[black, line width=2] (0,0) rectangle (6,3.6);
    
      \draw[black, line width=1.5] (1,1) rectangle (6,2.8);
    
    \end{tikzpicture}. 
\end{center} 
Here, how the conserved charges are cut does not matter. 
The space $\coker S_{11}$ 
is generated by the red and green elements 
in the following figure, 
\begin{center}
  \begin{tikzpicture}[scale=0.8]
    
      \node[left] at (0, 1.8) {
        \scalebox{1.2}
        {
        $\coker S_{11} = $ 
        }
        };

        \node[right] at (5.3, 2.4) {
          \scalebox{1.2}
          {$\gamma$} 
      };

      \draw[color=mydarkpurple, line width=1.8]
      (0.5, 0) -- (0.5,3.6);
    
      \draw[color=mydarkpurple, line width=1.8]
      (1.3, 0) -- (1.3,3.6);
      \draw[color=mydarkred, line width=1.8]
      (1.3, 1) -- (1.3,2.8); 
 
      \draw[color=mydarkpurple, line width=1.8]
      (2.1, 0) -- (2.1,3.6);

      \draw[color=mydarkpurple, line width=1.8]
      (2.9, 0) -- (2.9,1.6);
      \filldraw[color=mydarkpurple, fill=mydarkpurple, line width=1.4]
      (2.9,1.6) circle (0.08);
    
      \draw[color=mydarkpurple, line width=1.8]
      (3.8, 0.5) -- (3.8,3.2);
      \draw[color=mygreen, line width=1.8]
      (3.8, 1) -- (3.8,2.9);
      \filldraw[color=mydarkpurple, fill=mydarkpurple, line width=1.4]
      (3.8,3.2) circle (0.08);
      \filldraw[color=mydarkpurple, fill=mydarkpurple, line width=1.4]
      (3.8,0.5) circle (0.08);

      \draw[color=mydarkred, line width=1.8] 
      (4.7,2.1) ellipse (0.5 and 0.3); 

      \draw[color=mydarkpurple, line width=1.8]
      (5, 0.5) -- (5,1.5) -- (6,1.5);
      \filldraw[color=mydarkpurple, fill=mydarkpurple, line width=1.4]
      (5,0.5) circle (0.08);

      \draw[color=mygreen, line width=1.8]
      (5, 1) -- (5,1.5) -- (6,1.5);
      
      \draw[black, line width=2] (0,0) rectangle (6,3.6);
    
      \draw[black, line width=1.5] (1,1) rectangle (6,2.8);


      \filldraw[fill=mydarkred, draw=mydarkred, line width=1.4]
       (1.2,0.9) rectangle (1.4,1.1);
       \filldraw[fill=mydarkred, draw=mydarkred, line width=1.4]
       (1.2,2.7) rectangle (1.4,2.9);

       \filldraw[fill=mydarkpurple, draw=mydarkpurple, line width=1.4]
       (2,0.9) rectangle (2.2,1.1);
       \filldraw[fill=white, draw=mydarkpurple, line width=1.4]
       (2.0,2.7) rectangle (2.2,2.9);

       \filldraw[fill=mygreen, draw=mygreen, line width=1.4]
       (3.7,0.9) rectangle (3.9,1.1);
       \filldraw[fill=mygreen, draw=mygreen, line width=1.4]
       (3.7,2.7) rectangle (3.9,2.9);
       \filldraw[fill=mygreen, draw=mygreen, line width=1.4]
       (4.9,0.9) rectangle (5.1,1.1);

       \draw[->,>=stealth,color=black, line width=1]
       (7,3.3) -- (4.0,2.5); 
       \draw[->,>=stealth,color=black, line width=1]
       (7,3.2) -- (5.5,1.6); 
  
       \node at (9, 3.7) {
         \scalebox{1.0} {Emergent conserved charges} };
 
    \end{tikzpicture}, 
\end{center} 
where filled and open rectangles mean boundaries 
with chemical species and reactions, respectively. 
The parts we denoted by green lines, 
\begin{tikzpicture}
  \draw[color=mygreen, line width=1.8]
  (0, 0) -- (1,0);
  \path (0,-0.1); 
\end{tikzpicture}, 
are {\it emergent conserved charges}, 
which are conserved when it is seen in a subnetwork 
but not conserved in $\Gamma$. 
In fact, the appearance of emergent conserved charges 
typically leads to ``unphysical'' systems, 
in the sense that 
either a steady state does not exist or 
the matrix $A$ is not invertible
and the response of the system to the perturbation of parameters is not well-defined.\footnote{
We discuss more on this point in Appendix~\ref{sec:app-emergent}. 
} 
For example, one can consider the following 
network and subnetwork, 
\begin{center}
\begin{tikzpicture}

    \node (x) at (0.25,0) {};
    \node[species] (y) at (1.5,0) {$v_1$}; 
    \node[species] (z) at (3,0) {$v_2$}; 
    \node[species] (w) at (4.5,0) {$v_3$}; 

    \draw [-latex,draw,line width=0.5mm] (x) edge node[below]{$e_1$} (y);
    \draw [-latex,draw,line width=0.5mm] (y) edge node[below]{$e_2$} (z);
    \draw [-latex,draw,line width=0.5mm] (z) edge node[below]{$e_3$} (w);

    \draw[mydarkred, dashed,line width=1] (1.2,-1) rectangle (5,1);
    \node at (1.5, 0.7) {\scalebox{1.1} {\color{mydarkred} $\gamma$} };
    
\end{tikzpicture} 
\end{center}
The whole network does not have a conserved charge, 
but the subnetwork $\gamma$ has one, $v_1 + v_2 + v_3$. 
However, such a reaction network cannot reach a steady state, 
since the concentration of $v_3$ continues to increase. 
%
%
When we take the difference 
$|P^0_\gamma (\coker S)| - |\coker S_{11}|$ , 
we can count the number of lost conserved charges 
minus the number of emergent conserved charges (if any), 
\begin{center}
  \begin{tikzpicture}[scale=0.8, every node/.style={scale=0.8}]
    
      \node[left] at (0, 1.8) {
        \scalebox{1.3}
        {
        $|P^0_\gamma (\coker S)|-|\coker S_{11} |
        = $
        }
        };

        \node[right] at (5.3, 2.4) {
          \scalebox{1.2}
          {$\gamma$} 
      };

      \draw[color=mydarkpurple, line width=1.8]
      (0.5, 0) -- (0.5,3.6);
    
      \draw[color=mydarkpurple, line width=1.8]
      (1.3, 0) -- (1.3,3.6);

      \draw[color=mydarkpurple, line width=1.8]
      (2.1, 0) -- (2.1,3.6);
      \draw[color=mydarkred, line width=1.8]
      (2.1, 1) -- (2.1,2.8);  

      \draw[color=mydarkpurple, line width=1.8]
      (2.9, 0) -- (2.9,1.6);
      \filldraw[color=mydarkpurple, fill=mydarkpurple, line width=1.4]
      (2.9,1.6) circle (0.08);
    
      \draw[color=mydarkpurple, line width=1.8]
      (3.8, 0.5) -- (3.8,3.2);
      \filldraw[color=mydarkpurple, fill=mydarkpurple, line width=1.4]
      (3.8,3.2) circle (0.08);
      \filldraw[color=mydarkpurple, fill=mydarkpurple, line width=1.4]
      (3.8,0.5) circle (0.08);

      \draw[color=mydarkpurple, line width=1.8] 
      (4.7,2.1) ellipse (0.5 and 0.3); 

      \draw[color=mydarkpurple, line width=1.8]
      (5, 0.5) -- (5,1.5) -- (6,1.5);

      \filldraw[color=mydarkpurple, fill=mydarkpurple, line width=1.4]
      (5,0.5) circle (0.08);
      
      \draw[black, line width=2] (0,0) rectangle (6,3.6);
    
      \draw[black, line width=1.5] (1,1) rectangle (6,2.8);

      \filldraw[fill=mydarkpurple, draw=mydarkpurple, line width=1.4]
       (1.2,0.9) rectangle (1.4,1.1);
       \filldraw[fill=mydarkpurple, draw=mydarkpurple, line width=1.4]
       (1.2,2.7) rectangle (1.4,2.9);

       \filldraw[fill=mydarkred, draw=mydarkred, line width=1.4]
       (2,0.9) rectangle (2.2,1.1);
       \filldraw[fill=white, draw=mydarkred, line width=1.4]
       (2.0,2.7) rectangle (2.2,2.9);

\end{tikzpicture}\begin{tikzpicture}[scale=0.8]

      \node[left] at (0, 1.8) {
        \scalebox{1.5}
        {
        $-$ 
        }
        };

        \node[right] at (5.3, 2.4) {
          \scalebox{1.2}
          {$\gamma$} 
      };

      \draw[color=mydarkpurple, line width=1.8]
      (0.5, 0) -- (0.5,3.6);
    
      \draw[color=mydarkpurple, line width=1.8]
      (1.3, 0) -- (1.3,3.6);
      \draw[color=mydarkpurple, line width=1.8]
      (1.3, 1) -- (1.3,2.8); 
 
      \draw[color=mydarkpurple, line width=1.8]
      (2.1, 0) -- (2.1,3.6);

      \draw[color=mydarkpurple, line width=1.8]
      (2.9, 0) -- (2.9,1.6);
      \filldraw[color=mydarkpurple, fill=mydarkpurple, line width=1.4]
      (2.9,1.6) circle (0.08);
    
      \draw[color=mydarkpurple, line width=1.8]
      (3.8, 0.5) -- (3.8,3.2);
      \draw[color=mygreen, line width=1.8]
      (3.8, 1) -- (3.8,2.9);
      \filldraw[color=mydarkpurple, fill=mydarkpurple, line width=1.4]
      (3.8,3.2) circle (0.08);
      \filldraw[color=mydarkpurple, fill=mydarkpurple, line width=1.4]
      (3.8,0.5) circle (0.08);

      \draw[color=mydarkpurple, line width=1.8] 
      (4.7,2.1) ellipse (0.5 and 0.3); 

      \draw[color=mydarkpurple, line width=1.8]
      (5, 0.5) -- (5,1.5) -- (6,1.5);
      \filldraw[color=mydarkpurple, fill=mydarkpurple, line width=1.4]
      (5,0.5) circle (0.08);

      \draw[color=mygreen, line width=1.8]
      (5, 1) -- (5,1.5) -- (6,1.5);
      
      \draw[black, line width=2] (0,0) rectangle (6,3.6);
    
      \draw[black, line width=1.5] (1,1) rectangle (6,2.8);


      \filldraw[fill=mydarkpurple, draw=mydarkpurple, line width=1.4]
       (1.2,0.9) rectangle (1.4,1.1);
       \filldraw[fill=mydarkpurple, draw=mydarkpurple, line width=1.4]
       (1.2,2.7) rectangle (1.4,2.9);

       \filldraw[fill=mydarkpurple, draw=mydarkpurple, line width=1.4]
       (2,0.9) rectangle (2.2,1.1);
       \filldraw[fill=white, draw=mydarkpurple, line width=1.4]
       (2.0,2.7) rectangle (2.2,2.9);

       \filldraw[fill=mygreen, draw=mygreen, line width=1.4]
       (3.7,0.9) rectangle (3.9,1.1);
       \filldraw[fill=mygreen, draw=mygreen, line width=1.4]
       (3.7,2.7) rectangle (3.9,2.9);
       \filldraw[fill=mygreen, draw=mygreen, line width=1.4]
       (4.9,0.9) rectangle (5.1,1.1);  

  
 
\end{tikzpicture}, 
\end{center}
where the part colored in red 
in the first term indicates lost conserved charges, 
that are conserved in $\Gamma$ 
but their projections to $\gamma$ are not. 
This equation is equal to 
the latter two terms of the decomposition (\ref{eq:lambda-decom}), 
$d_l (\gamma) - \widetilde d(\gamma)$.

\begin{example}
Consider a monomolecular network 
$\Gamma = (V,E)= ( \{ v_1,v_2,v_3 \}, \{e_1,e_2\} )$ 
with the following structure, 
\begin{center}
    \begin{tikzpicture}

    \node[species] (x) at (0,0) {$v_1$}; 
    \node[species] (y) at (1.5,0) {$v_2$}; 
    \node[species] (z) at (3,0) {$v_3$}; 

    \draw [-latex,draw,line width=0.5mm] (x) edge node[below]{$e_1$} (y);
    \draw [-latex,draw,line width=0.5mm] (y) edge node[below]{$e_2$} (z);

    \draw[mydarkred, dashed,line width=1] (0.3,-1) rectangle (2.7,1);
    \node[left] at (1, 0.7) {\scalebox{1.1} {\color{mydarkred} $\gamma_1$} };
\end{tikzpicture} 
\end{center}
We take a subnetwork $\gamma_1 = ( \{ v_2 \}, \{e_1,e_2\})$ 
that is indicated by a box. 
The whole network does not have a cycle, and 
the subnetwork $\gamma_1$ has one emergent cycle given by 
$c = e_1 + e_2$. 
Also, $\Gamma$ has one conserved charge, $d=v_1+v_2+v_3$. 
Its projection to $\gamma_1$ is given by $v_2$ 
and it is not a conserved charge in $\gamma_1$. 
So we have one lost conserved charge. 
Each integer appearing in the decomposition of $\lambda(\gamma_1)$ is 
\begin{equation}
  \widetilde c(\gamma_1)=1, 
  \quad d_l (\gamma_1) = 1, 
  \quad \widetilde d (\gamma_1) = 0,
\end{equation}
and $\lambda(\gamma_1) = 2$. 
For the same $\Gamma$, let us consider a different choice of a subnetwork, 
\begin{center}
  \begin{tikzpicture}

    \node[species] (x) at (0,0) {$v_1$}; 
    \node[species] (y) at (1.5,0) {$v_2$}; 
    \node[species] (z) at (3,0) {$v_3$}; 

    \draw [-latex,draw,line width=0.5mm] (x) edge node[below]{$e_1$} (y);
    \draw [-latex,draw,line width=0.5mm] (y) edge node[below]{$e_2$} (z);

    \draw[mydarkred, dashed,line width=1] (1.2,-1) rectangle (2.7,1);
    \node at (1.5, 0.7) {
    \scalebox{1.1} {\color{mydarkred} $\gamma_2$} };
\end{tikzpicture}
\end{center}
The subnetwork $\gamma_2$  does not have a cycle,
and there is one lost conserved charge, so we have 
\begin{equation}
  \widetilde c(\gamma_2)=0,
  \quad d_l (\gamma_2) = 1, 
  \quad \widetilde d(\gamma_2)=0,
  \quad \lambda(\gamma_2)=1. 
\end{equation}
\end{example}

\begin{example}
Consider a network $(V,E) = ( \{ v_1,v_2,v_3 \}, \{e_1,e_2,e_3\} )$ 
with the following structure, 
\begin{center}
  \begin{tikzpicture}

    \node[species] (x) at (0,0) {$v_1$}; 
    \node[species] (y) at (2,0) {$v_2$}; 
    \node[species] (z) at (1,1.2) {$v_3$}; 

    \draw [-latex,draw,line width=0.5mm] (x) edge node[below]{$e_1$} (y);
    \draw [-latex,draw,line width=0.5mm] (y) edge node[above right]{$e_2$} (z);
    \draw [-latex,draw,line width=0.5mm] (z) edge node[above left]{$e_3$} (x);

    \draw[color=mydarkred, line width=1, dashed]  
    (1,0.89) ellipse (1.2 and 0.8); 

    \node at (0,1.8) { \scalebox{1.1} {\color{mydarkred} $\gamma$} };

\end{tikzpicture}
\end{center}
If we choose a subnetwork $\gamma = ( \{v_3 \},\{e_2,e_3\})$, 
\begin{equation}
  \widetilde c(\gamma)=1, \quad d_l (\gamma) = 1, 
  \quad \widetilde d(\gamma)=0, 
  \quad \lambda(\gamma) = 2. 
\end{equation}
The subnetwork $\gamma$ has one emergent cycle and one lost conserved charge and the influence index is $2$. 
\end{example}

\subsection{Embedding of A-matrices}\label{sec:app-embedding} 

\begin{figure}[tbh]
  \centering
  \includegraphics[keepaspectratio, scale=0.5]{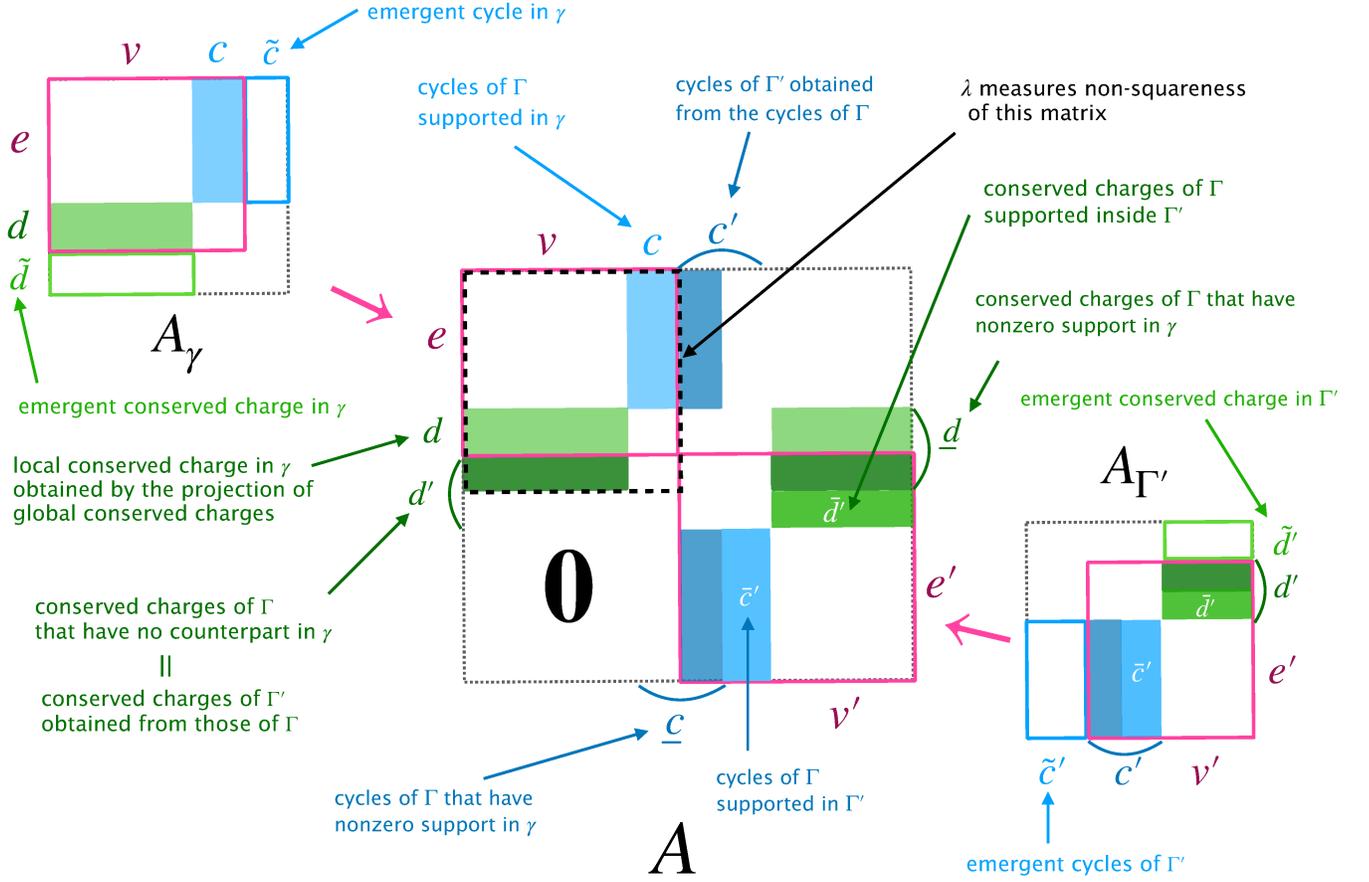} 
  \caption{
Embedding of the A-matrices 
for a generic output-complete subnetwork. 
  }
  \label{fig:mat-a}
 \end{figure}

It is useful to look at the A-matrix to visualize 
the relations among cycles/conserved charges of various types
in subnetworks and reduced networks. 

Let us first summarize the notations. 
In this section, we suppress the dependence on $\gamma$ for notational simplicity.
General rules are as follows. Quantities with a tilde are emergent ones, and we use character $c$ for cycles and $d$ for conserved charges. The numbers with a prime are associated with $\Gamma'$. 
The relevant numbers are listed as follows:
\begin{itemize}
    \item $v,v'$: number of chemical species in $\gamma,\Gamma'$
    \item $e,e'$: number of chemical reactions in $\gamma,\Gamma'$
    \item $\widetilde c, \widetilde c'$: number of emergent cycles of $\gamma, \Gamma'$
    \item $\widetilde d, \widetilde d'$: number of emergent conserved charges of $\gamma, \Gamma'$
    \item $c, c'$: number of cycles of $\Gamma$, whose projections to $\gamma, \Gamma'$  are also cycles of $\gamma, \Gamma'$
    \item $d, d'$: number of conserved charges of $\Gamma$, whose projections to $\gamma, \Gamma'$  are also conserved in $\gamma, \Gamma'$
    \item $\bar c', \bar d'$: number of cycles/conserved charges of $\Gamma$ that are locally supported in $\Gamma'$ 
    \item $\underline{c}, \underline{d}$: number of 
    cycles/conserved charges of $\Gamma$ that have nonzero support in $\gamma$ 
\end{itemize}

In Fig.~\ref{fig:mat-a}, 
we illustrate a more detailed structure of the matrix $A$ than Fig.~\ref{fig:mat-a-lol}. 
In the center is the matrix $A$ of the total system $\Gamma$. 
We choose an output-complete subnetwork $\gamma$, 
and bring the rows/columns related to $\gamma$ to the upper-left part. 
Then the matrix $A$ looks like one in the center of  Fig.~\ref{fig:mat-a}. 
We consider an output-complete subnetwork $\gamma$, 
and the A-matrix of $\gamma$, which we denote by $A_\gamma$, 
is shown in the upper-left part of  Fig.~\ref{fig:mat-a}. 
The part surrounded by a pink rectangle is the common part of $A_\gamma$ and $A$. 
The subnetwork $\gamma$ can in general contain
additional (i.e., emergent) cycles and conserved charges, 
whose numbers are denoted by $\widetilde c$ and $\widetilde d$. 
Because the matrix $A_\gamma$ is square, we have the relation, 
\begin{equation}
  e+d+\widetilde d = v+c+\widetilde c.  
  \label{eq:a-gamma-sq}
\end{equation}
This equation is in fact the same as Eq.~(\ref{eq:euler-sub}). 
Similarly, 
we can consider the matrix $A$ for the network $\Gamma' = \Gamma/\gamma$ \footnote{
  When $\widetilde c (\gamma) > 0$, 
  the equation of motion contains some terms 
  that cannot be determined, as in Eq.~(\ref{eq:rate-red-1}). 
  Here, we formally consider a reduced network $\Gamma'$ 
  which is defined with the generalized Schur complement $S'$. 
  In this sense, the property of the reduced network 
  defined this way cannot be fully constrained 
  from the properties of $\Gamma$ and $\gamma$. 
}
obtained by reducing $\gamma$ from $\Gamma$. 
The numbers of the emergent cycles and emergent conserved charges 
in $\Gamma'$ are denoted by $\widetilde c'$ and $\widetilde d'$. 
The matrix $A_{\Gamma'}$ is also square and we have 
\begin{equation}
  e'+d'+\widetilde d'  = v'+c'+\widetilde c'  . 
  \label{eq:a-gammap-sq}
\end{equation}
The influence index is given by 
\begin{equation}
\lambda\coloneqq e+d
+ d'  -\bar d' - v - c   , 
\end{equation}
which measures how far the rectangle 
in the upper-left part (indicated by black dashed lines)
is from a square matrix. 
Note that this expression is consistent with the one in Sec.~\ref{sec:lol-lol} since $\underline{d} = d+d'-\bar d'$. 
Using Eq.~(\ref{eq:a-gamma-sq}), we can also express $\lambda$ as 
\begin{equation}
  \lambda = \widetilde c
  + d' - \bar d' - \widetilde d , 
\end{equation}
which is the same the decomposition (\ref{eq:lambda-decom}). 
We can also consider a similar quantity 
that measure the non-squareness of the lower-right part, 
\begin{equation}
  \lambda'\coloneqq v'+c' 
  -e'-\bar d'
  = 
  \widetilde d' + d' 
  - \bar d'  
  - \widetilde c' , 
\end{equation}
where the second expression is obtained using Eq.~(\ref{eq:a-gammap-sq}). 
In fact, 
due to the squareness of the whole matrix $A$, 
$\lambda'$ is equal to the influence index, 
$\lambda = \lambda'$. 
This results in the following relation, 
\begin{equation}
  \widetilde c  - \widetilde d
  +
  \widetilde c' - \widetilde d'
  = 
  0 . 
  \label{eq:c-tc=td-tdp}
\end{equation}

\section{Emergent conserved charges in chemical reaction networks}\label{sec:app-emergent}

In this section, 
we discuss the role of emergent conserved charges in chemical reaction networks. 

\subsection{Systems with emergent conserved charges}

As far as we observe, the chemical reaction systems 
with emergent conserved charges in output-complete subnetworks 
are {\it pathological}, 
in either of the following senses: 
\begin{itemize}
  \item[(A)] The steady-state condition, 
  \begin{equation}
  \sum_A S_{iA} r_A (\bm x (\bm k, \bm \ell), k_A)=0, 
    \quad 
  \sum_i  d_i^{\bar\alpha} x_i = \ell^{\bar\alpha}  , 
  \end{equation}
  does not fully 
  determine the steady-state solution, and arbitrary parameters have to be introduced to specify the solution. 
  \item[(B)] No steady-state solution exists. 
  \item[(C)] The reaction kinetics is unphysical. 
\end{itemize}
Below, we discuss some examples of each case. 

\subsubsection{Pattern A : solutions have arbitrary parameters}

An example of pattern (A) is given by 
\begin{equation}
  \frac{d}{dt} 
  \begin{pmatrix}
    x \\
    y
  \end{pmatrix}
  = 
  \begin{pmatrix}
    1 & 0 
    \\ 
    -1 & -1 
  \end{pmatrix}
  \begin{pmatrix}
    r_1 \\
    r_2 
  \end{pmatrix}, 
  \quad \,\,\,
  \begin{pmatrix}
    r_1 \\
    r_2 
  \end{pmatrix}
 = 
 \begin{pmatrix}
   k_1 y \\
   k_2 y
 \end{pmatrix}. 
\end{equation}
The matrix $A$ for this system is 
\begin{equation}
  A 
  = 
  \begin{pmatrix}
    0 & k_1 \\
    0 & k_2 
  \end{pmatrix}. 
\end{equation}
This is not invertible. 
The steady-state solution 
for the mass-action kinetics is given by 
\begin{equation}
\begin{pmatrix}
  \bar r_1 \\
  \bar r_2 
\end{pmatrix}
= 
\bm 0 , 
\quad 
\begin{pmatrix}
  \bar x \\
  \bar y 
\end{pmatrix}
= 
\begin{pmatrix}
  m 
  \\ 
  0 
\end{pmatrix}, 
\end{equation}
where $m$ is an arbitrary parameter. 
If we choose a subnetwork $\gamma = \{ x \}$, 
there is an emergent conserved charge. 
Let us consider the fluctuations around the steady state, 
\begin{equation}
  \frac{d}{dt}
  \begin{pmatrix}
    \delta x \\
    \delta y 
  \end{pmatrix}
  = 
  \begin{pmatrix}
    1 & 0 \\
    -1 & -1 
  \end{pmatrix}
  \begin{pmatrix}
    k_1 \delta y \\
    k_2 \delta y 
  \end{pmatrix}
  = 
  \begin{pmatrix}
    k_1 \delta y  \\
    - (k_1 +k_2) \delta y 
  \end{pmatrix},  
\end{equation}
where $\delta x (t) \coloneqq x(t) - \bar x$ indicates the fluctuation from the steady state. 
The fluctuation 
associated with the emergent charge, $\delta x$, is a zero mode. 
This means that the system is not asymptotically stable. 

Generically, when we have to introduce arbitrary parameters $\bm m$, 
the matrix $A$ has a null vector, as we see below. 
The steady-state condition reads 
\begin{eqnarray}
 r_A (\bm x (\bm k,\bm \ell, \bm m),\bm k) 
  &= &
 - \sum_\alpha \mu_\alpha (\bm k, \bm \ell, \bm m) c^\alpha_A , 
  \\
 \sum_i  d^{\bar \alpha}_i x_i (\bm k,\bm \ell, \bm m) 
  &=& \ell^{\bar\alpha}. 
\end{eqnarray}
By taking the derivative of those equations 
with respect to $\bm m$, we find 
\begin{equation}
  \begin{pmatrix}
    r_{A,i} & c^\alpha_A \\
    d^{\bar \alpha}_i & \bm 0
  \end{pmatrix}
  \frac{\p}{\p m^a} 
  \begin{pmatrix}
    x_i \\
    \mu_\alpha 
  \end{pmatrix}
  = 
  \bm 0 . 
\end{equation}
This means that $A$ has a null vector
and $\det A = 0$.

\subsubsection{Pattern B: no steady-state solution}

An example of pattern (B) is given by 
the following reaction system
with the mass-action kinetics, 
\begin{equation}
\frac{d}{dt}
\begin{pmatrix}
  x_1 \\
  x_2 \\
  x_3
\end{pmatrix}
  = 
  \begin{pmatrix}
   -1 & 0 & 1 & 0 \\
   -1 & 0 & 0 & 1 \\
    1 & -1 & 0 & 0 
  \end{pmatrix}
  \begin{pmatrix}
    r_1 (x_1, x_2) \\
    r_2 (x_3)\\
    r_3 \\
    r_4 
  \end{pmatrix}, 
\quad 
  \begin{pmatrix}
    r_1 \\
    r_2 \\
    r_3 \\
    r_4 
  \end{pmatrix}
= 
\begin{pmatrix}
  k_1 x_1 x_2 \\
  k_2 x_3 \\
  k_3  \\
  k_4  \\
\end{pmatrix}  . 
\end{equation}
The steady-state solution does not exist in general. 
We need to fine-tune the parameters to have a solution. 
When $k_3 = k_4$ is satisfied, 
we have a steady state 
\begin{equation}
\bar {\bm r}
= 
k_3 
\begin{pmatrix}
  1 &  1&1&1
\end{pmatrix}^T . 
\end{equation}
The matrix $A$ is 
\begin{equation}
    A
    = 
    \begin{pmatrix}
     \p_1 r_1 & \p_2 r_1 & 0 & 1 \\
     0 & 0 & \p_3 r_2 & 1 \\
     0 & 0 & 0 & 1 \\
     0 & 0 & 0 & 1 \\
    \end{pmatrix},
\end{equation}
where $\p_j r_i \coloneqq \p r_i / \p x_j$ and it is evaluated at the steady state. 
This matrix is not regular, $\det A = 0$.

Let us choose an output-complete subnetwork 
$\gamma = (\{v_1,v_2\},\{e_1\})$. 
The matrix $A$ for the subnetwork is 
\begin{equation}
  A_\gamma 
  = 
  \begin{pmatrix}
    \p_1 r_1 & \p_2 r_1 \\
    1 & -1 
  \end{pmatrix}. 
\end{equation}
The subnetwork has an emergent conserved charge, 
$\widetilde {\bm d}_1^T 
  = 
  \begin{pmatrix}
    1 & -1 
  \end{pmatrix}  .
$
The time derivative of this charge is 
\begin{equation}
  \frac{d}{dt}
 \widetilde{\bm d}^T 
 \bm x_\gamma 
 = 
  \frac{d}{dt}
  \begin{pmatrix}
    1 & -1 
  \end{pmatrix}
  \begin{pmatrix}
    x_1 \\
    x_2
  \end{pmatrix}
  = 
  r_3 - r_4 
  = k_3 - k_4 .
\end{equation}
Although $\widetilde{\bm d}^T S \neq \bm 0$, 
where $\widetilde{\bm d}^T \coloneqq (\widetilde{\bm d}^T_1 \,\,a)$, 
for any parameter $a$, 
when the steady state exists, 
$k_3=k_4$, the combination $\widetilde{\bm d}^T \bm x$
is in fact a conserved charge of the whole system. 
It is not conserved unless the parameters are fine-tuned.

Let us consider the fluctuations around the steady state, 
\begin{equation}
  \frac{d}{dt}
  \begin{pmatrix}
    \delta x_1 \\
    \delta x_2 \\
    \delta x_3
  \end{pmatrix}
   = 
    \begin{pmatrix}
     -1 & 0 & 1 & 0 \\
     -1 & 0 & 0 & 1 \\
      1 & -1 & 0 & 0 
    \end{pmatrix}
    \begin{pmatrix}
     \delta  r_1 (x_1, x_2) \\
     \delta  r_2 (x_3)\\
     \delta r_3 \\
     \delta r_4 
   \end{pmatrix}. 
\end{equation}
The fluctuation associated with the emergent conserved charge 
leads to a zero mode, 
\begin{equation}
  \frac{d}{dt}
\delta ( x_1 -  x_2) 
= \delta (r_3 - r_4)
= 0. 
\end{equation}  

\subsubsection{Pattern C: example with emergent conserved charges and  unphysical kinetics} 

Here we discuss an example 
that has a subnetwork with vanishing influence index 
and also has an emergent conserved charge, 
while the kinetics is unphysical. 
The rate equation of this system is 
\begin{equation}
  \frac{d}{dt}
  \begin{pmatrix}
    x_1 \\
    x_2 \\
    x_3 \\
    x_4
  \end{pmatrix}
  = 
  \begin{pmatrix}
    -1 & 1 & 1&1 \\
    -1 & 1 & 1&2 \\
    1 & 0 & 0 & 0\\
    0&-2&-1&-1 
  \end{pmatrix}
  \begin{pmatrix}
    r_1 (x_1,x_2) \\
    r_2 (x_2,x_4)\\
    r_3 (x_3, x_4) \\
    r_4 (x_4) 
  \end{pmatrix}. 
  \label{eq:ex-unphysical}
\end{equation}
We have added catalytic dependencies 
in 
the reactions $r_2(x_2,x_4)$ and $r_3(x_3,x_4)$. 
The stoichiometric matrix has 
a trivial kernel and $\bar r_A=0$ at the steady state. 
The cokernel of $S$ is also trivial. 

The matrix $A$ is 
\begin{equation}
  A
  =
  \begin{pmatrix}
\p_1 r_1 & \p_2 r_1 & 0 & 0 \\
 0 & \p_2 r_2  & 0     & \p_4 r_2 \\
 0 & 0& \p_3 r_3 & \p_4 r_3 \\
0&0&0& \p_4 r_4  
\end{pmatrix} . 
\end{equation}
Its determinant is in general nonvanishing, 
\begin{equation}
  \det A = \p_1 r_1 \, \p_2 r_2 \, \p_3 r_3 \, \p_4 r_4  . 
\end{equation}

Let us consider an output-complete subnetwork 
$\gamma =( \{ x_1,x_2\},\{r_1,r_2 \})$. 
The index of $\gamma$ is zero, 
\begin{equation}
  \lambda(\gamma)
  = -2+2 - 0 + 0 
  = 0, 
\end{equation}
and hence it is a buffering structure. 
The matrix $A$ of the local system reads 
\begin{equation}
  A_\gamma 
  = 
  \begin{pmatrix}
    \p_1 r_1 & \p_2 r_1 &1 \\
    0 & \p_2 r_2 & 1 \\
    1 & - 1 & 0
\end{pmatrix}. 
\end{equation}
Although $\gamma$ is a buffering structure, 
the subnetwork $\gamma$ 
has one emergent cycle and one emergent conserved charge. 

For the mass-action kinetics, 
\begin{equation}
\begin{pmatrix}
  r_1 (x_1,x_2) \\
  r_2 (x_2,x_4)\\
  r_3 (x_3, x_4) \\
  r_4 (x_4) 
\end{pmatrix}
= 
\begin{pmatrix}
  k_1 x_1 x_2 \\
  k_2 x_2 x_4 \\
  k_3 x_3 x_4 \\
  k_4 x_4 
\end{pmatrix}  , 
\end{equation}
the steady-state concentrations are 
\begin{equation}
  \begin{pmatrix}
    \bar x_1 \\
    \bar x_2 \\
    \bar x_3 \\
    \bar x_4 \\
  \end{pmatrix}
  = 
  \begin{pmatrix} 
   0 \\
    m_1   \\
    m_2 \\
   0 
\end{pmatrix} 
\text{ or }
  \begin{pmatrix} 
   m'_1 \\
    0   \\
    m'_2 \\
   0 
\end{pmatrix} ,
\end{equation}
where $m_1,m_2,m'_1,m'_2$ are arbitrary parameters. 
With this kinetics, $\det A = 0$. 

Let us instead employ the following kinetics, 
\begin{equation}
  \begin{pmatrix}
    r_1 (x_1,x_2) \\
    r_2 (x_2,x_4)\\
    r_3 (x_3, x_4) \\
    r_4 (x_4) 
  \end{pmatrix}
  = 
  \begin{pmatrix}
    k_1 (x_1 + x_2) \\
    k_2 (x_2+ x_4) \\
    k_3 (x_3 +x_4 )\\
    k_4 x_4 
\end{pmatrix}  ,
\label{eq:kin-unphys}
\end{equation}
where all the concentrations vanish at the steady state, 
$\bar x_i = 0$. 
The matrix $A$ is now invertible, $\det A = k_1 k_2 k_3 k_4 \neq 0$. 
Although $A$ is regular, 
the sensitivity is trivial, $\p_A \bar x_i = 0$, 
since $\p_A \bar r_B = 0$ at the steady state. 
The kinetics (\ref{eq:kin-unphys}) is not physically sound, 
because the reaction $r_2 (x_2, x_4)$ can be nonzero 
even if the concentration of the reactant $x_4$ is zero (note that $x_2$ is catalytic). 
The same is true for $r_3$.

Let us consider the fluctuations around the steady state
of the emergent conserved charge. 
\begin{equation}
  \frac{d}{dt}
  (\delta x_1 - \delta x_2)
  = 
  -\delta r_4 (x_4)
  = - r_{4,4} \delta x_4 ,
\end{equation}
where we denote $r_{i,j}\coloneqq \p_j r_i$. 
The time derivative of $\delta x_4$ is 
\begin{equation}
\begin{split}
 \frac{d}{dt} \delta x_4 
 &= 
 - 2 (r_{2,2}\delta x_2 + r_{2,4} \delta x_4 )
 - (r_{3,3}\delta x_3 + r_{3,4} \delta x_4)
 - r_{4,4} \delta x_4 
 \\
&= 
- 2 r_{2,2}\, \delta x_2 
 - r_{3,3}\, \delta x_3 
 +(
   - 2 r_{2,4} 
   - r_{3,4}
   - r_{4,4}
 )\delta x_4 . 
  \end{split} 
\end{equation}
Hence, there is a zero mode when $r_{2,2}=0$ at the steady state, 
and then $A$ is not invertible.

\subsubsection{Emergent conserved charges and zero modes of fluctuations}

We denote the matrix $R$ whose components are given by 
$r_{i,j}= \p_j r_i$ and  
and separate it into block matrices, 
\begin{equation}
  R = 
  \begin{pmatrix}
    R_{11} & R_{12} \\
    R_{21} & R_{22}
  \end{pmatrix}, 
\end{equation}
according to the separation 
$\bm x = \begin{pmatrix}
  \bm x_1 \\
  \bm x_2
\end{pmatrix}$. 
The linear fluctuations around the steady state 
satisfy the following equations of motion, 
\begin{equation}
  \frac{d}{dt} 
\begin{pmatrix}
  \delta \bm x_1 \\
  \delta \bm x_2 
\end{pmatrix}  
= 
S R \delta \bm x 
= 
\begin{pmatrix}
S_{11 } R_{11}\delta \bm x_1 
+ 
(S_{11}  R_{12}+S_{12}R_{22} )\delta \bm x_2
\\
S_{21 } R_{11}\delta \bm x_1 
+ 
(S_{21}  R_{12}+S_{22}R_{22} )\delta \bm x_2 
\end{pmatrix} ,
\end{equation}
where we have also separated the stoichiometric matrix into submatrices,
and 
we used $R_{21}=\bm 0$, which follows from the output-completeness. 
We consider the fluctuation associated with 
the emergent conserved charge $\widetilde {\bm d}_1$, 
\begin{equation}
\frac{d}{dt} 
\widetilde{\bm d}_1^T \delta \bm x_1 
  = 
\widetilde{\bm d}_1^T S_{12} R_{22} \delta \bm x_2 . 
\label{eq:d-delta-x}
\end{equation}
Since it is an emergent charge, $\widetilde {\bm d}_1^T S_{12} \neq \bm 0$. 
The time derivative of the RHS of Eq.~(\ref{eq:d-delta-x}) reads 
\begin{equation}
  \frac{d}{dt} 
  \widetilde {\bm d}_1^T S_{12}R_{22}  \delta \bm x_2 
=
\widetilde {\bm d}_1^T S_{12}R_{22} S_{21 } R_{11}\delta \bm x_1 
 + (\cdots) \delta \bm x_2 .  
 \label{eq:dsrx}
\end{equation}
Therefore, 
an emergent conserved charge results in a zero mode when 
$\widetilde {\bm d}_1^T S_{12}R_{22} S_{21 } R_{11}$ vanishes. 
We are not aware of a physical example 
in which $\widetilde {\bm d}_1^T S_{12}R_{22} S_{21 } R_{11}$ does not vanish. 
In the example given by Eq.~(\ref{eq:ex-unphysical}), 
the first term of Eq.~(\ref{eq:dsrx}) is computed as 
\begin{equation}
  \begin{split}
  \widetilde {\bm d}_1^T
  S_{12}
  R_{22}
  S_{21}
  R_{11}
  \delta \bm x_1 
  &= 
  \begin{pmatrix}
    1 & -1
  \end{pmatrix}
  \begin{pmatrix}
    1 & 1 \\
    1 & 2
  \end{pmatrix}
  \begin{pmatrix}
    r_{3,3}&r_{3,4} \\
    0 & r_{4,4}
  \end{pmatrix}
  \begin{pmatrix}
    1 & 0 \\
    0 & -2
  \end{pmatrix}
  \begin{pmatrix}
    r_{1,1}&r_{1,2} \\
    0 & r_{2,2}
  \end{pmatrix}  
  \begin{pmatrix}
    \delta x_1 \\
    \delta x_2 
  \end{pmatrix}
\\
&= 
\begin{pmatrix}
  0 & 2 r_{4,4} r_{2,2}
\end{pmatrix}      
\begin{pmatrix}
  \delta x_1 \\
  \delta x_2 
\end{pmatrix}. 
\end{split}
\end{equation}
When this vanishes, the matrix $A$ acquires a zero mode and is not invertible.

\subsection{Absence of emergent conserved charges in monomolecular reaction networks }\label{sec:app-mono}

Here we consider monomolecular reaction networks and 
we show that, 
if there exists a nonzero emergent charge
in an output-complete subnetwork, 
the index $\lambda(\gamma)$ is necessarily negative. 
If the index is negative in an output-complete subnetwork, 
the matrix $A$ is not invertible, 
and the response of the system to the parameter-perturbation is not well-defined. 
We here show the following statement:  
\begin{theorem}
  Suppose that $\gamma$ is a connected and output-complete subnetwork of a monomolecular reaction network $\Gamma$. 
  If $\widetilde d(\gamma) > 0$, then  
  the influence index $\lambda(\gamma)$ is negative. 
\end{theorem}
\begin{proof}
To have an emergent conserved charge in $\gamma$, 
all the boundaries of $\gamma$ should be chemical species and 
not reactions, in a monomolecular reaction network. 
Then, all the reactions in $\gamma$ should end on 
the chemical species inside $\gamma$, which means $S_{21} = \bm 0$. 

Recall that an emergent cycle is 
$\bm c_1 \in \ker S_{11}$ which is not a cycle of the whole network, 
\begin{equation}
  S 
  \begin{pmatrix}
    \bm c_1 \\
    \bm 0
  \end{pmatrix}
  = 
  \begin{pmatrix}
    \bm 0 \\
    S_{21} \bm c_1 
  \end{pmatrix}
  \neq \bm 0 . 
\end{equation}
When $S_{21}=\bm 0$, 
there is no such $\bm c_1$ meaning that all the local cycles are also a global cycle. 
Namely, for a given 
$\bm c_1 \in \ker S_{11}$, 
$
\begin{pmatrix}
  \bm c_1 \\
  \bm 0
\end{pmatrix}
\in \ker S 
$
always holds. 
Thus, we have $\widetilde c(\gamma) = 0$. 

This also results in $d_l(\gamma) = d'(\gamma) - \bar d'(\gamma)=0$ as follows. 
If $d'(\gamma) - \bar d'(\gamma)$
is nonzero, 
there should exist $\bm d_1$ and $\bm d_2$ such that 
[recall the definitions of spaces, 
Eqs.~(\ref{eq:Dp-gamma}) and (\ref{eq:Dbp-gamma})] 
\begin{eqnarray}
  \bm d_1^T S_{11} + \bm d^T_2 S_{21}  &=& \bm 0,\label{eq:ds-1} \\
  \bm d_1^T S_{12} + \bm d^T_2 S_{22}  &=& \bm 0, \\
  \bm d^T_1 S_{11} \neq \bm 0. \label{eq:ds-3}
\end{eqnarray}
However, when $S_{21} = \bm 0$, 
Eqs.~(\ref{eq:ds-1}) and 
(\ref{eq:ds-3}) are contradictory. 
Thus $d'(\gamma) - \bar d'(\gamma)=0$ 
and we have $d_l (\gamma)=0$. 

Therefore, 
we have shown that $\widetilde c (\gamma)= 0$ and 
$d_l (\gamma) = 0$, 
and the index is written as 
\begin{equation}
\lambda (\gamma) 
  = \widetilde c (\gamma) + d_l (\gamma) - \widetilde d (\gamma) 
  = - \widetilde d (\gamma),
\end{equation}
which is negative due to the assumption $\widetilde d (\gamma)> 0$. 

\end{proof}

\section{Metabolic pathways of {\it E.~coli}}\label{sec:app-ecoli}

We here provide the details of the metabolic pathways 
discussed in Sec.~\ref{sec:ex-ecoli}. 

\subsection{List of reactions}\label{sec:ecoli-list}

1: Glucose  +  PEP  $\rightarrow$  G6P  +  PYR. 
 
 2: G6P   $\rightarrow$  F6P. 
 
 3: F6P   $\rightarrow$  G6P.
  
  4: F6P  $\rightarrow$  F16P. 
   
   5: F16P  $\rightarrow$  G3P  +  DHAP. 
   
   6: DHAP   $\rightarrow$  G3P.
   
 7: G3P   $\rightarrow$  3PG. 
  
  8: 3PG   $\rightarrow$  PEP. 
  
  9: PEP   $\rightarrow$  3PG.  
  
  10: PEP   $\rightarrow$  PYR. 
  
  11: PYR   $\rightarrow$  PEP. 
  
  12: PYR    $\rightarrow$  AcCoA  +   CO2. 
  
  13: G6P  $\rightarrow$  6PG. 
  
  14: 6PG  $\rightarrow$   Ru5P  +  CO2. 
  
  15: Ru5P  $\rightarrow$  X5P.
  
   16: Ru5P  $\rightarrow$   R5P. 
   
   17: X5P  +  R5P  $\rightarrow$   G3P  +  S7P. 
   
   18: G3P  +  S7P  $\rightarrow$   X5P  +  R5P. 
   
   19: G3P  +  S7P  $\rightarrow$   F6P  +  E4P.
   
    20: F6P  +  E4P  $\rightarrow$  G3P  +  S7P. 
    
21:  X5P  +  E4P  $\rightarrow$   F6P  +  G3P. 
  
  22: F6P   +  G3P  $\rightarrow$   X5P  +  E4P.  
  
  23: AcCoA  +  OAA  $\rightarrow$  CIT. 
  
  24: CIT   $\rightarrow$  ICT. 
  
  25: ICT  $\rightarrow$  2${\rm \mathchar`-}$KG  +  CO2. 
  
  26: 2-KG  $\rightarrow$   SUC  +  CO2.
  
   27: SUC   $\rightarrow$  FUM. 
   
  28:  FUM  $\rightarrow$  MAL. 
   
   29: MAL   $\rightarrow$  OAA.
   
 30: OAA   $\rightarrow$  MAL.

 31: PEP  +  CO2  $\rightarrow$  OAA.

 32: OAA  $\rightarrow$  PEP  +   CO2. 

 33: MAL  $\rightarrow$   PYR  +  CO2.

34: ICT   $\rightarrow$  SUC  +  Glyoxylate. 

 35: Glyoxylate  +  AcCoA  $\rightarrow$  MAL. 

 36: 6PG  $\rightarrow$   G3P  +  PYR. 

 37: AcCoA  $\rightarrow$   Acetate. 

38:  PYR  $\rightarrow$  Lactate. 

 39: AcCoA  $\rightarrow$  Ethanol. 

 40: R5P  $\rightarrow$ (output).

 41: OAA  $\rightarrow$ (output).

 42: CO2  $\rightarrow$ (output).

43:  (input) $\rightarrow$  Glucose. 

 44:  Acetate $\rightarrow$ (output).
 
  45: Lactate $\rightarrow$ (output).

46:  Ethanol $\rightarrow$ (output).
\\

\subsection{List of buffering structures}\label{sec:ecoli-buffering-list}
 
$\gamma_1=(\{ \rm 
Glucose
\},\{
1
\})$,

$\gamma_2=(\{ \rm 
Glucose,
PEP,
G6P,
F6P,
F16P,
DHAP,
G3P,
3PG,
PYR,
6PG,
Ru5P,
X5P,
R5P, 
S7P,
E4P,
AcCoA,
OAA, \\
CIT,
ICT, 
2{\rm \mathchar`-KG},
SUC,
FUM,
MAL,
CO2,
Glyoxylate,
Acetate,
Lactate,
Ethanol
\},  \{
1,
2 ,
3 ,
4 ,
5 ,
6 ,
7 ,
8 ,
9 ,
10 ,
11 ,
12 ,
13 ,
14 , \\
15 , 
16 ,
17 ,
18 ,
19 ,
20 , 
21 ,
22 ,
23 ,
24 ,
25 ,
26 ,
27 , 
28 ,
29 ,
30 ,
31 ,
32 ,
33 ,
34 ,
35 ,
36 ,
37 ,
38 ,
39 ,
40 ,
41 ,
42 ,
44 ,
45 ,
46 
\})$,

$\gamma_3=(\{ \rm 
F16P
\},\{
5 
\})$,

$\gamma_4=(\{ \rm 
DHAP
\},\{
6 
\})$,

$\gamma_5=(\{ \rm 
G3P,
X5P,
S7P,
E4P
\},\{
7 ,
17 ,
18 ,
19 ,
20 ,
21 ,
22 
\})$,

$\gamma_6=(\{ \rm 
3PG
\},\{
8 
\})$,

$\gamma_7=(\{ \rm 
Glucose,
PEP,
3PG,
PYR,
AcCoA,
OAA,
CIT,
ICT,
2{\rm \mathchar`-KG},
SUC,
FUM, 
MAL,
CO2,
Glyoxylate,
Acetate, \\
Lactate,  
Ethanol
\},
\{
1,
8 ,
9 ,
10 ,
11 ,
12 ,
23 ,
24 ,
25 ,
26 ,
27 ,
28 ,
29 ,
30 ,
31 ,
32 ,
33 ,
34 ,
35 ,
37 ,
38 ,
39 ,
41 ,
42 ,
44 ,
45 ,
46 
\})$,

$\gamma_8=(\{ \rm 
X5P,
S7P,
E4P
\},\{
17 ,
18 ,
19 ,
20 ,
21 
\})$ ,

$\gamma_9=(\{ \rm 
CIT
\},\{
24 
\})$,

$\gamma_{10}=(\{ \rm 
2{\rm \mathchar`-KG}
\},\{
26 
\})$,

$\gamma_{11}=(\{ \rm 
SUC
\},\{
27 
\})$ ,

$\gamma_{12}=(\{ \rm 
FUM
\},\{
28 
\})$ 

$\gamma_{13}=(\{ \rm 
Glyoxylate
\},\{
35 
\})$,

$\gamma_{14}=(\{ \rm 
X5P,
R5P,
S7P,
E4P
\},\{
17 ,
18 ,
19 ,
20 ,
21 ,
40 
\})$ ,

$\gamma_{15}=(\{ \rm 
Acetate
\},\{
44 
\})$,

$\gamma_{16}=(\{ \rm 
Lactate
\},\{
45 
\})$,

$\gamma_{17}=(\{ \rm 
Ethanol
\},\{
46 
\})$.

\subsection{Parameter values used in Figure~\ref{fig:dyn}}\label{sec:parameters}
In Fig.~\ref{fig:dyn}, for an illustration purpose, we employ the mass-action kinetics, where the rate of the $i$-th reaction is given by the product of its substrate concentrations, $r_i = k_i \prod_A (x_A(t))^{y_{iA}}$ (see Eq.~\eqref{eq:yybar} for the definition of $y_{iA}$).

In the simulation, the initial concentrations and the reaction rate constants are chosen randomly:
$x_{\text {6PG}} = 0.8$,\,$x_{\text {AcCoA}} = 0.8$,\,$x_{\text {Acetate}} = 0.4$,\,$x_{\text {CIT}} = 0.3$,\,$x_{\text {CO2}} = 0.6$,\,$x_{\text {DHAP}} = 0.1$,\,$x_{\text {E4P}} = 0.8$,\,$x_{\text {Ethanol}} = 0.2$,\,$x_{\text {F16P}} = 0.2$,\,$x_{\text {F6P}} = 0.5$,\,$x_{\text {FUM}} = 0.3$,\,$x_{\text {G3P}} = 0.3$,\,$x_{\text {G6P}} = 0.2$,\,$x_{\text {Glucose}} = 0.7$,\,$x_{\text {Glyoxylate}} = 0.6$,\,$x_{\text {ICT}} = 0.4$,\,$x_{\text {KG2}} = 0.5$,\,$x_{\text {Lactate}} = 1.$,\,$x_{\text {MAL}} = 0.4$,\,$x_{\text {OAA}} = 1.$,\,$x_{\text {PEP}} = 0.6$,\,$x_{\text {PG3}} = 1.$,\,$x_{\text {PYR}} = 0.1$,\,$x_{\text {R5P}} = 0.2$,\,$x_{\text {Ru5P}} = 0.4$,\,$x_{\text {S7P}} = 0.7$,\,$x_{\text {SUC}} = 0.1$,\,$x_{\text {X5P}} = 0.6$
and 
$k_{1} = 1$,\,$k_{2} = 4.7$,\,$k_{3} = 7.8$,\,$k_{4} = 5.7$,\,$k_{5} = 3.8$,\,$k_{6} = 9.7$,\,$k_{7} = 5.0$,\,$k_{8} = 6.2$,\,$k_{9} = 3.5$,\,$k_{10} = 9.8$,\,$k_{11} = 2.5$,\,$k_{12} = 6.1$,\,$k_{13} = 4.0$,\,$k_{14} = 3.8$,\,$k_{15} = 7.8$,\,$k_{16} = 2.6$,\,$k_{17} = 3.8$,\,$k_{18} = 5.5$,\,$k_{19} = 5.7$,\,$k_{20} = 4.7$,\,$k_{21} = 8.0$,\,$k_{22} = 7.3$,\,$k_{23} = 9.2$,\,$k_{24} = 1.1$,\,$k_{25} = 9.6$,\,$k_{26} = 7.4$,\,$k_{27} = 7.4$,\,$k_{28} = 8.3$,\,$k_{29} = 6.2$,\,$k_{30} = 6.4$,\,$k_{31} = 6.2$,\,$k_{32} = 7.9$,\,$k_{33} = 9.1$,\,$k_{34} = 6.7$,\,$k_{35} = 1.6$,\,$k_{36} = 9.6$,\,$k_{37} = 4.7$,\,$k_{38} = 5.1$,\,$k_{39} = 7.3$,\,$k_{40} = 3.8$,\,$k_{41} = 8.4$,\,$k_{42} = 9.7$,\,$k_{43} = 4.8$,\,$k_{44} = 2.0$,\,$k_{45} = 8.0$,\,$k_{46} = 3.7$.


\bibliographystyle{utphys.bst}

\bibliography{refs}

\end{document}